\documentclass{article}

\usepackage{iclr2027_conference,times}
\iclrfinalcopy %for arxiv
\usepackage[T1]{fontenc}
\usepackage[utf8]{inputenc}
\usepackage{microtype}
\usepackage{amsmath,amssymb,amsfonts,amsthm,mathtools}
\usepackage{bm}
\usepackage{mathrsfs}
\usepackage{graphicx}
\usepackage{booktabs}
\usepackage{tabularx}
\usepackage{array}
\usepackage{multirow}
\usepackage{makecell} 
\usepackage{adjustbox}
\usepackage{threeparttable}
\usepackage{enumitem}
\usepackage{xcolor}
\usepackage{url}
\usepackage{nicefrac}
\usepackage{hyperref}
\hypersetup{hidelinks,pdfauthor={},pdftitle={Task-Aware Answer Preservation under Audio Compression for Large Audio Language Models}}
\usepackage{tablefootnote}
\usepackage[capitalize,noabbrev]{cleveref}
\usepackage{tikz}
\usepackage{siunitx}
\newcounter{protocol}
\renewcommand{\theprotocol}{\arabic{protocol}}
\crefname{protocol}{Algorithm}{Algorithms}
\usetikzlibrary{arrows.meta,positioning,calc,fit,backgrounds}

\numberwithin{equation}{section}
\allowdisplaybreaks
\setlist[itemize]{leftmargin=1.25em,topsep=2pt,itemsep=1.5pt,parsep=0pt}
\setlist[enumerate]{leftmargin=1.5em,topsep=2pt,itemsep=1.5pt,parsep=0pt}

\theoremstyle{plain}
\newtheorem{theorem}{Theorem}[section]
\newtheorem{proposition}[theorem]{Proposition}
\newtheorem{corollary}[theorem]{Corollary}

\newtheorem{assumption}[theorem]{Assumption}
\newtheorem{definition}[theorem]{Definition}
\theoremstyle{remark}
\newtheorem{remark}[theorem]{Remark}

\newcommand{\davg}[1]{\widehat{\Delta}_{\mathrm{avg}}(#1)}
\newcommand{\dQF}[1]{\widehat{\Delta}^{\mathrm{fam}}_{\mathcal P}(#1)}
\newcommand{\dQP}[1]{\widehat{\Delta}^{\mathrm{fam}}_{\mathcal P}(#1)}
\newcommand{\Rfb}[2]{\widehat{R}^{\star}_{\mathcal F,b}(#1,#2)}
\newcommand{\Rftok}[2]{\widehat{R}^{\star}_{\mathcal F,\mathrm{tok}}(#1,#2)}
\newcommand{\Rfworst}[2]{\widehat{R}^{\star,\mathrm{worst}}_{\mathcal F}(#1,#2)}
\newcommand{\Gcondop}[2]{\widehat{G}_{\mathrm{cond}}^{\mathrm{op}}(#1,#2)}
\newcommand{\Gcondinfo}[2]{\widehat{G}_{\mathrm{cond}}^{\mathrm{info}}(#1,#2)}
\newcommand{\tauhat}[1]{\widehat{\tau}(#1)}
\newcommand{\Qfam}{\mathcal Q}

\title{Task-Aware Answer Preservation under\\Audio Compression for Large\\Audio Language Models}
\author{
  Amir Ivry\\
  Technion--Israel Institute of Technology\\
  \texttt{aivry@ieee.org}
}

\begin{document}
\maketitle
%for arxiv
\fancyhead{}
\renewcommand{\headrulewidth}{0pt}

\begin{abstract}
Large audio language models (LALMs) increasingly reason over long audio clips, motivating audio compression to reduce memory use and inference latency. However, audio compression can leave the model's overall answer accuracy acceptable while severely degrading accuracy for particular query families. We introduce a framework for auditing a given audio compression method by measuring the increase in an LALM's answer error relative to uncompressed audio. We formalize an acceptance criterion that limits this increase for every required query family and derive a practical protocol for testing compression settings with statistical confidence. Experiments spanning five multiple-choice audio question-answering benchmarks and four LALMs expose compression-induced damage hidden by overall accuracy.
\end{abstract}

\section{Introduction}
\label{sec:introduction}
Large audio language models (LALMs) can reason and answer questions about speech, music, and audio events~\citep{chu2024qwen2audio,xu2025qwen25omni}. In deployment, long clips, interactive use, and multimodal contexts motivate compressing audio before LALM inference to reduce memory use and response latency~\citep{arora2025landscape}. However, compression can affect answer accuracy differently across query families: questions about lexical content, event timing, or prosody and speaker state depend on different information in the recording. A small overall accuracy drop can therefore conceal severe degradation for a particular query family. An audio compression method should therefore be judged against the answer-accuracy requirements of the tasks the LALM must perform in deployment.

We introduce a task-aware answer-preservation framework for auditing a supplied audio compressor, such as EnCodec~\citep{defossez2023encodec}, with a fixed LALM. The contribution is a criterion and practical protocol for accepting or rejecting the compressor's tested settings, e.g., a target bitrate of 1.5~kbps. The practitioner specifies the query families to evaluate, such as questions about speaker emotion, and the allowed increase in answer error relative to uncompressed audio. The criterion requires that this tolerance be met for every evaluated family. The protocol selects compression settings on development data and tests the selected settings on independent held-out data. It distinguishes three outcomes: evidence that every evaluated family meets the tolerance, evidence that a family exceeds the tolerance, or insufficient evidence for either conclusion. When a setting cannot be approved, the audit identifies the families with excessive degradation or insufficient evidence of answer preservation. The audit can thus change which compression settings are approved for deployment and identify specific failures that subsequent compression methods should address.

The paper has three main contributions. First, it defines answer preservation mathematically by limiting how much compression can reduce the best achievable answer accuracy for each query family. Second, it provides a practical protocol for testing, with statistical confidence, whether a supplied compressor keeps a fixed LALM's increase in answer error within the specified tolerance for every evaluated query family. Third, experiments spanning five English-prompted, multiple-choice audio question-answering benchmarks and four LALMs demonstrate compression-induced damage to particular query families that dataset averages conceal. We also show that grouping queries into different families changes which compression settings pass empirical accuracy checks, and identify when using the query to guide compression reduces the audio budget at the same accuracy tolerance.

\paragraph{Related work.}
Task-oriented source coding studies compression through its effect on task performance, including limitations of the model performing the task~\citep{enttsel2026modelaware}. Prompt-compression theory studies the tradeoff between prompt length and a fixed language model's answer quality, including the benefit of giving the compressor access to the question~\citep{nagle2024fundamental}. Learn then Test provides statistical procedures for choosing settings that satisfy specified risk limits~\citep{angelopoulos2025learn}, and CALM uses this approach to reduce language-model computation while limiting the increase in task loss relative to the full model~\citep{schuster2022confident}.
Our focus is auditing audio compression against the answer-accuracy requirements of individual query families. Additional related work is in App.~\ref{app:related}.

\section{Problem Formulation and Setup}
\label{sec:problem-main}
Let \(\mathcal X\) be the space of finite-duration, sampled audio clips and let
\(\mathcal Q\) be the family of queries that may be asked about such
clips. Examples include speech-content, event-counting, and
speaker or prosody queries.
Conditioned on the raw audio input \(X\in\mathcal X\), for each fixed query \(q\in\mathcal Q\), let \(Y_q\) be the random correct answer taking values in \(\mathcal Y_q\), let \(\mathcal A_q\) be the action space of allowed predictions, and let \(\ell_{q}:=\ell:\mathcal A_q\times\mathcal Y_q\to[0,\infty)\) be the answer loss.

We define the set \({C=\{C_b(X, q):b\in\mathcal B\}}\) as a family of budgeted, $q$-query-conditioned compressors on $X$, where \(\mathcal B\) are
retained-interface budgets. 
A compressed interface is denoted
\begin{equation}
\label{eq:compressed-interface}
Z_{q,b}=C_b(X,q)\in\mathcal Z_b ,
\end{equation}
where $\mathcal Z_b$ holds the space of admissible compressed interfaces at budget $b$. Here, admissible means measurable, finite, valid for the downstream LALM, and within budget $b$.
For example, if \(q\) asks about voice-activity events and $C$ retains audio fractions, then \(C_b(X,q)\) may retain the most energetic 1-second chunks of
\(X\), with \(b\) controlling how many chunks are retained. 

For each \(q\) and \(b\), we compare the best answer rule using \(X\), denoted $\mathcal R_X^\star(q)$,
with the best answer rule using the compressed interface \(Z_{q,b}\), denoted $\mathcal R_{Z_{q,b}}^\star(q)$:
\begin{equation}
\label{eq:bayes-answer-risks}
\mathcal R_X^\star(q)
=
\inf_{\psi:\mathcal X\to\mathcal A_q}
\mathbb E\!\left[\ell(\psi(X),Y_q)\right],
\qquad
\mathcal R_{Z_{q,b}}^\star(q)
=
\inf_{\phi:\mathcal Z_b\to\mathcal A_q}
\mathbb E\!\left[\ell(\phi(Z_{q,b}),Y_q)\right].
\end{equation}
The expectation is over the distribution of \(X\) and the
query-specific, conditional answer law \(Y_q\mid X\).
For a fixed query \(q\), the per-query excess answer-risk at budget \(b\)
is
\begin{equation}
\label{eq:per-query-excess-risk}
\Delta_q(b;C)
=
\mathcal R_{Z_{q,b}}^\star(q)-\mathcal R_X^\star(q).
\end{equation}
Over the deployment query family \(\mathcal Q\), the relevant guarantee is
the largest such degradation:
\begin{equation}
\label{eq:family-excess-risk}
\Delta_{\mathcal Q}(b;C)
=
\sup_{q\in\mathcal Q}
\Delta_q(b;C).
\end{equation}
The budgeted interface \(C_b(X,q)\) is \(\varepsilon\)-answer-preserving on
\(\mathcal Q\) if, for $\varepsilon\geq0$,
\begin{equation}
\label{eq:epsilon-answer-preserving}
\Delta_{\mathcal Q}(b;C)\le\varepsilon.
\end{equation}
The supremum is intentional: a compressor should not be approved merely
because its average error is small if it substantially damages a
deployment-critical query type.

The answer-preservation frontier of \(C\) is the smallest retained
budget that satisfies
\begin{equation}
\label{eq:answer-preservation-frontier}
b_{\mathcal Q}^\star(\varepsilon;C)
=
\inf
\left\{
b\in\mathcal B:
\Delta_{\mathcal Q}(b;C)\le\varepsilon
\right\},
\end{equation}
with \(b_{\mathcal Q}^\star(\varepsilon;C)=\infty\) if no budget $b$ yields an $\varepsilon$-answer-preserving interface.

The frontier in \eqref{eq:answer-preservation-frontier} concerns Bayes-optimal answer rules. The operational audit instead fixes the deployed answerer. For $Z=Z_{q,b}$, write $\mathcal R_f(q;U)$ for its expected loss from interface $U$. Its excess risk satisfies
\begin{equation}
\begin{aligned}
\mathcal R_f(q;Z)-\mathcal R_f(q;X)
={}&\Delta_q(b;C)
+\bigl[\mathcal R_f(q;Z)-\mathcal R_Z^\star(q)\bigr]\\
&-\bigl[\mathcal R_f(q;X)-\mathcal R_X^\star(q)\bigr].
\end{aligned}
\label{eq:fixed-answerer-gap}
\end{equation}
Thus observed degradation combines information loss with changes in the answerer's ability to use the interface. It can be negative and does not identify the Bayes frontier. Section~\ref{sec:estimating-main} audits the specified answerer and compressor together. Further setup details are in App.~\ref{app:extended-setup}.

\section{Theoretical Foundations: Task-aware Answer-Preservation}
\label{sec:theory-main}
\paragraph{Notational convention.}
For a fixed budget \(b\) and compressor \(C\), we write \(Z=Z_{q,b}\)
and use \(\Delta_{\mathcal Q}(Z;X)\) as shorthand for
\(\Delta_{\mathcal Q}(b;C)\).

The following results specialize classical comparison-of-experiments and task-oriented rate-distortion ideas~\citep{blackwell1953equivalent,lecam1964sufficiency,witsenhausen1980indirect}. They provide the audit's foundations.

\begin{assumption}[Regularity]
\label{ass:bayes}
For every \(q\in\mathcal Q\), Bayes-optimal answer rules based on \(X\)
and on any admissible compressed interface \(Z\) exist. The losses are
uniformly bounded and all displayed
expectations are finite.
\end{assumption}

\paragraph{Answer-risk gaps are sufficiency deficits.}
For \(\pi \in \mathcal P(\mathcal Y_q)\), where \(\mathcal P(\mathcal Y_q)\)
denotes probability laws on the answer space, define the Bayes envelope and induced posteriors
\begin{equation}
\label{eq:bayes-envelope}
L_q(\pi)=\inf_{a\in\mathcal A_q}\int \ell(a,y)\,\pi(dy),\qquad
\Pi_q^X=P_{Y_q\mid X,q},\quad \Pi_q^Z=P_{Y_q\mid Z,q},
\end{equation}
then \(\mathcal R_X^\star(q)=\mathbb E_{X}\!\left[L_q(\Pi_q^X)\right]\) and \(\mathcal R_Z^\star(q)=\mathbb E_{Z}\!\left[L_q(\Pi_q^Z)\right]\), where the expectations use the marginal laws induced by the deployment distribution and compressor.

We define the one-sided answer deficiency of \(Z\) relative to raw audio $X$ as
\begin{equation}
\label{eq:ans-deficiency}
\delta_{\Qfam}(Z\|X)=
\sup_{q\in\Qfam}\left\{
\mathbb E_{Z}\!\left[L_q(\Pi_q^Z)\right]
-
\mathbb E_{X}\!\left[L_q(\Pi_q^X)\right]
\right\}.
\end{equation}

\begin{theorem}[Restricted sufficiency-risk equivalence]
\label{thm:equiv}
For every \(Z\),
\begin{equation}
\label{eq:suff-risk-equivalence}
\delta_{\Qfam}(Z\|X)=\Delta_{\Qfam}(Z;X).
\end{equation}
Thus \(Z\) is exact task-aware answer sufficient for \(\Qfam\) iff \(\Delta_{\Qfam}(Z;X)=0\), and is \(\varepsilon\)-approximately sufficient iff \(\Delta_{\Qfam}(Z;X)\le\varepsilon\).
\end{theorem}
The proof is in App. \ref{app:theory-proofs} using a Jensen-gap identity for the Bayes envelope. This identity interprets the excess \emph{Bayes} answer risk as the restricted sufficiency deficit for the declared queries and loss. The fixed-answerer measurements in Section~\ref{sec:estimating-main} retain the same raw-versus-compressed comparison but do not identify this deficit.

\paragraph{Partition refinement and operational budgets.}
The query-level object in \eqref{eq:suff-risk-equivalence} is a worst-query guarantee. In experiments, however, we
observe finite query-family partitions rather than repeated samples for every individual
query. 
% We write
% \begin{equation}
% d_q(Z)=\mathcal R^\star_Z(q)-\mathcal R^\star_X(q)
% \end{equation}
% for the Bayes excess answer-risk of interface \(Z\) on query \(q\). 
Let \(\mu\) denote the deployment distribution over queries, and for a partition
\(\mathcal P\) of \(\mathcal Q\), define the partition-level family gap:
\begin{equation}
\Delta^{\rm fam}_{\mathcal P}(Z;X)
=
\max_{F\in\mathcal P:\mu(F)>0}
\bar d_F(Z),
\qquad
\bar d_F(Z)
=
\mathbb E_{\mu}\!\left[R_{Z}^\star(Q)-\mathcal R_X^\star(Q)\mid Q\in F\right],
\end{equation}
and the deployment-average gap
\begin{equation}
\Delta^{\rm avg}_{\mu}(Z;X)=\mathbb E_{\mu}\!\left[R_{Z}^\star(Q)-\mathcal R_X^\star(Q)\right].
\end{equation}
Thus, \(\Delta^{\rm fam}_{\mathcal P}\) is a population-level, lower-resolution version of the worst-query gap
in \eqref{eq:family-excess-risk}.

\begin{theorem}[Monotonicity under partition refinement]
\label{thm:partition-refinement}
Let \(\mathcal P'\preceq \mathcal P\) denote \(\mathcal P'\) refines \(\mathcal P\): for every
fine cell \(G\in\mathcal P'\), there exists a coarse cell \(F\in\mathcal P\) such that
\(G\subseteq F\). Then, for all \(Z\),
\begin{equation}
\Delta^{\rm avg}_{\mu}(Z;X)
\le
\Delta^{\rm fam}_{\mathcal P}(Z;X)
\le
\Delta^{\rm fam}_{\mathcal P'}(Z;X)
\le
\Delta_{\mathcal Q}(Z;X).
\end{equation}
\end{theorem}
As in \eqref{eq:compressed-interface}, let \(C\) be a budgeted compressor and \(Z_b\) a budget-\(b\) interface, and define
\begin{equation}
\mathcal B_{\mathcal P}(\varepsilon;C)
=
\left\{
b\in\mathcal B:
\Delta^{\rm fam}_{\mathcal P}(Z_b;X)\le \varepsilon
\right\},
\quad
b^\star_{\mathcal P}(\varepsilon;C)
=
\inf\left\{
b\in\mathcal B:
\Delta^{\rm fam}_{\mathcal P}(Z_b;X)\le \varepsilon
\right\},
\end{equation}
with \(b^\star_{\mathcal P}(\varepsilon;C)=\infty\) if the feasible set is empty.
% For a budgeted compressor
% \(C\) as in \eqref{eq:compressed-interface}, define the partition-indexed feasible set 
% \begin{equation}
% \mathcal B_{\mathcal P}(\epsilon;C)
% =
% \left\{
% b\in\mathcal B:
% \Delta^{\rm fam}_{\mathcal P}(Z_b;X)\le \epsilon
% \right\},
% \qquad
% b^\star_{\mathcal P}(\epsilon;C)
% =
% \inf \mathcal B_{\mathcal P}(\epsilon;C).
% \end{equation}
If \(\mathcal P'\preceq\mathcal P\), then
\begin{equation}
\mathcal B_{\mathcal P'}(\varepsilon;C)
\subseteq
\mathcal B_{\mathcal P}(\varepsilon;C),
\qquad
b^\star_{\mathcal P}(\varepsilon;C)
\le
b^\star_{\mathcal P'}(\varepsilon;C).
\end{equation}

% If \(\mathcal P'\preceq\mathcal P\), then
% \begin{equation}
% \mathcal B_{\mathcal P'}(\epsilon;C)
% \subseteq
% \mathcal B_{\mathcal P}(\epsilon;C),
% \qquad
% b^\star_{\mathcal P}(\epsilon;C)
% \le
% b^\star_{\mathcal P'}(\epsilon;C).
% \end{equation}

A budget that passes a refined partition also passes the coarse partition, but a budget
that passes the coarse partition may fail once the same examples are split into more
specific query families. The proof, including the budget consequence, is in App.~\ref{app:proof-partition-refinement}. 

This theorem is the formal reason for using familywise sign-off rather than a dataset
average. Refinement does not create new evidence, it removes averaging across heterogeneous
subfamilies. A coarse speech family can pass at tolerance \(\varepsilon\) because
low-damage lexical queries are averaged together with high-damage prosody queries. After
refinement, the prosody cell can become the bottleneck, forcing a larger
budget or a family-specific fallback.

\paragraph{Answer-risk frontiers and query conditioning.}
Keep the deployment query distribution \(\mu\) from the partition definition, and let \(Q\) denote the runtime query with support \(\operatorname{supp}(\mu)=\Qfam\). The prior \(\mu\) weights average rate or length only, whereas the answer-preservation constraint remains worst-case over \(\Qfam\). We write
\(I(\cdot;\cdot)\) and \(H(\cdot)\) for mutual information and entropy~\citep{coverthomas1991}, respectively, and
\(\tau\) for the length or cost of the compressed interface.
The frontiers below translate an answer-preservation tolerance into the smallest
interface size that can meet it: \(R^\star_{\rm Bayes}\) measures information
rate, not Bayes answer-risk, while \(L^\star\) measures expected length or cost.
% The frontiers below translate an answer-preservation tolerance into the smallest interface
% size that can meet it: \(R^\star\) measures information rate, while \(L^\star\)
% measures expected length or cost. 
% All preservation constraints are the budgeted ones from \eqref{eq:family-excess-risk}. 
For query-agnostic
compressors, i.e., when \(C_b(X,q)\) does not depend on \(q\):
\begin{equation}
\label{eq:bayes-frontiers}
\begin{aligned}
R_{\rm Bayes}^\star(\varepsilon,\Qfam)
&=
\inf_{\substack{C,\;b\in\mathcal B\\
C_b(X,q)=C_b(X,q')\ \forall q,q'\in\Qfam}}
\left\{
I\!\left(X;C_b(X,Q)\right):
\Delta_{\Qfam}(b;C)\le\varepsilon
\right\},\\
L_{\rm Bayes}^\star(\varepsilon,\Qfam)
&=
\inf_{\substack{C,\;b\in\mathcal B\\
C_b(X,q)=C_b(X,q')\ \forall q,q'\in\Qfam}}
\left\{
\mathbb E_{X,\,Q\sim\mu}[\tau(C_b(X,Q))]:
\Delta_{\Qfam}(b;C)\le\varepsilon
\right\}.
\end{aligned}
\end{equation}
For query-conditioned compressors, when \(C_b(X,q)\) may depend on the realized query, then:
\begin{equation}
\label{eq:cond-frontiers}
\begin{aligned}
R_{\rm Bayes,cond}^\star(\varepsilon,\Qfam)
&=
\inf_{\substack{C,\;b\in\mathcal B}}
\left\{
I\!\left(X;C_b(X,Q)\mid Q\right):
\Delta_{\Qfam}(b;C)\le\varepsilon
\right\},\\
L_{\rm Bayes,cond}^\star(\varepsilon,\Qfam)
&=
\inf_{\substack{C,\;b\in\mathcal B}}
\left\{
\mathbb E_{X,\,Q\sim\mu}[\tau(C_b(X,Q))]:
\Delta_{\Qfam}(b;C)\le\varepsilon
\right\}.
\end{aligned}
\end{equation}
% The prior on \(Q\) affects the average rate or length of a conditioned
% compressor, but the preservation constraint remains worst-case over the
% signed-off query family.

\begin{theorem}[Conditioned compression advantage]
\label{thm:condadv}
For any design prior \(Q\sim\mu\) with support \(\Qfam\), independent of \(X\),
\begin{equation}
R_{\rm Bayes,cond}^\star(\varepsilon,\Qfam)
\le
R_{\rm Bayes}^\star(\varepsilon,\Qfam),
\qquad
L_{\rm Bayes,cond}^\star(\varepsilon,\Qfam)
\le
L_{\rm Bayes}^\star(\varepsilon,\Qfam).
\end{equation}
The rate inequality can be strict in an idealized representation class permitting the answer-factor encodings below. Let \(\Qfam=\{q_1,q_2\}\) and
\(P(Q=q_1)=\lambda\in(0,1)\), and let the audio contain two independent
answer factors, \(X=(V_1,V_2,W)\), with \(0<H(V_1),H(V_2)<\infty\), where \(W\) is
irrelevant side information. Query \(q_i\) asks for factor \(V_i\), so the
corresponding correct-answer variable is \(Y_{q_i}=V_i\), and 0-1 loss is
evaluated on the label alphabet of \(V_i\). Then
\begin{equation}
R_{\rm Bayes}^\star(0,\{q_1,q_2\})=H(V_1,V_2),
\quad
R_{\rm Bayes,cond}^\star(0,\{q_1,q_2\})
=
\lambda H(V_1)+(1-\lambda)H(V_2).
\end{equation}
Thus the strict gap is
\((1-\lambda)H(V_1)+\lambda H(V_2)>0\): the agnostic interface must retain both
answer factors to answer either possible query, whereas the conditioned
interface observes \(Q\) first and retains only the realized answer factor.
% The rate inequality can be strict. Let \(\Qfam=\{q_1,q_2\}\),
% \(P(Q=q_1)=\lambda\in(0,1)\), \(X=(S_1,S_2)\) with independent
% finite-alphabet factors satisfying \(H(S_1),H(S_2)>0\), \(Y_{q_i}=S_i\),
% and zero-one loss. Then
% \begin{equation}
% R_{\rm Bayes}^\star(0,\{q_1,q_2\})=H(S_1,S_2),
% \quad
% R_{\rm Bayes,cond}^\star(0,\{q_1,q_2\})
% =
% \lambda H(S_1)+(1-\lambda)H(S_2).
% \end{equation}
% Thus the strict gap is
% \((1-\lambda)H(S_1)+\lambda H(S_2)>0\): the agnostic interface must retain
% both factors to answer either possible query, whereas the conditioned
% interface observes \(Q\) first and retains only the realized answer factor.
\end{theorem}
Conditioning is thus a frontier-level
possibility result, not an operational guarantee: it helps when different
queries require different answer-relevant audio factors. At the ideal frontier
the gain cannot be negative, but in a fixed learned system the measured gain
can vanish or reverse when factors are shared, conditioning barely changes the
retained interface, or the downstream LALM cannot use the retained information. The proof is in App.~\ref{app:e4}.

\section{Practical Methodology and Performance Metrics}
\label{sec:estimating-main}
% As denoted in Section~\ref{sec:problem-main}, 
Consider \(N\) dataset examples
\(\{(x_i,q_i,y_i)\}_{i=1}^N\), where example \(i\) consists of an audio clip,
query, and correct answer. Let \(\mathcal P\) be a query-family partition of
\(\mathcal Q\), and let \(F_i\in\mathcal P\) be the family of query \(q_i\).
A candidate compression method \(m\) is a rule that, for each budget \(b\),
produces compressed interfaces, where \(z^{(m)}_{i,b}=\mathcal{C}_{b}^{(m)}(x_i,q_i)\), following notation from~\eqref{eq:compressed-interface}. We fix one method \(m\) and suppress it by writing
\(z_{i,b}=z^{(m)}_{i,b}\), unless we compare
frontiers across methods.
% For a candidate compression method, define the compressed interface \(z_{i,b}=C_{m,b}(x_i,q_i)\), with agnostic methods ignoring \(q_i\). 
Under a frozen answerer \(f\), the raw-audio and compressed interface losses are
\begin{equation}
\ell^x_i=\ell_{q_i}(f(x_i,q_i),y_i),
\qquad
\ell^z_{i,b}=\ell_{q_i}(f(z_{i,b},q_i),y_i).
\label{eq:losses}
\end{equation}
For a query family \(F\in\mathcal P\), let \(N_F=|\{i:F_i=F\}|\).  We estimate three quantities:
\begin{equation}
\widehat{\Delta}_{F}(b)
=
\frac{1}{N_F}\sum_{i:F_i=F}\bigl(\ell_{i,b}^{z}-\ell_i^x\bigr), \qquad
% \label{eq:emp-gap1}
% \\
\widehat{\Delta}^{\rm fam}_{\mathcal P}(b)
=
\max_{F\in\mathcal P:\,N_F>0}\widehat{\Delta}_{F}(b),
\label{eq:emp-gap2}
\end{equation}
\begin{equation}
\widehat{\Delta}_{\rm avg}(b)
=
\frac{1}{N}\sum_{i=1}^{N}\bigl(\ell_{i,b}^{z}-\ell_i^x\bigr),
\label{eq:emp-gap3}
\end{equation}
where \(\davg{b}\) asks how much damage appears after pooling all
queries, while \(\dQP{b}\) asks how much damage appears in the worst observed
deployment family, using per-family $\widehat{\Delta}_{F}(b)$ quantities. The family partition is fixed before observing losses, and unsupported families remain inconclusive rather than being removed from the guarantee.
We define the hidden-damage margin:
\begin{equation}
\label{eq:emp-hidden}
\widehat H_{\mathcal P}(b)
=
\widehat{\Delta}^{\rm fam}_{\mathcal P}(b)
-
\widehat{\Delta}_{\rm avg}(b),
\end{equation}
which is the degradation hidden by replacing the partition
\(\mathcal P\) with a dataset mean.  Large margins on budget-degenerate controls
can also reflect baseline heterogeneity or cancellation rather than
budget-sensitive compression damage. 

For prediction $a$ and answer $y$, we use the 0-1 multiple choice loss, defined as \({\ell(a,y)=\mathbf 1\{a\neq y\}}\). This matches our empirical claim: the benchmarks ask for one discrete answer, and compression is judged by whether it preserves answer correctness under the same LALM.  
The risk criterion extends to bounded losses, but the exact binomial endpoints below specifically require binary correctness losses.
We use two tolerances: the average tolerance \(\varepsilon_{\rm avg}\) limits
overall regression, while the family tolerance \(\varepsilon_F\) limits damage
to the worst deployment family.  
They may be set equal: when we report a single \(\varepsilon=0.05\), both tolerances are \(0.05\) in 0-1 loss units, i.e., at most five percentage points of additional error.
% They may be set equal, e.g., when we report a single
% \(\varepsilon=0.05\), it denotes the common five-percentage-point setting.

For \(r\in\{\mathrm{avg},\mathrm{fam}\}\), define
\begin{equation}
\widehat{\Delta}_{r}^{\mathcal P}(b)=
\begin{cases}
\widehat{\Delta}_{\rm avg}(b), & r=\mathrm{avg},\\
\widehat{\Delta}^{\rm fam}_{\mathcal P}(b), & r=\mathrm{fam},
\end{cases}
\qquad
\varepsilon_r=
\begin{cases}
\varepsilon_{\rm avg}, & r=\mathrm{avg},\\
\varepsilon_F, & r=\mathrm{fam}.
\end{cases}
\end{equation}

% \eqref{eq:emp-frontier} and \eqref{eq:cert-frontier} for \(r\in\{\mathrm{avg},\mathrm{fam}\}\).
%     \item \textbf{Audit conditioning:} compare agnostic and conditioned methods with \eqref{eq:emp-gcond};
    
For a budget grid $B$, the empirical point frontier for a fixed method is
\begin{equation}
\widehat b^\star_{\mathcal P}(\varepsilon_r,\mathcal Q;r)
=
\min_{b\in B}
\left\{b:\widehat\Delta^{\mathcal P}_r(b)\le\varepsilon_r\right\},
\label{eq:emp-frontier}
\end{equation}
with value $\infty$ for an empty feasible set. Pointwise confidence intervals alone do not certify a budget selected by searching this grid.
For a method \(m\), write
\({\widehat b^\star_{\mathcal P}(\varepsilon_r,\mathcal Q;r\mid m)}\)
for the same frontier computed after generating compressed interfaces with
method \(m\). For the agnostic and conditioned selectors, the operational
conditioning gain is
\begin{equation}
\widehat G^{\rm op}_{\rm cond}(\varepsilon_r,\mathcal Q;r)
=
\widehat b^\star_{\mathcal P}(\varepsilon_r,\mathcal Q;r\mid \mathrm{agn})
-
\widehat b^\star_{\mathcal P}(\varepsilon_r,\mathcal Q;r\mid \mathrm{cond}).
\label{eq:emp-gcond}
\end{equation}
Positive gain means that the conditioned selector reaches the same tolerance
with less retained-audio budget, not that it has higher accuracy at a fixed
budget.  Unless stated otherwise, the headline conditioning claim uses
\(r=\mathrm{fam}\), since sign-off is family-wise. 

The original curve intervals use paired example bootstraps~\citep{efron1993bootstrap}. Across-seed Student-$t$ intervals summarize selector-training variability. These are descriptive rather than post-selection certificates. Estimator and bootstrap details are in App.~\ref{app:estimators}.

\phantomsection
\label{par:independent-confirmation}
\paragraph{Independent confirmation.}
Choosing a compression setting and evaluating it on the same recordings can make approval overly optimistic. We therefore separate selection from confirmation: sample independent recordings, choose one question uniformly per recording, and split the recordings into 20\% development and 80\% confirmation. This gives recordings equal weight, whereas the selector sweeps average question rows. The model, compressor, preprocessing, prompts, scoring, candidate grid, and query families are fixed before confirmation.

On development data, we select the lowest-cost setting whose estimated error increase meets the five-point tolerance overall, $c_{\rm avg}$, and the lowest-cost setting that meets it for every family, $c_{\rm fam}$. Raw audio is available as a zero-degradation fallback. If any family lacks development examples, its preservation cannot be assessed, so $c_{\rm fam}$ defaults to raw audio. We also identify the supported family most damaged by $c_{\rm avg}$, denoted $F^\dagger$. The two settings and this witness family are then locked before confirmation outcomes are inspected.

For a fixed candidate and population or family, let $n_+$ count raw-correct/compressed-wrong pairs and $n_-$ count raw-wrong/compressed-correct pairs among $n$ independent confirmation recordings. Write $u_t(k,n)$ and $l_t(k,n)$ for the one-sided Clopper-Pearson binomial endpoints with tail error $t$. Set $\eta=0.05/3$ and $t=\eta/2$ (exact formulas in App.~\ref{app:exact-confirmation}):
\begin{equation}
U=u_t(n_+,n)-l_t(n_-,n),\qquad
L=l_t(n_+,n)-u_t(n_-,n).
\label{eq:paired-exact-bounds}
\end{equation}
Each endpoint has one-sided coverage at least $1-\eta$. On the confirmation data, test three assertions:
\begin{equation}
% \begin{aligned}
\mathrm{A}: U_{\rm avg}(c_{\rm avg})\le\varepsilon, \qquad
\mathrm{B}: L_{F^\dagger}(c_{\rm avg})>\varepsilon, \qquad
\mathrm{C}: \max_{F\in\mathcal P}U_F(c_{\rm fam})\le\varepsilon.
% \end{aligned}
\label{eq:decision-rule}
\end{equation}
A asks whether the average-selected setting meets the overall tolerance, B asks whether that same setting nevertheless harms the preselected family beyond tolerance, and C asks whether the family-selected alternative meets every family's requirement. Together, they test whether considering families changes a defensible deployment decision.

Bonferroni across A-C limits the probability of any false assertion to 5\% within one dataset-model-compressor evaluation. Because C requires every family to pass, a false approval entails an erroneous upper bound for at least one unsafe family. No additional family correction is needed for this decision (App.~\ref{app:exact-confirmation}). This does not give joint coverage of every displayed family interval. Missing family support or incomplete confirmation prevents approval of a compressed setting. Failure to establish an assertion means insufficient evidence, not proof of its opposite. Raw-only approval means retaining uncompressed audio. A revised setting requires fresh confirmation data. These guarantees concern the evaluated population, not distribution shift or all experiments jointly.

\paragraph{Audit procedure.}
Algorithm~\ref{alg:signoff-protocol} summarizes the audit for a candidate audio compressor. An evidence map is in App.~\ref{app:algorithm-evidence-map}. 

\begin{center}
\footnotesize
\refstepcounter{protocol}\label{alg:signoff-protocol}
{\setlength{\fboxsep}{3pt}\fbox{\begin{minipage}{0.965\linewidth}
\textbf{Algorithm~\theprotocol: sign-off protocol for a candidate compressor.}
\begin{enumerate}[leftmargin=1.1em,itemsep=0pt,topsep=0pt,parsep=0pt]
    \item \textbf{Specify:}
    Fix the LALM, compressor, candidate settings and costs, query-family
    partition, and tolerance $\varepsilon=0.05$
    (Section~\ref{sec:experiments-main}).
    Sample and split recordings into development and confirmation data
    as described in
    \hyperref[par:independent-confirmation]{``Independent confirmation''}
    in Section~\ref{sec:estimating-main}.

    \item \textbf{Evaluate paired paths:}
    Run the same LALM on raw and compressed audio for each sampled
    question, and record the paired binary losses in \eqref{eq:losses}.
    Keep confirmation outcomes sealed until the choices in Step~4
    are locked.

    \item \textbf{Expose hidden damage:}
    On development data, compute per-family and pooled error increases
    using \eqref{eq:emp-gap2}-\eqref{eq:emp-gap3}.
    Identify the worst-affected family and compute the hidden-damage
    margin using \eqref{eq:emp-hidden}.

    \item \textbf{Select and lock:}
    Follow the development-data selection rules in
    \hyperref[par:independent-confirmation]{``Independent confirmation''} in Section~\ref{sec:estimating-main}:
    choose the lowest-cost overall-passing setting $c_{\rm avg}$,
    the lowest-cost all-family-passing setting $c_{\rm fam}$,
    and the supported family $F^\dagger$ most damaged by $c_{\rm avg}$.
    Apply the stated raw-audio fallback and missing-family rules.
    Lock these choices before inspecting confirmation outcomes.

    \item \textbf{Confirm:}
    On confirmation data, count harmful and beneficial answer flips
    and compute the bounds in \eqref{eq:paired-exact-bounds}
    using the endpoints in App.~\ref{app:exact-confirmation}.
    Test A-C in \eqref{eq:decision-rule}.
    Report the number of evaluated recordings in each family.

    \item \textbf{Accept, reject, or withhold approval:}
    Using \eqref{eq:decision-rule}, \textbf{accept} $c_{\rm fam}$ if
    $\max_{F\in\mathcal P} U_F(c_{\rm fam})\le\varepsilon$ (C):
    every evaluated family meets the tolerance.
    \textbf{Reject} $c_{\rm avg}$ if
    $L_{F^\dagger}(c_{\rm avg})>\varepsilon$ (B):
    the preselected family exceeds the tolerance, even if the
    overall check A passes.
    A setting with neither established family-wise acceptance nor
    established above-tolerance harm is \textbf{inconclusive};
    withhold deployment approval.
    Raw-only acceptance means no compression.
    A revised setting requires fresh confirmation data.
\end{enumerate}
\end{minipage}}}
\end{center}

\section{Data and Experimental Setup}
\label{sec:experiments-main}
We evaluate with five English-prompted multiple-choice audio-question benchmarks:
DCASE 2026 dev, a single-family control under the keyword partition \citep{dcase2026adqa};
AudioMCQ-StrongAC, a multi-family benchmark and a selector-training
source \citep{he2025audiomcq}; MMSU, a speech and prosody taxonomy stress test
\citep{wang2025mmsu}; MMAR, a mixed speech, audio, and music reasoning check
\citep{ma2025mmar}; and BigBench Audio, a text-dominated control for
budget-degenerate behavior \citep{artificialanalysis2024bigbenchaudio,suzgun2022bbh}.
We evaluate four LALMs: Qwen2-Audio-7B-Instruct~\citep{chu2024qwen2audio}, Qwen2.5-Omni-7B~\citep{xu2025qwen25omni}, MiniCPM-o-4.5~\citep{minicpm2026model}, and MERaLiON-2-10B~\citep{meralion2025model}. All four receive identical compressed waveforms in the comparison in Table~\ref{tab:four-model-main}. Learned-selector sweeps use Qwen2-Audio and Qwen2.5-Omni and independent confirmation uses Qwen2-Audio and MiniCPM.

% Our empirical compressor is hard chunk retention: audio of example $i$ is split into \(N_i\) one-second chunks, scored, and reduced to the top \(k_i(b)=\max(1,\lfloor bN_i\rfloor)\) chunks in temporal order. The retained fraction, determined by budget $b$, is an operational length-cost proxy, not a calibrated bitrate, latency, or entropy-rate measurement. Learned selectors are small MLP scoring heads over 768-dimensional log-mel chunk features. The conditioned version adds a 128-dimensional query embedding. They are trained on AudioMCQ-StrongAC chunk-relevance targets with BCE/KL-style losses and AdamW (learning rate \(10^{-4}\), weight decay \(0.01\), batch 32, cosine warmup), using Gumbel-softmax in training and hard top-\(k\) inference~\citep{jang2017gumbel}. At test time the selector sees only the audio and query, not labels or evaluation outcomes. The original AudioMCQ selector sweep is treated as in-domain evidence; it does not establish performance on a content-disjoint selector test. Other datasets are zero-shot transfers where null or negative gains may reflect domain shift as well as limits of query conditioning. 
We evaluate five compressions methodologies as chunk-retention policies: uniform, random, energy-based, learned query-agnostic, and learned query-conditioned. Audio is first divided into $N$ one-second chunks. At budget $b$, a selector retains $k(b)=\max(1,\lfloor bN\rfloor)$ chunks in temporal order. Learned selectors use MLP scoring heads over 768-dimensional log-mel chunk features, with a 128-dimensional query embedding added in the conditioned case. The objective matches soft chunk-relevance targets by KL loss with count and entropy terms. Targets are normalized nonnegative increases in correct-answer negative log-likelihood when a chunk is removed from a frozen reference model's audio input. At inference, selectors use only audio and, when applicable, the question. The original AudioMCQ selector sweep is in-domain. The other datasets test zero-shot transfer. Implementation details are in App.~\ref{app:protocols}.

\paragraph{Confirmation experiments.}
We evaluate two distinct compression methods: uniform retention at $b\in\{0.1,0.2,0.4,0.6,0.8,0.95\}$ and the EnCodec neural audio codec~\citep{defossez2023encodec} at nominal rates $\{1.5,3,6,12\}$~kbps, selected on development data before confirmation.
AudioMCQ uses 6,400 confirmation sources, with query families
defined by source dataset and question type. A separate
four-family evaluation uses 5,580 fresh confirmation sources.
MMSU uses 3,812 confirmation waveform groups.
Each group contributes one question, and recordings are at
most 30 seconds. EnCodec reduces bitrate while preserving
audio duration, so bitrate savings do not imply faster LALM
inference. Protocol details are in App.~\ref{app:forward-experiments}. 

% The main sweep uses Qwen2-Audio at \(b\in\{0.05,0.1,0.2,0.4,1\}\), the query-use audit adds \(b\in\{0.01,0.02\}\), seeds \(\{42,123,456\}\), both backbones, and same- or cross-backbone gain estimates when applicable.s
% Unless a result explicitly states otherwise, the hidden-damage results use Qwen2-Audio, the learned query-conditioned selector, \(B=\{0.05,0.1,0.2,0.4,1\}\), and \(b=0.2\) for one-budget summaries. 
% Table~\ref{tab:v2-gcond-main} reports replicated conditioned-minus-agnostic frontier gains over three seeds and both backbones. Tables~\ref{tab:phase-b-dgcond-main}-\ref{tab:contamination-ratio-main} report the selector-query audit, which additionally evaluates \(b\in\{0.01,0.02\}\).

\paragraph{Query-family partitions.}
We report three partitions:
The keyword (kw) partition maps dataset category strings into shared coarse
families, e.g., \texttt{speech\_content}, \texttt{paralinguistic},
\texttt{sound\_event}, \texttt{sound\_scene}, \texttt{music},
\texttt{temporal}, and \texttt{general}.  The native partition uses
benchmark-provided task labels or taxonomies.  The semantic partition clusters query text using sentence embeddings and cosine \(K_{\rm sem}\)-means clustering. We use one fixed, predeclared semantic partition per dataset, merge rare clusters before reporting, and do not sweep \(K_{\rm sem}\) or choose it by the resulting gap. The effective post-merge family counts are in Table~\ref{tab:data-partition-map}.  The three partitions are not assumed to
form one nested chain, and they should be read as alternative operational scopes for
family-wise sign-off. 
The main sweep uses native partitions. 
Dataset and query details are in App.~\ref{app:datasets}.

\begin{figure}[t!]
\centering
\includegraphics[width=\linewidth]{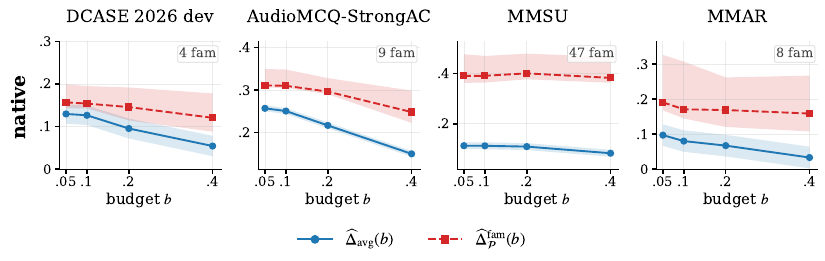}
\caption{Compression damage across native query families with the learned query-conditioned selector and Qwen2-Audio. Red dashed curves show worst-family excess error~\eqref{eq:emp-gap2} and blue curves show overall excess error~\eqref{eq:emp-gap3}. Bands are paired-bootstrap intervals. The full partition is in App.~\ref{app:full-partition-curves}.}
\label{fig:partition-gap-3x4}
% \vspace{-1ex}
\end{figure}

% \paragraph{Decoupled selector-query audit.}
% For query-conditioned selectors, we permute only the selector's query while the LALM still receives the true query. The resulting change in the empirical budget frontier measures downstream query use. App.~\ref{app:secondary-query-audits} gives the definition and controls.

\section{Main Results}
\label{sec:main-results}
Accuracies are reported in percent. Excess-risk and hidden-damage quantities are
reported as percentage points (pp) of 0-1 error. Budget frontiers and gains are
reported as retained-budget points, i.e., \(100\times\) retained-audio fraction.
For example, a gain of 4.75 retained-budget points means that the same tolerance is met while retaining 0.0475 less of the audio.
Higher accuracy is better and lower excess risk, hidden
damage, and required budget are better. Positive conditioning gains mean that query-guided compression requires a smaller audio budget to meet the same accuracy tolerance. Conditioning results are in App.~\ref{app:additional-control-tables}.
Unless stated otherwise, reported 95\% intervals use the
paired bootstrap described in Section~\ref{sec:estimating-main}.

\paragraph{Averages hide family damage, and the partition determines what is hidden.}
Figure~\ref{fig:partition-gap-3x4} uses Qwen2-Audio-7B-Instruct with the learned query-conditioned selector. The result is that, in the multi-family budget-sensitive panels, the worst-family excess risk \(\dQF{b}\) in \eqref{eq:emp-gap2} lies above the dataset mean \(\davg{b}\) in \eqref{eq:emp-gap3}. A budget that looks acceptable on average can therefore still damage the worst deployment family. Table~\ref{tab:data-partition-map} provides a snapshot form: DCASE and BigBench have zero keyword gap because the keyword partition has one
effective family, while their native partitions expose 5.04 pp and 39.9 pp
hidden-damage margins.  Where one partition truly refines another, this is the empirical counterpart of \Cref{thm:partition-refinement}, otherwise it is the operational warning that sign-off is partition-indexed. BigBench is the control case: because its compressed predictions are effectively budget-invariant, its large native margin is evidence of baseline cancellation under a degenerate compression frontier.  In the original descriptive decision analysis at $\varepsilon=0.05$, switching between keyword and native families changes 29 of 125 matched pass/inconclusive/fail outcomes. The practitioner takeaway is to report both \(\davg{b}\) and \(\dQF{b}\) for the declared deployment partition, rather than certifying a compressor from the mean alone.

\begin{table}[t!]
\centering
\footnotesize
\setlength{\tabcolsep}{2.5pt}
\caption{For each benchmark with Qwen2-Audio as the fixed model: number of examples $N$, raw-audio accuracy \(\mathrm{Acc}_{X}=100(1-N^{-1}\sum_i \ell_i^x)\) in \%, partition family counts (Section~\ref{sec:experiments-main}), and hidden-damage margin \(100\times\widehat H_{\mathcal P}(0.2)\)~\eqref{eq:emp-hidden} under the learned query-conditioned selector.}
\label{tab:data-partition-map}
\begin{tabular}{@{}lrrrrrrrr@{}}
\toprule
& & & \multicolumn{3}{c}{\#families} & \multicolumn{3}{c}{hidden-damage margin} \\
\cmidrule(lr){4-6}\cmidrule(l){7-9}
Dataset & examples & $\textrm{Acc}_{X}\%$ & kw & native & semantic & kw & native & semantic \\
\midrule
DCASE 2026 dev & 1{,}607 & 44.4 & 1 & 4 & 5 & 0 & 5.04 & 7.68 \\
AudioMCQ-StrongAC & 19{,}480 & 74.2 & 4 & 9 & 12 & 6.79 & 7.96 & 14.29 \\
MMSU & 5{,}000 & 55.3 & 6 & 47 & 12 & 1.56 & 29.17 & 22.31 \\
MMAR & 1{,}000 & 49.6 & 3 & 8 & 6 & 1.88 & 10.18 & 13.34 \\
BigBench Audio$^{\dagger}$ & 1{,}000 & 87.7 & 1 & 4 & 1 & 0 & 39.9 & n/a \\
\bottomrule
\end{tabular}\\[2pt]
{\footnotesize $^{\dagger}$single-family control under the keyword partition. semantic partition collapses to a single cluster.}
% \vspace{-1ex}
\end{table}
% \begin{table}[h]
% \centering
% \footnotesize
% \setlength{\tabcolsep}{2.5pt}
% \caption{For each benchmark with Qwen2-Audio as the fixed model: number of examples $N$, raw-audio accuracy \(\mathrm{Acc}_{X}=100(1-N^{-1}\sum_i \ell_i^x)\) in \%, partition family counts (Section~\ref{sec:experiments-main}), and hidden-damage margin \(100\times\widehat H_{\mathcal P}(0.2)\)~\eqref{eq:emp-hidden} under the learned query-conditioned selector.}
% \label{tab:data-partition-map}
% \begin{tabular}{@{}lrrrrrrrr@{}}
% \toprule
% & & & \multicolumn{3}{c}{\#families} & \multicolumn{3}{c}{hidden-damage margin} \\
% \cmidrule(lr){4-6}\cmidrule(l){7-9}
% Dataset & examples & $\textrm{Acc}_{X}\%$ & kw & native & semantic & kw & native & semantic \\
% \midrule
% DCASE 2026 dev & 1{,}607 & 44.4 & 1 & 4 & 5 & 0 & 5.04 & 7.68 \\
% AudioMCQ-StrongAC & 19{,}480 & 74.2 & 4 & 9 & 12 & 6.79 & 7.96 & 14.29 \\
% MMSU & 5{,}000 & 55.3 & 6 & 47 & 12 & 1.56 & 29.17 & 22.31 \\
% MMAR & 1{,}000 & 49.6 & 3 & 8 & 6 & 1.88 & 10.18 & 13.34 \\
% BigBench Audio$^{\dagger}$ & 1{,}000 & 87.7 & 1 & 4 & 1 & 0 & 39.9 & n/a
% \\
% \bottomrule
% \end{tabular}
% {\footnotesize $^{\dagger}$single-family control under the keyword partition. semantic partition collapses to a single cluster.}
% \vspace{-1ex}
% \end{table}

% Candidate main-body insertion. All numerical results below are transcribed
% from the author's human-review rebuttal; no historical recomputation was done.
% Provenance and interpretation notes: rebuttal_insert_notes.md.
\paragraph{Answer preservation across four LALMs.}
Table~\ref{tab:four-model-main} compares all four LALMs on the same 200 AudioMCQ examples, balanced over four query families (50 examples each). Each model receives identical compressed waveforms from uniform and energy-based retention at $b\in\{0.1,0.2,0.4\}$. All 24 estimated error increases are positive and 23 paired-bootstrap 95\% intervals exclude zero (10,000 family-stratified draws). Thus degradation recurs across answerers and both selection methods. This shared-waveform comparison does not identify model-independent information loss.

\begin{table}[t]
\centering
\footnotesize
\setlength{\tabcolsep}{6pt}
\caption{Four LALMs on the same 200 AudioMCQ examples, using identical compressed waveforms. Error increases span uniform and energy-based retention at $b\in\{0.1,0.2,0.4\}$. Ranges are not CIs. The last column counts the six paired-bootstrap 95\% intervals wholly above zero.}
\label{tab:four-model-main}
\begin{tabular}{@{}lrrr@{}}
\toprule
LALM & Raw accuracy (\%) & Error increase (pp) & CIs above zero \\
\midrule
Qwen2-Audio-7B & 70.5 & $9$-$17$ & 6/6 \\
Qwen2.5-Omni-7B & 80 & $7.5$-$13$ & 6/6 \\
MiniCPM-o-4.5 & 78.5 & $2.5$-$8$ & 5/6 \\
MERaLiON-2-10B & 60 & $6$-$9.5$ & 6/6 \\
\bottomrule
\end{tabular}
\end{table}

% A content-disjoint 2,000-example comparison also finds that better audio features can improve mean accuracy without improving worst-family preservation. Across the full-width and capacity-matched feature comparisons, worst-family loss rises in 20 of 24 contrasts (App.~\ref{app:additional-control-tables}).

\paragraph{Compression methods damage different query families.}
Figure~\ref{fig:full-grid-audiomcq} reports all six retention settings and four EnCodec rates on the same 6,400 AudioMCQ sources. Increasing nominal retention from $0.1$ to $0.95$ reduces pooled excess error from 24.12 to 1.81~pp for Qwen and from 9.98 to 1.02~pp for MiniCPM. At $0.1$ retention, speech has the largest coarse-family point estimate in both models (30.12 and 14.13~pp). At 1.5~kbps, the largest EnCodec estimates instead occur for music with Qwen (10.54~pp) and sound with MiniCPM (6.18~pp). At 12~kbps, pooled losses are 0.59 and 0.50~pp, with pointwise intervals including zero. The audit thus reveals how family-level damage changes with the compression method and setting. Details are in App.~\ref{app:full-grid-results}.

\begin{figure}[t]
\centering
\includegraphics[width=0.9\linewidth]{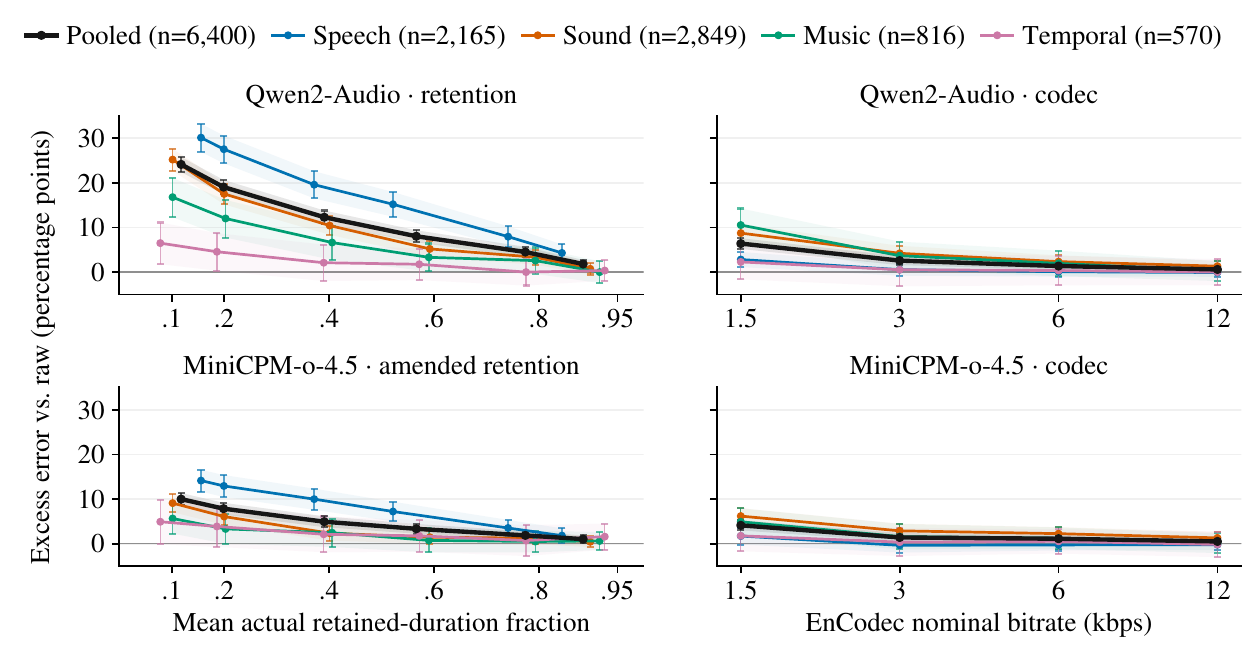}
\caption{Complete AudioMCQ curves on the original 6,400-source cohort. Black shows pooled excess error and colors show four prespecified question types. Bars and bands are pointwise 95\% paired Clopper-Pearson intervals. Retention uses group-specific measured duration fractions, while EnCodec uses nominal bitrate.}
\label{fig:full-grid-audiomcq}
\end{figure}

\paragraph{Independent confirmation and neural audio coding.}
The following bounds use \eqref{eq:paired-exact-bounds}.
On the original 6,400-source AudioMCQ holdout, the primary Qwen retention setting $b=0.8$ increases overall error by $4.50\ [3.24,5.75]$~pp. Neither pooled approval nor above-tolerance harm in its preselected family is established. Qwen with EnCodec at 3~kbps passes the pooled check ($2.58\ [1.50,3.66]$~pp), whereas MiniCPM at 1.5~kbps remains inconclusive ($4.11\ [2.90,5.31]$~pp). Missing development support for one fine family forces the original AudioMCQ family-selected policies to raw audio. This fallback reflects missing evidence, not demonstrated harm from every compressed setting.

\begin{table}[t]
\centering
\footnotesize
\setlength{\tabcolsep}{4pt}
\caption{Fresh AudioMCQ: excess answer error (pp) for Qwen2-Audio on 5,580 sources. Settings were selected on development data. Brackets use \eqref{eq:paired-exact-bounds}, without joint coverage of all displayed intervals. Speech and sound have 1,991 and 2,511 sources. Full decisions are in Appendix~\ref{app:forward-experiments}.}
\label{tab:forward-family-profiles}
\begin{tabular}{@{}lccc@{}}
\toprule
Compression setting & Overall & Speech & Sound\\
\midrule
Retention $b=0.8$ & $5.23\ [3.87,6.59]$ & $8.59\ [6.03,11.11]$ & $4.18\ [2.22,6.13]$\\
Retention $b=0.95$ & $2.46\ [1.37,3.53]$ & $4.67\ [2.64,6.67]$ & $1.51\ [-0.02,3.04]$\\
EnCodec 3~kbps & $2.63\ [1.50,3.77]$ & $0.45\ [-0.94,1.84]$ & $4.98\ [3.02,6.91]$\\
\bottomrule
\end{tabular}
\end{table}

On the fresh holdout, average-based selection chooses $b=0.8$, while family-based selection chooses $b=0.95$. The preselected speech family at $0.8$ has an error increase above the five-point tolerance even at its lower bound (Table~\ref{tab:forward-family-profiles}). Moving to $0.95$ improves observed accuracy by 2.78~pp overall and 3.92~pp on speech, while actual retained duration rises from 77.4\% to 88.4\%. All-family approval remains unresolved: speech's upper bound exceeds five points. Thus the audit identifies a useful budget change without asserting that the higher budget is certified.

Retention at $0.95$ and EnCodec at 3~kbps have similar overall point estimates, but different family profiles: retention causes more observed speech damage, while EnCodec causes more sound damage. These are descriptive comparisons, not an efficiency ranking across different resource measures. EnCodec passes the pooled check, but music and sound prevent all-family approval. 

On MMSU, EnCodec at 1.5~kbps gives another concrete contrast: Qwen loses 2.99~pp overall versus 25.00~pp on emotion recognition. MiniCPM loses 3.02 versus 20.21~pp on non-verbal sound detection. Across these confirmation experiments, no selected compressed setting is confirmed to meet the error tolerance for every required query family.
% MMSU's unresolved original-source independence precludes unconditional population certificates (Appendix~\ref{app:forward-experiments}).

% \paragraph{Query-conditioned compression is a regime map, not a method label.}
% \Cref{thm:condadv} says that conditioning can save budget when questions depend on different answer-relevant audio factors, but a learned selector and fixed LALM need not realize that gain. At $\varepsilon=0.05$, the mean retained-budget savings on AudioMCQ-StrongAC are 4.75 points for Qwen2-Audio and 2.15 for Qwen2.5-Omni. 
% MMSU is negative or per-seed negative with wide spread, DCASE and MMAR are mixed, and BigBench is zero (Table~\ref{tab:v2-gcond-main}, App.~\ref{app:additional-control-tables}). 
% Bracketed values are across-seed Student-\(t\) intervals~\citep{goutis1992increasing} over selector-training stochasticity. 
% Use conditioned compression only in dataset-family-backbone regimes where the signed gain is positive and stable; otherwise prefer agnostic selection or collect more evidence under \eqref{eq:decision-rule}.

\section{Conclusion}
\label{sec:conclusion-main}
We provide an acceptance criterion and practical protocol for auditing whether a supplied audio compressor preserves a fixed LALM's answer accuracy across specified query families. Across five benchmarks and four LALMs, the experiments show compression-induced degradation, damage hidden by dataset averages, and dependence on the evaluated families. 
We do not introduce an improved compressor, but the audit identifies failures and compression decisions that can guide subsequent improvements and enhance reliability of deployed audio compression techniques. Limitations are discussed in App.~\ref{app:extended-limitations}. 

% \paragraph{Limitations.}
% The audit supports decisions for the evaluated LALM, recording population, and query families. It does not guarantee every answer or performance after distribution shift. Sparse families can prevent approval. The learned-selector AudioMCQ sweeps are in-domain, and MMSU recording independence remains unresolved. Our bitrate measurements do not establish faster inference. Observed answer degradation also does not isolate information destroyed by compression from information the LALM cannot use. The App. opens with the complete limitations (App.~\ref{app:extended-limitations}).

\clearpage
\subsection*{AI Use Statement}
Generative AI tools assisted with aid or polish writing, and were reviewed by the human authors. The authors are responsible for the scientific claims, experimental evidence, and final text.
\subsection*{Ethical Statement}
No ethical concerns to disclose.
\subsection*{Reproducibility Statement}
A .zip file has been attached as a supplementary material to this submission. It reproduces the empirical results in the main body and the appendix.

\bibliographystyle{iclr2027_conference}
\bibliography{iclr2027_references_corrected}

@inproceedings{herrmann2020channel,
  title={Channel selection using {Gumbel} softmax},
  author={Herrmann, Charles and Bowen, Richard Strong and Zabih, Ramin},
  booktitle={European conference on computer vision},
  pages={241--257},
  year={2020},
  organization={Springer}
}

@article{
arora2025landscape,
title={On The Landscape of Spoken Language Models: A Comprehensive Survey},
author={Siddhant Arora and Kai-Wei Chang and Chung-Ming Chien and Yifan Peng and Haibin Wu and Yossi Adi and Emmanuel Dupoux and Lee, Hung-yi and Karen Livescu and Shinji Watanabe},
journal={Transactions on Machine Learning Research},
issn={2835-8856},
year={2025},
url={https://openreview.net/forum?id=BvxaP3sVbA},
note={}
}

@article{blackwell1953equivalent,
  author  = {David Blackwell},
  title   = {Equivalent Comparisons of Experiments},
  journal = {The Annals of Mathematical Statistics},
  year    = {1953},
  volume  = {24},
  number  = {2},
  pages   = {265--272}
}

@article{lecam1964sufficiency,
  author={Le Cam, Lucien}, title={Sufficiency and Approximate Sufficiency},
  journal={The Annals of Mathematical Statistics}, volume={35}, number={4},
  pages={1419--1455}, year={1964}
}

@article{witsenhausen1980indirect,
  author  = {Hans S. Witsenhausen},
  title   = {Indirect Rate Distortion Problems},
  journal = {IEEE Transactions on Information Theory},
  year    = {1980},
  volume  = {26},
  number  = {5},
  pages   = {518--521}
}

@book{coverthomas1991,
  author={Cover, Thomas M. and Thomas, Joy A.}, title={Elements of Information Theory},
  year={1991}, publisher={John Wiley \& Sons}
}

@article{tishby2000information,
  title={The information bottleneck method},
  author={Tishby, Naftali and Pereira, Fernando C and Bialek, William},
  journal={arXiv preprint physics/0004057},
  year={2000}
}

@article{nagle2024fundamental,
  title={Fundamental limits of prompt compression: A rate-distortion framework for black-box language models},
  author={Nagle, Alliot and Girish, Adway and Bondaschi, Marco and Gastpar, Michael and Makkuva, Ashok Vardhan and Kim, Hyeji},
  journal={Advances in Neural Information Processing Systems},
  volume={37},
  pages={94934--94970},
  year={2024}
}

@article{enttsel2026modelaware,
  title={Model-Aware Rate-Distortion Limits for Task-Oriented Source Coding},
  author={Enttsel, Andriy and Corlay, Vincent},
  journal={arXiv preprint arXiv:2602.12866},
  year={2026}
}

@inproceedings{dcase2026adqa,
      title={Multi-Domain Audio Question Answering Benchmark Toward Acoustic Content Reasoning}, 
      author={Chao-Han Huck Yang and Sreyan Ghosh and Qing Wang and Jaeyeon Kim and Hengyi Hong and Sonal Kumar and Guirui Zhong and Zhifeng Kong and S Sakshi and Vaibhavi Lokegaonkar and Oriol Nieto and Ramani Duraiswami and Dinesh Manocha and Gunhee Kim and Jun Du and Rafael Valle and Bryan Catanzaro},
      booktitle={Proc. IEEE International Conference on Acoustics, Speech and Signal Processing (ICASSP)},
      year={2026},
      doi={10.1109/ICASSP55912.2026.11464798},
      eprint={2505.07365},
      archivePrefix={arXiv},
      primaryClass={cs.SD},
      url={https://arxiv.org/abs/2505.07365}, 
}

@inproceedings{he2025audiomcq,
    author = "He, Haolin and Du, Xingjian and Sun, Renhe and Dai, Zheqi and Xiao, Yujia and Yang, Mingru and Zhou, Jiayi and Li, Xiquan and Liu, Zhengxi and Liang, Zining and Wu, Chunyat and He, Qianhua and Lee, Tan and Chen, Xie and Zheng, Wei-Long and Wang, Weiqiang and Plumbley, Mark and Liu, Jian and Kong, Qiuqiang",
    title = "Measuring Audio's Impact on Correctness: Audio-Contribution-Aware Post-Training of Large Audio Language Models",
    booktitle = "International Conference on Learning Representations (ICLR)",
    year = "2026",
}

@inproceedings{wang2025mmsu,
    author = "Wang, Dingdong and Li, Junan and Wu, Jincenzi and Yang, Dongchao and Chen, Xueyuan and Zhang, Tianhua and Meng, Helen",
    title = "{MMSU}: A Massive Multi-task Spoken Language Understanding and Reasoning Benchmark",
    booktitle = "International Conference on Learning Representations",
    year = "2026",
    eprint = "2506.04779",
    archiveprefix = "arXiv",
    primaryclass = "cs.CL",
    url = "https://openreview.net/forum?id=yHzCDP1tXw",
}

@inproceedings{ma2025mmar,
    author = "Ma, Ziyang and Ma, Yinghao and Zhu, Yanqiao and Yang, Chen and Chao, Yi-Wen and Xu, Ruiyang and Chen, Wenxi and Chen, Yuanzhe and Chen, Zhuo and Cong, Jian and Li, Kai and Li, Keliang and Li, Siyou and Li, Xinfeng and Li, Xiquan and Lian, Zheng and Liang, Yuzhe and Liu, Minghao and Niu, Zhikang and Wang, Tianrui and Wang, Yuping and Wang, Yuxuan and Wu, Yihao and Yang, Guanrou and Yu, Jianwei and Yuan, Ruibin and Zheng, Zhisheng and Zhou, Ziya and Zhu, Haina and Xue, Wei and Benetos, Emmanouil and Yu, Kai and Chng, Eng-Siong and Chen, Xie",
    title = "{MMAR}: A Challenging Benchmark for Deep Reasoning in Speech, Audio, Music, and Their Mix",
    booktitle = "Advances in Neural Information Processing Systems",
    volume = {38},
    pages = {74837--74868},
    year = "2025",
    doi = {10.52202/085713-2252},
    eprint = "2505.13032",
    archiveprefix = "arXiv",
    primaryclass = "cs.SD",
    url = "https://proceedings.neurips.cc/paper_files/paper/2025/hash/610a7d6507d55be70c6d13d0b663227d-Abstract-Datasets_and_Benchmarks_Track.html",
}

@misc{artificialanalysis2024bigbenchaudio,
  author={{Artificial Analysis}},
  title={{Big Bench Audio}},
  howpublished={Hugging Face dataset},
  year={2024},
  url={https://huggingface.co/datasets/ArtificialAnalysis/big_bench_audio}
}

@inproceedings{suzgun2022bbh,
  title={Challenging {BIG}-Bench tasks and whether chain-of-thought can solve them},
  author={Suzgun, Mirac and Scales, Nathan and Sch{\"a}rli, Nathanael and Gehrmann, Sebastian and Tay, Yi and Chung, Hyung Won and Chowdhery, Aakanksha and Le, Quoc and Chi, Ed and Zhou, Denny and Wei, Jason},
  booktitle={Findings of the Association for Computational Linguistics: ACL 2023},
  pages={13003--13051},
  year={2023}
}

@article{chu2024qwen2audio,
  title={{Qwen2-Audio} technical report},
  author={Chu, Yunfei and Xu, Jin and Yang, Qian and Wei, Haojie and Wei, Xipin and Guo, Zhifang and Leng, Yichong and Lv, Yuanjun and He, Jinzheng and Lin, Junyang and Zhou, Chang and Zhou, Jingren},
  journal={arXiv preprint arXiv:2407.10759},
  year={2024}
}

@book{efron1993bootstrap,
  title={An introduction to the bootstrap},
  author={Efron, Bradley and Tibshirani, Robert J.},
  year={1993},
  publisher={Chapman \& Hall}
}

@article{zeghidour2021soundstream,
  title={{SoundStream}: An end-to-end neural audio codec},
  author={Zeghidour, Neil and Luebs, Alejandro and Omran, Ahmed and Skoglund, Jan and Tagliasacchi, Marco},
  journal={IEEE/ACM Transactions on Audio, Speech, and Language Processing},
  volume={30},
  pages={495--507},
  year={2022},
  publisher={IEEE}
}

@article{defossez2023encodec,
title = "High Fidelity Neural Audio Compression",
year = "2023",
journal = "Transactions on Machine Learning Research",
issn = "2835-8856",
author = "Alexandre D{\'e}fossez and Jade Copet and Gabriel Synnaeve and Yossi Adi",
}

@article{borsos2023audiolm,
  title={{AudioLM}: a language modeling approach to audio generation},
  author={Borsos, Zal{\'a}n and Marinier, Rapha{\"e}l and Vincent, Damien and Kharitonov, Eugene and Pietquin, Olivier and Sharifi, Matt and Roblek, Dominik and Teboul, Olivier and Grangier, David and Tagliasacchi, Marco and Zeghidour, Neil},
  journal={IEEE/ACM Transactions on Audio, Speech, and Language Processing},
  volume={31},
  pages={2523--2533},
  year={2023},
  publisher={IEEE}
}

@inproceedings{schuster2022confident,
  author = {Schuster, Tal and Fisch, Adam and Gupta, Jai and Dehghani, Mostafa and Bahri, Dara and Tran, Vinh and Tay, Yi and Metzler, Donald},
  title = {Confident Adaptive Language Modeling},
  booktitle = {Advances in Neural Information Processing Systems},
  volume = {35},
  pages = {17456--17472},
  year = {2022},
  publisher = {Curran Associates, Inc.},
  doi = {10.52202/068431-1269},
  url = {https://proceedings.neurips.cc/paper_files/paper/2022/file/6fac9e316a4ae75ea244ddcef1982c71-Paper-Conference.pdf}
}

@article{angelopoulos2025learn,
  author = {Angelopoulos, Anastasios N. and Bates, Stephen and Cand{\`e}s, Emmanuel J. and Jordan, Michael I. and Lei, Lihua},
  title = {Learn then test: Calibrating predictive algorithms to achieve risk control},
  journal = {The Annals of Applied Statistics},
  volume = {19},
  number = {2},
  pages = {1641--1662},
  year = {2025},
  doi = {10.1214/24-AOAS1998},
  url = {https://doi.org/10.1214/24-AOAS1998}
}

@article{xu2025qwen25omni,
  author = {Xu, Jin and Guo, Zhifang and He, Jinzheng and Hu, Hangrui and He, Ting and Bai, Shuai and Chen, Keqin and Wang, Jialin and Fan, Yang and Dang, Kai and Zhang, Bin and Wang, Xiong and Chu, Yunfei and Lin, Junyang},
  title = {{Qwen2.5-Omni} Technical Report},
  journal = {arXiv preprint arXiv:2503.20215},
  year = {2025},
  doi = {10.48550/arXiv.2503.20215},
  url = {https://arxiv.org/abs/2503.20215}
}

@misc{minicpm2026model,
  author = {{OpenBMB}},
  title = {{MiniCPM-o 4.5}},
  howpublished = {Hugging Face model card},
  year = {2026},
  url = {https://huggingface.co/openbmb/MiniCPM-o-4_5},
  note = {Accessed September 25, 2026}
}

@misc{meralion2025model,
  author = {{Institute for Infocomm Research, A*STAR}},
  title = {{MERaLiON-2-10B}},
  howpublished = {Hugging Face model card},
  year = {2025},
  url = {https://huggingface.co/MERaLiON/MERaLiON-2-10B},
  note = {Accessed September 25, 2026}
}

@inproceedings{lipping2022clothoaqa,
  author={Lipping, Samuel and Sudarsanam, Parthasaarathy and Drossos, Konstantinos and Virtanen, Tuomas},
  title={{Clotho-AQA}: A Crowdsourced Dataset for Audio Question Answering},
  booktitle={European Signal Processing Conference},
  year={2022},
  publisher={IEEE},
  doi={10.23919/EUSIPCO55093.2022.9909680}
}

\newpage
\appendix

\section{Limitations}
\label{app:extended-limitations}

\textbf{Query-family coarsening.}
The empirical audit averages fixed-answerer excess loss within declared query families before taking the family maximum. Keyword rules, dataset-native labels, and semantic clusters can hide differences within a family. The reported $\dQF{}$ values are therefore partition-indexed operational estimates; they do not estimate the worst-query Bayes risk in \eqref{eq:family-excess-risk}.

\textbf{Fixed-answerer evaluation.}
Raw and compressed risks are measured from the predictions of a fixed LALM under declared prompts, decoding, answer normalization, and loss. The resulting frontier applies to that evaluation setting. Bootstrap intervals describe sampling uncertainty, not uncertainty from model choice, prompt choice, or distribution shift. Multiple queries sharing one source recording require uncertainty estimation that accounts for their dependence.

\textbf{Model-class and interface gap.}
Observed degradation can reflect removed information, retained information the answerer cannot use, or interface mismatch. Comparisons across answerers and interfaces provide diagnostics but do not identify a causal decomposition of these sources. The audit assesses whether the tested configuration preserves answer quality regardless of the internal source of its errors.

\textbf{Query-conditioned compression assumptions.}
Conditioned compressors require the query at encoding time. The independent-query assumption is specific to the idealized comparison in Theorem~\ref{thm:condadv}; the operational audit evaluates the declared joint audio-query population.  This matches interactive question answering but not offline archival compression.  It also requires care to avoid leakage: the query may guide what to retain, but the compressor must not use answer labels or evaluation metadata unavailable at deployment.

\textbf{Dataset, model, and seed coverage.}
The original selector evidence covers five benchmarks and two LALM backbones, with a separate four-answerer control; the evaluated prompts, chunk budgets, and interfaces limit generalization. Some appendix results are single-seed or partial-coverage: V2.1 $\alpha$-sweep wins are seed-42 only; several Qwen2-Audio V2 per-seed values were reported only as aggregates; V1 Qwen2.5-Omni coverage is incomplete for DCASE, AudioMCQ, and MMSU; and some native-family analyses rely on cross-backbone rather than same-backbone estimators.

\textbf{Confirmation and codec scope.}
The completed neural-codec experiments add a bitrate-controlled interface alongside chunk retention. Their nominal and quantizer-payload rates do not establish complete serialized bitrate or LALM latency savings. MMSU waveform deduplication does not resolve original-source independence. Exception-only and minimum-length amendments are separately reported in Appendix~\ref{app:forward-amendments}. The supplementary free-response scores measure lexical reference agreement; human validation was not performed.

\textbf{Deployment scope.}
The paper gives a framework, estimators, and empirical guidance.  It does not prove that any particular compressor is universally safe, nor that a family partition chosen on a benchmark will remain correct under distribution shift, longer audio, different languages, new answerers, or different cost models.

\section{Exact Paired Confirmation Bounds}
\label{app:exact-confirmation}
For $n>0$ and $t=1/120$, define the one-sided Clopper-Pearson endpoints
\begin{equation}
\begin{aligned}
u_t(k,n)&=\begin{cases}1,&k=n,\\ \operatorname{Beta}^{-1}(1-t;k+1,n-k),&k<n,\end{cases}\\
l_t(k,n)&=\begin{cases}0,&k=0,\\ \operatorname{Beta}^{-1}(t;k,n-k+1),&k>0.\end{cases}
\end{aligned}
\end{equation}
Here $\operatorname{Beta}^{-1}$ denotes a beta-distribution quantile. Each of $n_+$ and $n_-$ has a binomial marginal under independent sampling; their mutual dependence is allowed. The union bound gives one-sided error at most $2t=\eta$ for each difference endpoint in \eqref{eq:paired-exact-bounds}. Conditional on the development choices, A and B each have false-assertion probability at most $\eta$. If the family-aware candidate violates any family tolerance, choose one such population family: C implies an erroneous upper-bound decision for that family, also with probability at most $\eta$. A union bound over A, B, and C gives error at most $3\eta=0.05$, also after averaging over development data. This is decision validity within one declared instance, not simultaneous interval coverage across families or across the study. A missing family ($n=0$) is unsupported; raw identity is exempt because its paired difference is identically zero.

\section{Additional experiments with frozen deployment decisions}
\label{app:forward-experiments}

\paragraph{Models, inputs, and compression.}
The additional experiments use frozen Qwen2-Audio-7B-Instruct and MiniCPM-o-4.5 answerers, with bfloat16 weights and deterministic decoding. The checkpoint revisions are \texttt{0a095220c30b7b31434169c3086508ef3ea5bf0a} and \texttt{1f761131fa83f5ed3cd6f2f22b225c4501d154fa}, respectively. Qwen uses constrained choice-letter decoding with at most five new tokens; MiniCPM uses its native greedy decoder with at most 16. Labels are used only for scoring. A refusal or unparsable answer is incorrect; an execution failure is missing evidence.

All conditions start from the same full mono 24\,kHz recording, limited to 30 seconds, with the same subsequent resampling to 16\,kHz for each answerer. Uniform retention uses nonoverlapping one-second chunks, preserves the final partial chunk when selected, and concatenates selected chunks chronologically. The nominal fractions are $\{0.10,0.20,0.40,0.60,0.80,0.95\}$; with $N$ chunks, $k=\max(1,\lfloor bN\rfloor)$ and selected indices are $\lfloor(j+0.5)N/k\rfloor$, $j=0,\ldots,k-1$. EnCodec~\citep{defossez2023encodec} uses its frozen 24\,kHz mono checkpoint, revision \texttt{c1dbe2ae3f1de713481a3b3e7c47f357092ee040}, at $\{1.5,3,6,12\}$\,kbps. Decoded audio retains the original duration. Raw identity is a separate candidate.

\paragraph{Cohorts and frozen choices.}
AudioMCQ contains 8,000 original-source groups, with one randomly selected question per group, split into 1,600 development and 6,400 confirmation groups. Sampling and splitting use seed 20260925. Its 16 fine families are the metadata pairs \texttt{source\_dataset::question\_type}. MMSU contains 4,764 deduplicated waveform groups, split into 952 development and 3,812 confirmation groups, with 47 native tasks. Distinct MMSU waveforms could originate from the same recording: original-source independence was not established. Its numerical bounds therefore do not establish unconditional population certificates.

The primary instance is AudioMCQ/Qwen/uniform retention. For each instance, development data select the cheapest candidate with pooled mean excess error at most $\varepsilon=0.05$, $c_{\rm avg}$, and the cheapest satisfying every family mean, $c_{\rm fam}$. Raw is included. The witness family is the largest development-supported mean excess under $c_{\rm avg}$. Choices and witnesses are frozen before confirmation. In the original AudioMCQ cohort, \texttt{MusicCaps::sound} has no development observation, forcing $c_{\rm fam}$ to raw; this is insufficient family support, not evidence that every compressed setting is harmful.

\paragraph{Interpretation of the bounds.}
For a fixed candidate and population or family, let $n_+$ count raw-correct/compressed-wrong pairs and $n_-$ count raw-wrong/compressed-correct pairs. The estimated excess error is $(n_+-n_-)/n$. With one-sided Clopper--Pearson limits $U_t,L_t$ and $t=1/120$, the endpoints for the selected-comparison A/B/C results are
\begin{equation}
\begin{aligned}
 U_\Delta&=U_t(n_+,n)-L_t(n_-,n),\\
 L_\Delta&=L_t(n_+,n)-U_t(n_-,n).
\end{aligned}
\label{eq:forward-paired-bounds}
\end{equation}
Assertion A is pooled approval of $c_{\rm avg}$ ($U_\Delta\le0.05$); B is above-tolerance harm in its frozen witness ($L_\Delta>0.05$); C is approval of $c_{\rm fam}$ in every required family. Under independent sampling and the frozen-selection procedure, the allocation supports at least 95\% simultaneous validity for A/B/C within one instance. C is an intersection--union decision, not simultaneous coverage of every displayed family interval. No study-wide or distribution-shift guarantee is asserted. For raw identity, all differences and bounds are zero by construction.

\begin{table}[t]
\centering
\scriptsize
\setlength{\tabcolsep}{3pt}
\caption{The 11 selected-comparison records. Q = Qwen2-Audio; M = MiniCPM; R = retention; E = EnCodec. Settings are nominal retention fractions or codec kbps. Pooled excess and bounds are percentage points. In A/B/C, $+$ denotes an established assertion under the row's stated scope, $?$ denotes unestablished, and I denotes raw identity only. MMSU rows ($\dagger$) have unresolved source independence; amended rows ($*$) are exploratory and their entries are numerical outcomes, not original confirmatory claims. Row X4 reuses O5 exactly.}
\label{tab:forward-inventory}
\begin{tabular}{@{}llccrlrl@{}}
\toprule
ID & Cohort & Model & Interface & $n$ & $c_{\rm avg}/c_{\rm fam}$ & Pooled excess [bounds] & A/B/C\\
\midrule
O1 & AudioMCQ & Q & R & 6,400 & $.80$/raw & $4.50\ [3.24,5.75]$ & $?/?/\mathrm{I}$\\
O2 & AudioMCQ & Q & E & 6,400 & $3$/raw & $2.58\ [1.50,3.66]$ & $+/?/\mathrm{I}$\\
O3 & AudioMCQ & M & E & 6,400 & $1.5$/raw & $4.11\ [2.90,5.31]$ & $?/?/\mathrm{I}$\\
O4 & MMSU$^{\dagger}$ & Q & E & 3,812 & $1.5$/raw & $2.99\ [1.41,4.56]$ & $+/+/\mathrm{I}$\\
O5 & MMSU$^{\dagger}$ & M & E & 3,812 & $1.5$/raw & $3.02\ [1.31,4.72]$ & $+/+/\mathrm{I}$\\
\midrule
F1 & Fresh AudioMCQ & Q & R & 5,580 & $.80/.95$ & $5.23\ [3.87,6.59]$ & $?/+/?$\\
F2 & Fresh AudioMCQ & Q & E & 5,580 & $3/3$ & $2.63\ [1.50,3.77]$ & $+/?/?$\\
\midrule
X1 & AudioMCQ$^*$ & M & R & 6,400 & $.40$/raw & $4.92\ [3.62,6.22]$ & $?/+/\mathrm{I}$\\
X2 & MMSU$^{\dagger *}$ & M & R & 3,812 & $.95$/raw & $3.73\ [2.09,5.36]$ & $?/?/\mathrm{I}$\\
X3 & MMSU$^{\dagger *}$ & Q & R & 3,812 & $.80$/raw & $6.79\ [4.95,8.63]$ & $?/?/\mathrm{I}$\\
X4 & MMSU$^{\dagger *}$, reused & M & E & 3,812 & $1.5$/raw & $3.02\ [1.31,4.72]$ & $+/+/\mathrm{I}$\\
\bottomrule
\end{tabular}
\end{table}

\paragraph{Original results and limitations.}
Five of the eight original instances have complete selected comparisons (O1--O5). The primary O1 has 452 harmful and 164 beneficial flips among 6,400 pairs: observed excess error is 4.50 points, but pooled approval and witness harm remain unestablished. O2 establishes pooled approval of EnCodec at 3\,kbps. The MMSU codec comparisons exhibit substantial witness damage (Table~\ref{tab:forward-witnesses}), while retaining the source-independence limitation. No result establishes a complete source-independence-supported A/B/C decision reversal, and no selected compressed family policy receives all-family approval. Raw-only C is not an empirically discovered compression solution.

\begin{table}[t]
\centering
\small
\setlength{\tabcolsep}{4pt}
\caption{Selected development-locked witnesses; row IDs and statistical scopes follow Table~\ref{tab:forward-inventory}. These families were not chosen from confirmation outcomes. Error increases and bounds are percentage points.}
\label{tab:forward-witnesses}
\begin{tabular}{@{}llrrr@{}}
\toprule
ID & Witness & $n$ & Harmful/beneficial & Excess [bounds]\\
\midrule
O1 & MusicCaps::music & 70 & $4/1$ & $4.29\ [-8.25,16.01]$\\
O4 & Emotion recognition & 88 & $23/1$ & $25.00\ [8.21,38.88]$\\
O5 & Non-verbal sound detection & 94 & $20/1$ & $20.21\ [5.06,33.07]$\\
F1 & Speech & 1,991 & $218/47$ & $8.59\ [6.03,11.11]$\\
X1 & SpeechCraft::speech & 1,743 & $249/46$ & $11.65\ [8.63,14.60]$\\
\bottomrule
\end{tabular}
\end{table}

\paragraph{Fresh coarse-family evaluation.}
A separately registered supplement retains the original 1,600 development groups and evaluates 5,580 previously unused, source-disjoint AudioMCQ groups. Both policy locks precede access to this holdout. The four native question types contain 629 music, 2,511 sound, 1,991 speech, and 449 temporal questions. These coarser requirements do not establish protection for the original fine families. For retention, $c_{\rm avg}=0.80$ and $c_{\rm fam}=0.95$. The frozen speech witness exceeds the tolerance at 0.80 (F1). At 0.95, pooled excess is $2.4552\ [1.3728,3.5332]$ points, while speech excess is $4.6710\ [2.6432,6.6742]$: its upper endpoint alone prevents C. Pooled accuracy rises descriptively from 69.44\% to 72.22\%, against raw 74.68\%; this improvement is not an additional registered assertion. For EnCodec, both policies select 3\,kbps; A holds, but music and sound upper endpoints (5.49 and 6.91 points) prevent C. Identical policies cannot demonstrate a policy-selection reversal.

\paragraph{Measured resources.}
On the fresh cohort, actual retained-duration fractions are 0.7740 and 0.8844 at nominal 0.80 and 0.95, with mean Qwen audio-token counts 223.87 and 258.67, respectively, against raw 286.56. EnCodec 3\,kbps preserves duration and that raw token count. Its mean quantizer-index payload rate is 3,000.52 bits/s, excluding framing and entropy coding; official serialized bitrate was unavailable. Codec results therefore report nominal bandwidth and explicitly identified payload bits, not measured bitstream size. These measurements do not establish end-to-end latency or memory savings.

\subsection{Complete descriptive compression grids}
\label{app:full-grid-results}

All eight dataset--model--interface grids are complete under their stated protocols or amendments: two datasets, two models, and uniform retention or EnCodec. Each of the six retention fractions and four codec rates has 6,400 paired AudioMCQ sources or 3,812 paired MMSU waveform groups. Figure~\ref{fig:full-grid-audiomcq} uses the original AudioMCQ cohort, not the separate 5,580-source evaluation. Its four question types contain 2,165 speech, 2,849 sound, 816 music, and 570 temporal questions. These summaries retain the original question per source and do not replace the 16 fine-family requirements or remedy missing development support.

For these descriptive curves, \eqref{eq:forward-paired-bounds} uses tail $t=0.0125$: a union bound over the four marginal Clopper--Pearson tails gives pointwise 95\% interval coverage. This differs from the registered A/B/C allocation $t=1/120$. The intervals do not give simultaneous coverage across settings or families and do not justify selecting a newly certified budget. Retention points use each displayed group's mean actual retained signal fraction. At nominal $b=0.10$, this fraction is 0.156 for speech and 0.102 for sound and music; equal nominal budgets need not retain equal duration fractions. EnCodec uses a separate nominal-kbps axis. Neither axis establishes inference latency or memory savings.

AudioMCQ/MiniCPM retention uses the exception-only amendment: one recording is repaired at each of six settings, while original failures remain recorded. MMSU/Qwen retention uses its exception-only amendment, and MMSU/MiniCPM retention completes the existing minimum-length amendment on the exposed cohort. Figure~\ref{fig:full-grid-mmsu} reports the complete pooled MMSU curves. Their unresolved original-source independence prevents unconditional population certification. The completed grids leave the primary designation, original policy locks, tolerance, and registered A/B/C outcomes unchanged.

\begin{figure}[!htbp]
\centering
\includegraphics[width=\linewidth]{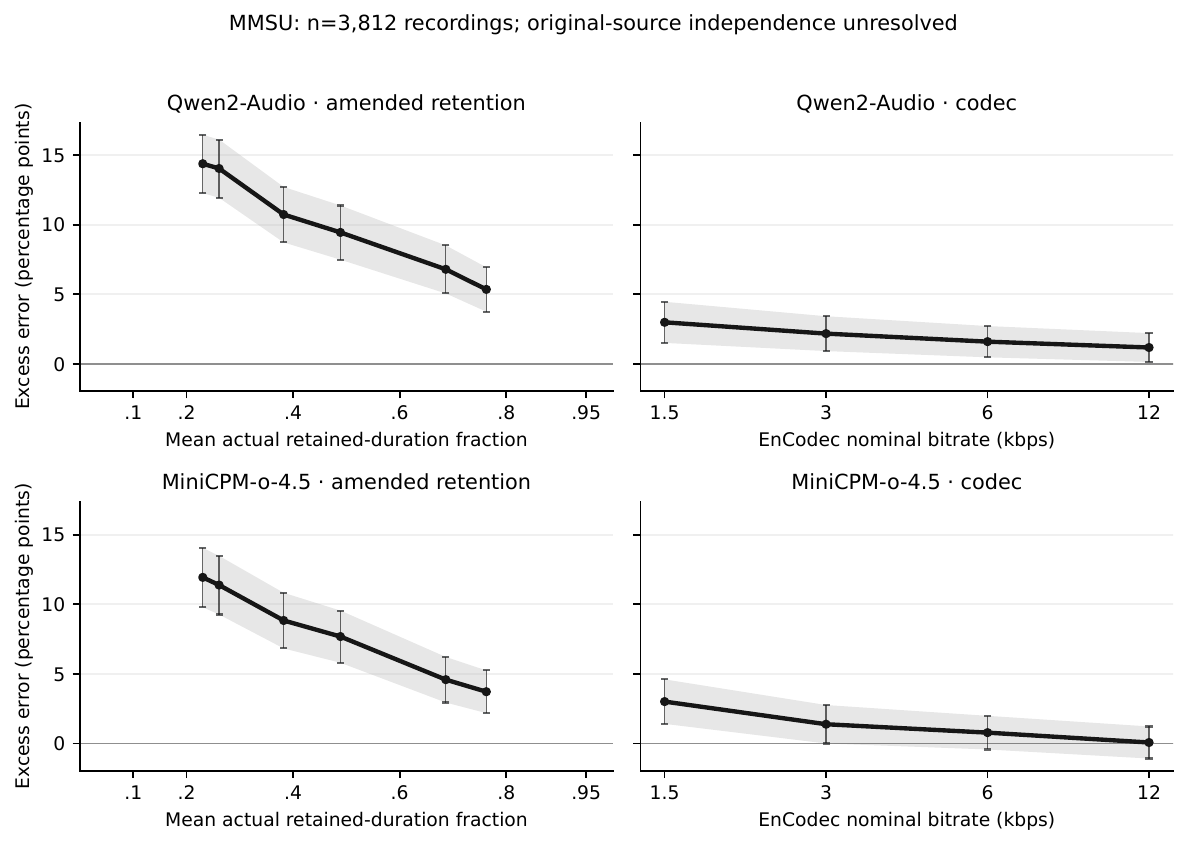}
\caption{Complete descriptive MMSU curves, with 3,812 paired waveform groups per setting. Retention uses the disclosed amendments, and codec runs use the original protocol. Intervals are pointwise 95\% calculations conditional on the unresolved source-independence assumption. Duration and bitrate are separate resource axes. These results do not change the registered decisions.}
\label{fig:full-grid-mmsu}
\end{figure}

\subsection{Incomplete original comparisons and disclosed amendments}
\label{app:forward-amendments}

The original AudioMCQ/MiniCPM retention comparison retains 6,399 of 6,400 selected pairs. Original MMSU/Qwen retention retains 3,807 of 3,812. Original MMSU/MiniCPM retention has 6,634 of 6,664 development-condition outputs and no original policy lock. Missing outputs were not dropped, scored as model errors, or overwritten by amendments.

Three separate retention amendments complete those evaluations. MiniCPM's exception-only extension calls the original path first and pads only identified failing subthreshold inputs to 1,281 samples at 16\,kHz; X1 uses 6,399 unchanged successes and one repaired answer. Qwen's corresponding minimum is 321 samples; X3 uses 3,807 unchanged successes and five repaired answers. Their original development choices and witnesses remain unchanged. The separate MMSU/MiniCPM padding evaluation X2 applies the 1,281-sample minimum and selects nominal retention 0.95. These modifications followed exposure to original confirmation results and remain exploratory. MMSU source lineage also remains unresolved. No codec or raw input required MiniCPM padding, so X4 reuses O5 exactly and is not a new replication.

As a disclosed sensitivity analysis, all compatible binary completions of the missing original answers leave pooled approval unestablished. AudioMCQ/MiniCPM pooled excess ranges from 4.9219 to 4.9375 points; its fully observed witness remains above tolerance. MMSU/Qwen ranges from 6.6632 to 6.7943 points, with both A and B unestablished in every completion. These are completion ranges, not confidence intervals or manufactured predictions. The amendments do not alter those A/B classifications or complete the original pipelines retroactively.

\subsection{Native free-response short-answer evaluation}
\label{app:free-response}

We additionally evaluate Qwen2-Audio on 473 official Clotho-AQA~\citep{lipping2022clothoaqa} test recordings, selecting one native non-yes/no short-answer question per recording with seed 20260925. No multiple-choice options or answer vocabulary are supplied. The prompt requests a short answer, preferably one word; greedy generation allows at most 32 new tokens. The five fixed conditions yield 2,365 completed answers. Three released recordings lack original Freesound URLs; their released filenames and decoded-waveform duplicate checks provide the available provenance.

Scoring applies Unicode NFKC normalization, lowercasing, punctuation-to-space conversion, and whitespace collapse. Exact agreement accepts any normalized official reference. Bounded token-F1 loss is $1-\max_{r\in\mathcal R_i}\mathrm{F1}(\hat y_i,r)$, where $\mathcal R_i$ contains the official answers. Empty answers and refusals receive loss one. Table~\ref{tab:forward-free-response} reports paired loss changes against raw audio; intervals use 2,000 paired recording bootstrap resamples, seed 20260925, and percentile 2.5--97.5\% endpoints.

\begin{table}[t]
\centering
\small
\setlength{\tabcolsep}{5pt}
\caption{Complete native free-response results ($n=473$ per condition). Exact agreement is a percentage and token-F1 lies in $[0,1]$; loss changes and descriptive interval endpoints are signed differences on that scale. Positive loss change means worse lexical agreement.}
\label{tab:forward-free-response}
\begin{tabular}{@{}lrrrr@{}}
\toprule
Condition & Exact (\%) & Token-F1 & Loss increase & Descriptive 95\% interval\\
\midrule
Raw & 31.50 & 0.3485 & 0 & $[0,0]$\\
Retention 0.20 & 27.91 & 0.2953 & 0.0532 & $[0.0192,0.0877]$\\
Retention 0.80 & 30.02 & 0.3304 & 0.0181 & $[-0.0039,0.0391]$\\
EnCodec 1.5\,kbps & 29.81 & 0.3183 & 0.0302 & $[-0.0030,0.0633]$\\
EnCodec 12\,kbps & 30.66 & 0.3403 & 0.0082 & $[-0.0123,0.0287]$\\
\bottomrule
\end{tabular}
\end{table}

Aggressive retention at 0.20 increases bounded F1 loss, with a descriptive interval excluding zero. The other compressed conditions have positive point changes with intervals containing zero; these results establish neither degradation nor preservation in those conditions. The scores measure lexical reference agreement, not human-rated semantic correctness, and valid alternative wording may be penalized. \textbf{Human validation was not performed.} This supplementary evaluation provides no A/B/C certificate, familywise semantic guarantee, or study-wide inference.

\section{Additional Model, Conditioning, and Feature Controls}
\label{app:additional-control-tables}
\paragraph{Four-model control.}
Table~\ref{tab:four-model-main} in the main body reports the model-specific results on the 200-example AudioMCQ cohort. The six conditions combine uniform and energy-based retention at $b\in\{0.1,0.2,0.4\}$. Raw and compressed waveforms are held fixed across the four answerers. Intervals are paired, family-stratified bootstrap intervals.

\paragraph{Conditioning gains.}
The full three-seed comparison supplements the numerical examples in the main text.
\begin{table}[!htbp]
\centering
\footnotesize
\setlength{\tabcolsep}{4pt}
\caption{Operational conditioning gain
\(\widehat G^{\rm op}_{\rm cond}(\varepsilon_F,\mathcal Q;\mathrm{fam})\)
at \(\varepsilon_F=0.05\) on the retained-budget frontier. Entries are means
over seeds \(\{42,123,456\}\), reported in retained-budget points, where one
point equals \(0.01\) retained-audio fraction. Positive values mean that the
conditioned selector reaches the same familywise tolerance with less retained
audio than the agnostic selector. Brackets are symmetric across-seed
Student-\(t\) 95\% CIs over selector-training stochasticity, not paired
example-bootstrap CIs. Sign labels indicate per-seed sign agreement.}
% \caption{Replicated operational conditioned gain at \(\varepsilon=0.05\) on the retained-budget frontier. Entries are three-seed means, reported as budget-fraction points \(\times100\). CIs are symmetric across-seed Student-\(t\) 95\% intervals~\citep{goutis1992increasing} and represent selector-training stochasticity, not sample-level bootstrap uncertainty. Sign labels follow per-seed sign agreement, and CI-based sign decisions are stricter when the across-seed std is large.}
\label{tab:v2-gcond-main}
\begin{tabular}{
@{}l
S[table-format=-1.2]@{\,[}S[table-format=-2.1]@{,\,}S[table-format=+2.1]@{]}l
S[table-format=-1.2]@{\,[}S[table-format=-2.1]@{,\,}S[table-format=+2.1]@{]}l
@{}}
\toprule
& \multicolumn{4}{c}{Qwen2-Audio} & \multicolumn{4}{c}{Qwen2.5-Omni} \\
\cmidrule(lr){2-5}\cmidrule(l){6-9}
Dataset
& \multicolumn{3}{c}{\(\widehat G^{\rm op}_{\rm cond}\) [95\% CI] (pts)} & \, sign
& \multicolumn{3}{c}{\(\widehat G^{\rm op}_{\rm cond}\) [95\% CI] (pts)} & \, sign \\
\midrule
DCASE 2026 dev
& -3.00 & -15.2 &  9.2 & \, mixed
& -3.10 & -12 &  5.8 & \, mixed \\
AudioMCQ-StrongAC
&  4.75 &  2.3 &  7.2 & \, consistent\(+\)
&  2.15 &  0.2 &  4.1 & \, consistent\(+\) \\
MMSU
& -3.36 & -6.3 & -0.4 & \, consistent\(-\)
& -7.34 & -18.2 &  3.5 & \, consistent\(-\) \\
MMAR
& -5.40 & -27.5 & 16.7 & \, mixed
&  0.80 & -10.1 & 11.7 & \, mixed \\
BigBench Audio
& \multicolumn{3}{c}{\(0\)} & \, zero
& \multicolumn{3}{c}{\(0\)} & \, zero \\
\bottomrule
\end{tabular}\\[2pt]
\vspace{-1ex}
\end{table}

\paragraph{Capacity-matched feature comparison.}
On a content-disjoint set of 2,000 AudioMCQ examples, separate PCA transforms reduce log-mel and native audio features to 512 dimensions for identical 19,740,930-parameter selectors. Native-minus-mel mean accuracy improves by 1.58, 1.63, 1.80, and 1.38 percentage points at retained budgets 0.05, 0.1, 0.2, and 0.4. Across the full-width and PCA-controlled comparisons together, worst-family 0/1 loss is higher with native features in 20 of 24 seed-by-budget contrasts, while mean loss is conclusively lower in 9 of 24. These descriptive comparisons show that the feature representation can improve average accuracy without resolving worst-family preservation. They do not identify a causal decomposition of information loss and model error.

\section{Complete Partition-Curve Comparison}
\label{app:full-partition-curves}
\begin{figure}[!htbp]
\centering
\includegraphics[width=\linewidth]{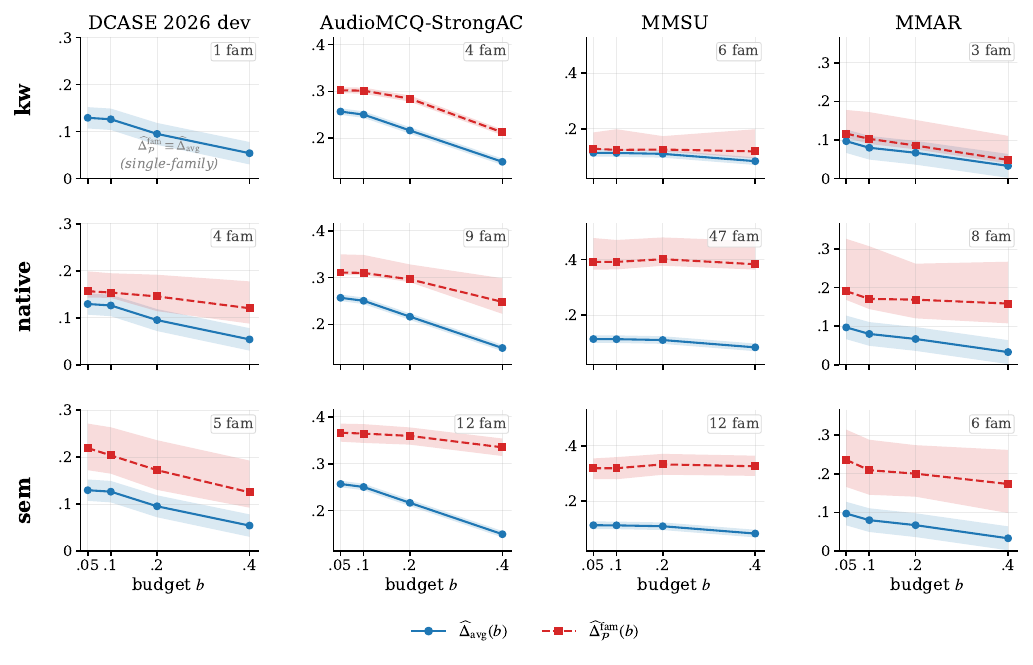}
\caption{Complete version of Figure~\ref{fig:partition-gap-3x4}: keyword, native, and semantic partitions (rows), with Qwen2-Audio and the learned query-conditioned selector. Red dashed curves show worst-family excess error and blue curves show overall excess error. Bands are paired-bootstrap intervals. The main figure retains the native-partition row.}
\label{fig:full-partition-gap}
\end{figure}

\section{Appendix Roadmap}
\label{app:roadmap}
Appendix~\ref{app:forward-experiments} reports the confirmation experiments, amendments, and supplementary free-response evaluation. Appendix~\ref{app:additional-control-tables} provides further context for the four-model comparison, conditioning results, and feature comparison. The remaining material supplements the original analysis.  \Cref{app:related} adds related-work context.  \Cref{app:extended-setup} records supplementary problem-setup and operational-budget details.  \Cref{app:theory-proofs} gives proofs and auxiliary derivations for the theory stated in the main text.  \Cref{app:estimators} defines the practical estimators, uncertainty aggregation, and axis conventions.  \Cref{app:datasets} and \Cref{app:protocols} document datasets, query families, models, compression methods, budgets, and implementation details.  \Cref{app:additional-results} contains all secondary figures, tables, ablations, and family-level analyses.  \Cref{app:notes} records incomplete or suspicious results and caveats.  
% \Cref{app:checklist} provides the checklist.

\section{Supplementary Algorithm Evidence Map}
\label{app:algorithm-evidence-map}
\begin{table}[!htbp]
\centering
\scriptsize
\caption{Supplementary evidence map for Algorithm~1.  This table supplements the main audit; it records how the main empirical objects map to practitioner actions.}
\label{tab:algorithm_evidence_appendix}
\resizebox{\linewidth}{!}{%
\begin{tabular}{p{0.20\linewidth}p{0.30\linewidth}p{0.24\linewidth}p{0.20\linewidth}}
\toprule
\textbf{Algorithm step} & \textbf{Empirical evidence} & \textbf{Meaning} & \textbf{Practitioner consequence} \\
\midrule
Specify partition and tolerances & Table~\ref{tab:data-partition-map}: native partitions expose margins hidden by keyword partitions. & The partition is part of the claim. & Report the sign-off partition and tolerance before approving a budget. \\
Run paired raw/compressed interfaces & Figure~\ref{fig:partition-gap-3x4}: worst-family excess exceeds the dataset mean. & Compression damage is answer-level paired excess risk. & Evaluate raw and compressed answers under the same frozen LALM and query. \\
Expose hidden damage & Per-family excess losses identify the families blocking approval. & Remediation depends on the failing-family structure. & Use bypass/retraining for local harm; raise budget/change interface for diffuse harm. \\
Evaluate conditioning & Table~\ref{tab:v2-gcond-main}: conditioning is positive, negative, mixed, or zero depending on dataset/backbone. & Conditioning is regime-dependent. & Report signed, seeded, backbone-specific conditioned gain. \\
Audit query use & Table~\ref{tab:phase-b-dgcond-main}: downstream query-use is substantial only in one completed cell. & Selector sensitivity need not change final answers. & Use the decoupled audit for conditioned compressors. \\
Check contamination & Table~\ref{tab:contamination-ratio-main}: naive shadow queries inflate the signal. & Prompt mismatch is not selector query-use. & Treat shadow-query runs as controls, not headline evidence. \\
\bottomrule
\end{tabular}%
}
\end{table}

\section{Supplementary Selector-Query Audits}
\label{app:secondary-query-audits}
\paragraph{Decoupled selector-query audit.} For query-conditioned selectors, let \(s_i\) denote the query text supplied to
the selector.  The anchor run uses \(s_i=q_i\).  A permuted run uses
\(s_i=\sigma(q_i)\), where \(\sigma\) is drawn from a finite set \(\Sigma\) of
valid query permutations, either globally or within family, while the downstream
LALM still receives the true query \(q_i\).  We average over \(|\Sigma|\le 10\)
valid permutations per training seed. Thus the audit changes only the
selector-side query signal and does not ask the LALM the wrong question.  For \(\varepsilon_F=0.05\), we report
\begin{equation}
\label{eq:selector-query-audit}
\Delta\widehat G_{\rm cond}^{\rm op}
=
\widehat G_{\rm cond}^{\rm op}(0.05,\Qfam;\mathrm{fam}\mid s_i=q_i)
-
\frac{1}{|\Sigma|}\sum_{\sigma\in\Sigma}
\widehat G_{\rm cond}^{\rm op}(0.05,\Qfam;\mathrm{fam}\mid s_i=\sigma(q_i)),
\end{equation}
which means that \eqref{eq:emp-gcond} is recomputed
with the indicated selector-side query stream. Positive values mean the
correct selector query changes the downstream answer-preservation frontier. 
% Reporting bands use PHI1 for \(|\Delta\widehat G|\le0.01\), PHI2 for \(0.01<\Delta\widehat G<0.05\), and PHI3 for \(\Delta\widehat G\ge0.05\).
% \paragraph{Decoupled downstream selector-query audit.}
% For example $i$ and query-conditioned selectors, let \(s_i\) denote the query text supplied to
% the selector.  The anchor run uses \(s_i=q_i\).  A permuted run uses
% \(s_i=\pi(q_i)\), where \(\pi\) is drawn either globally or within family, while the
% downstream LALM still receives the true query \(q_i\).  We average up to ten valid permutations per training seed. Thus, the audit changes
% only the selector-side query signal and does not ask the LALM the wrong
% question.  For the family metric at \(\varepsilon_F=0.05\), we report
% \begin{equation}
% \label{eq:selector-query-audit}
% \Delta\widehat G_{\rm cond}^{\rm op}
% =
% \widehat G_{\rm cond}^{\rm op}(0.05,\Qfam;\mathrm{fam}\mid s_i=q_i)
% -
% \frac{1}{|\Pi|}\sum_{\pi\in\Pi}
% \widehat G_{\rm cond}^{\rm op}(0.05,\Qfam;\mathrm{fam}\mid s_i=\pi(q_i)),
% \end{equation}
% where the conditional notation means that \eqref{eq:emp-gcond} is recomputed
% with the indicated selector-side query stream.  A positive value means that the
% correct selector query changes the downstream answer-preservation frontier. Reporting bands use PHI1 for \(|\Delta\widehat G|\le0.01\), PHI2 for \(0.01<\Delta\widehat G<0.05\), and PHI3 for \(\Delta\widehat G\ge0.05\).

\paragraph{Selector query use must survive downstream.}
Selector-internal query sensitivity is necessary but insufficient. Table~\ref{tab:phase-b-dgcond-main} recomputes \(\widehat G_{\rm cond}^{\rm op}\) after permuting only the selector-side query, as the LALM still receives the true query. 
Reporting bands use PHI1 for
\(|\Delta\widehat G^{\rm op}_{\rm cond}|\le0.01\), PHI2 for
\(0.01<|\Delta\widehat G^{\rm op}_{\rm cond}|<0.05\), and PHI3 for
\(|\Delta\widehat G^{\rm op}_{\rm cond}|\ge0.05\).
Both AudioMCQ-StrongAC cells are PHI3 substantial and exclude zero, MMSU is null-to-small with CIs including zero, and BigBench has a zero observed query-use gain. Thus selected-chunk changes count as deployment evidence only when the decoupled audit changes the frontier seen by the frozen answerer.

\begin{table}[!htbp]
\centering
\footnotesize
\setlength{\tabcolsep}{4pt}
\caption{Decoupled selector-query audit
\(\Delta\widehat G^{\rm op}_{\rm cond}\) at \(\varepsilon_F=0.05\).
Only the selector query is permuted; the LALM receives the true query.
Entries are retained-budget points, where one point equals \(0.01\)
retained-audio fraction. Positive values mean that the correct selector-side
query improves the downstream retained-budget frontier. PHI bands use the
magnitude thresholds defined in the text.}
\label{tab:phase-b-dgcond-main}
\begin{tabular}{
@{}ll
S[table-format=+1.2]@{\,[}S[table-format=+1.1]@{,\,}S[table-format=+2.1]@{]\,}c
@{}}
\toprule
Backbone & Dataset, regime
& \multicolumn{3}{c}{\(\Delta\widehat G^{\rm op}_{\rm cond}\) [95\% CI] (pts)} & band\\
\midrule
Qwen2-Audio
& AudioMCQ-StrongAC, global
& +7.87 & +5.1 & +10.6 & \textbf{PHI3}\\

Qwen2-Audio
& MMSU, global
& +1.11 & -4.6 &  +6.9 & PHI2\\

Qwen2-Audio
& MMSU, within-family
& +0.92 & -5.3 &  +7.1 & PHI1\\

Qwen2-Audio
& BigBench Audio, global
& \multicolumn{3}{c}{\(0\)} & PHI1\(^\dagger\)\\

Qwen2.5-Omni
& AudioMCQ-StrongAC, global
& +7.17 & +4.5 &  +9.9 & \textbf{PHI3}\\

Qwen2.5-Omni
& MMSU, global
& -0.43 & -4.5 &  +3.7 & PHI1\\

Qwen2.5-Omni
& MMSU, within-family
& -0.05 & -2.8 &  +2.7 & PHI1\\

Qwen2.5-Omni
& BigBench Audio, global
& \multicolumn{3}{c}{\(0\)} & PHI1\(^\dagger\)\\
\bottomrule
\end{tabular}\\[2pt]
{\footnotesize \(^\dagger\)Observed zero query-use gain; a single family alone does not imply identical query text.}
\vspace{-1ex}
\end{table}

\paragraph{Naive shadow-query evaluation contaminates the audit.}
Table~\ref{tab:contamination-ratio-main} isolates an audit failure mode: if both
selector and LALM receive the permuted query, the effect mixes selector query-use
with asking the answerer the wrong question. On AudioMCQ-StrongAC with
Qwen2-Audio, the naive protocol is \(1.22\)-\(1.25\times\) the decoupled signal.
The \(\varepsilon_F=0.05\) decoupled entry intentionally matches
Table~\ref{tab:phase-b-dgcond-main}. Hence the decoupled protocol is the
headline query-use audit and shadow-query runs are contamination controls.
% \paragraph{Naive shadow-query evaluation contaminates the audit.}
% Table~\ref{tab:contamination-ratio-main} isolates an audit failure mode: if both selector and LALM receive the permuted query, the effect mixes selector query-use with asking the answerer the wrong question. On AudioMCQ-StrongAC with Qwen2-Audio, this naive protocol is $\times$1.22-1.25 the decoupled signal. The \(\varepsilon_F=0.05\) decoupled entry intentionally matches Table~\ref{tab:phase-b-dgcond-main}.
% % At \(b=0.05\), the paired naive(-)decoupled difference is \(0.017\) in retained-budget fraction, equivalently \(+1.7\) retained-budget points in the table's \(100\times\) scale, with 95\% CI \([+0.4,+3.0]\) retained-budget points.
% % At \(b=0.05\), the paired naive(-)decoupled difference is (+0.017) with 95\% CI ([+0.004,+0.03]). 
% Hence the decoupled protocol is the headline query-use audit and shadow-query runs are contamination controls.

\begin{table}[t!]
\centering
\footnotesize
\setlength{\tabcolsep}{3pt}
\newcommand{\ci}[3]{%
\makebox[2.05em][r]{\(#1\)}\,{\([\,\makebox[1.75em][r]{\(#2\)},\,\makebox[2.25em][r]{\(#3\)}\,]\)}%
}
\newcommand{\cib}[3]{%
{\boldmath\bfseries
\makebox[2.05em][r]{\(#1\)}\,{\([\,\makebox[1.75em][r]{\(#2\)},\,\makebox[2.25em][r]{\(#3\)}\,]\)}%
}%
}
\caption{Naive shadow-query contamination versus the decoupled selector-only
audit on AudioMCQ-StrongAC with Qwen2-Audio at
\(\varepsilon_F\in\{0.01,0.02,0.05\}\). The decoupled protocol permutes only
the selector query, whereas the naive protocol permutes both the selector and
LALM queries. Entries are retained-budget points, where one point equals
\(0.01\) retained-audio fraction. The naive-minus-decoupled column uses paired
runs; positive values quantify shadow-query contamination.}
\label{tab:contamination-ratio-main}
\begin{tabular}{@{}c@{\hspace{5pt}}c@{\hspace{5pt}}c@{\hspace{5pt}}c@{\hspace{5pt}}r@{}}
\toprule
\(\varepsilon_F\) &
decoupled [95\% CI] (pts) &
naive [95\% CI] (pts) &
naive\(-\)decoupled [95\% CI] (pts) &
ratio \\
\midrule
\(0.01\) &
\ci{1.9}{0}{3.8} &
\ci{2.37}{0.3}{4.4} &
\ci{0.47}{0}{1} &
\(\times1.25\)\\

\(0.02\) &
\ci{3.8}{0}{7.6} &
\ci{4.74}{0.6}{8.8} &
\ci{0.94}{-0.1}{2} &
\(\times1.25\)\\

\(0.05\) &
\ci{7.87}{5.1}{10.6} &
\ci{9.57}{6.7}{12.5} &
\ci{1.7}{0.4}{3} &
\(\times1.22\)\\
\bottomrule
\end{tabular}\\[2pt]
% {\footnotesize The \(b=0.05\) paired naive\(-\)decoupled CI excludes zero; lower-budget paired CIs touch zero but have positive per-seed signs.}
\vspace{-1ex}
\end{table}

\clearpage

\section{Additional Related Work}
\label{app:related}
The main paper cites four threads: audio-language modeling, audio tokenization/compression, rate-distortion/information bottleneck theory, and efficient multimodal inference.  Qwen2-Audio and Qwen2.5-Omni provide the frozen answerers used here \citep{chu2024qwen2audio,xu2025qwen25omni}.  Neural audio codecs and tokenizers such as SoundStream, EnCodec, and AudioLM demonstrate that audio can be mapped to compact discrete or continuous interfaces \citep{zeghidour2021soundstream,defossez2023encodec,borsos2023audiolm}.  Those methods optimize reconstruction, generation, or perceptual fidelity; TAAC asks whether the interface preserves the answer to a query family.  Classical experiment comparison and approximate sufficiency provide the decision-theoretic language \citep{lecam1964sufficiency}, while indirect rate-distortion and information bottleneck work provide the rate-constrained view \citep{tishby2000information,coverthomas1991}.  Recent prompt-compression rate-distortion work is closest in spirit, but focuses on text prompts and black-box language models rather than audio-query families \citep{nagle2024fundamental}.  Our empirical design is also related to task-family evaluation: the key methodological claim is that benchmark averages are not adequate when compression damage is heterogeneous across query families.

\section{Extended Problem Setup and Operational Budgets}\label{app:extended-setup}

We study \emph{task-aware answer-preserving compression} for large audio language models. Let \((\Omega,\mathcal F,\mathbb P)\) carry a raw audio observation \(X\in\mathcal X\), a random query \(Q\in\mathcal Q\), and a family of correct-answer random variables \(\{Y_q:q\in\mathcal Q\}\), where \(Y_q\in\mathcal Y_q\).
For example, $X$ may be a 10-second audio clip, $Q$ may be a query such as
``what word was spoken?'' or ``how many times did the bell ring?'', and $Y_q$
denotes the correct answer associated with query $q$. At inference time only the
realized query $Q$ is asked, so only $Y_Q$ is observed, while the full family
$\{Y_q\}_{q \in \mathcal{Q}}$ is a bookkeeping device that lets us formalize a
task family.
For each query \(q\), let \(\mathcal A_q\) denote the action space of predicted answers, typically \(\mathcal A_q=\mathcal Y_q\), and let
\begin{equation}
\ell:\mathcal A_q\times \mathcal Y_q\to [0,\infty)
\end{equation}
be the task loss. A compressor \(C\) outputs a compressed representation \(Z\), either from \(X\) alone or from \((X,q)\) in the query-conditioned setting.
Intuitively, for each query $q$ we are specifying a prediction problem. The set
$\mathcal A_q$ is the set of answers the downstream system is allowed to output,
while $\mathcal Y_q$ is the set of possible correct answers. In most settings
these coincide, so $\mathcal A_q=\mathcal Y_q$. For example, if $q$ asks
``how many times does the bell ring?'', then $\mathcal Y_q$ may be the set of
nonnegative integers; if $q$ asks ``which instrument is playing?'', then
$\mathcal Y_q$ may be a finite label set such as
$\{\text{piano}, \text{violin}, \text{drums}, \ldots\}$. The loss
$\ell(a,y)$ measures the penalty for predicting $a$ when the correct answer
is $y$, for instance the $0/1$ loss $\mathbf 1\{a \neq y\}$. The compressor
then produces a representation $Z$ of the audio: in the query-agnostic setting
$Z=C(X)$ depends only on the audio, whereas in the query-conditioned setting
$Z=C(X,q)$ may adapt to the realized query and preserve only the information
most relevant to answering that query.

\begin{assumption}[Bayes solvability and integrability]\label{app:ass:bayes}
For every \(q\in\mathcal Q\), Bayes-optimal decision rules based on \(X\) and on any admissible \(Z\) exist, and all displayed expectations are finite. For the familywise worst-case results we additionally assume bounded losses, \(0\le \ell(a,y)\le L_{\max}<\infty\), except in the Gaussian subsection where boundedness is replaced by finite second moments. Here and throughout, ``familywise'' means that the relevant quantity is taken uniformly over the entire query family $\mathcal Q$, typically through a worst-case operation such as $\sup_{q\in\mathcal Q}$, rather than for a single fixed query $q$.
\end{assumption}

\paragraph{Notation summary.}
For a probability law \(\pi\) on \(\mathcal Y_q\), define the Bayes envelope
\begin{equation}
\underline L_q(\pi):=\inf_{a\in\mathcal A_q}\int_{\mathcal Y_q}\ell(a,y)\,\pi(dy).
\end{equation}
The map \(\underline L_q\) is concave because it is the infimum of linear functionals of \(\pi\). We write
\begin{equation}
\Pi_q^X:=P_{Y_q|X,q},\qquad \Pi_q^Z:=P_{Y_q|Z,q},
\end{equation}
for the posterior laws induced by raw and compressed observations. All information quantities are measured in nats. Division by \(\log 2\) converts them to bits.

For a realized query \(q\), the Bayes-optimal risk when the decision rule sees raw audio \(X\) is
\begin{equation}
\mathcal R_X^\star(q):=\inf_{\psi:\mathcal X\to \mathcal A_q}\mathbb E\big[\ell(\psi(X),Y_q)\big]
=\mathbb E\big[\underline L_q(\Pi_q^X)\big].
\end{equation}
Similarly, when only \(Z\) is available,
\begin{equation}
\mathcal R_Z^\star(q):=\inf_{\phi:\mathcal Z\to \mathcal A_q}\mathbb E\big[\ell(\phi(Z),Y_q)\big]
=\mathbb E\big[\underline L_q(\Pi_q^Z)\big].
\end{equation}
The familywise excess Bayes risk is
\begin{equation}
\Delta_{\mathcal Q}(Z;X):=\sup_{q\in\mathcal Q}\big(\mathcal R_Z^\star(q)-\mathcal R_X^\star(q)\big).
\end{equation}
This is the answer-preservation criterion of interest: \(Z\) is \(\varepsilon\)-answer-preserving over \(\mathcal Q\) when \(\Delta_{\mathcal Q}(Z;X)\le \varepsilon\).

We distinguish three compressor regimes.
\emph{Query-agnostic compressors} map $X \mapsto Z$ at run time. Among them,
the frontiers below are \emph{family-aware}: the compressor may depend on the
fixed family $\mathcal Q$ and losses $\{\ell\}$ at design time, but not on
the realized query $q$. Concretely, these include policies that compress the
audio in the same way regardless of which question will later be asked, such
as uniform chunk selection, energy-based selection, or a learned selector that
scores chunks from audio alone. By contrast, \emph{query-conditioned compressors}
map $(X,q)\mapsto Z$ and are allowed to adapt to the realized query. For
example, if the query asks what word was spoken, the compressor may preferentially
retain speech-bearing regions, whereas if the query asks how many times a bell
rang, it may preferentially retain temporally informative event regions. We
instantiate both regimes empirically in \S\ref{sec:compression-methods}.

We also distinguish three interface budgets. The information-theoretic rates are
\begin{equation}
\mathrm{Rate}_{\mathrm{info}}(Z):=I(X;Z),\qquad 
\mathrm{Rate}_{\mathrm{info}}^{\mathrm{cond}}(Z):=I(X;Z|Q).
\end{equation}
If \(Z\) is transmitted through a token interface, let \(\tau(Z)\) be its sequence length. If the interface vocabulary has cardinality \(B\), then the operational budget associated with a \(B\)-ary code is
\begin{equation}
\mathrm{Rate}_{\mathrm{op}}(Z):=\log B\cdot \mathbb E[\tau(Z)].
\end{equation}

The information-theoretic Bayes frontier for family-aware query-agnostic compression is
\begin{equation}
R^{\star}_{\mathrm{Bayes}}(\varepsilon,\mathcal Q)
:=
\inf_{C_{\mathcal Q}:X\mapsto Z}
\Big\{
I(X;Z):\Delta_{\mathcal Q}(Z;X)\le \varepsilon
\Big\},
\end{equation}
and the corresponding sequence-length frontier is
\begin{equation}
L^{\star}_{\mathrm{Bayes}}(\varepsilon,\mathcal Q)
:=
\inf_{C_{\mathcal Q}:X\mapsto Z}
\Big\{
\mathbb E[\tau(Z)]:\Delta_{\mathcal Q}(Z;X)\le \varepsilon
\Big\}.
\end{equation}
The two budget notions play different roles. The information-theoretic rate
$I(X;Z)$ measures how much information about the raw audio $X$ is preserved in
the compressed representation $Z$: it is zero when $Z$ is independent of $X$,
and increases as $Z$ becomes more informative about $X$. By contrast, if $Z$
is transmitted through a discrete token interface, then $\tau(Z)$ is the number
of transmitted tokens and $\log B \cdot \mathbb E[\tau(Z)]$ is the corresponding
operational communication budget, where $B$ is the interface vocabulary size.
Thus $I(X;Z)$ is the idealized information cost, whereas $\mathbb E[\tau(Z)]$
and $\log B \cdot \mathbb E[\tau(Z)]$ are concrete length-based surrogates for
that cost.

\paragraph{Query-conditioned risks.}
When the compressor is allowed to depend on the realized query, \(P_{Z|X,Q}\), the relevant Bayes risk for query \(q\) is computed under the conditional law induced by \(Q=q\):
\begin{equation}
\mathcal R_{Z|q}^\star(q)
:=
\inf_{\phi:\mathcal Z\to \mathcal A_q}
\mathbb E\big[\ell(\phi(Z),Y_q)\mid Q=q\big].
\end{equation}
Under our standing interpretation that the runtime query variable is exogenous to the audio-generation mechanism and only selects the task, the raw-audio benchmark remains \(\mathcal R_X^\star(q)\). We therefore define the conditioned familywise excess Bayes risk by
\begin{equation}
\Delta_{\mathcal Q}^{\mathrm{cond}}(Z;X)
:=
\sup_{q\in\mathcal Q}
\big(
\mathcal R_{Z|q}^\star(q)-\mathcal R_X^\star(q)
\big).
\end{equation}
The corresponding query-conditioned information-theoretic and sequence-length frontiers are
\begin{equation}
R^{\star}_{\mathrm{Bayes,cond}}(\varepsilon,\mathcal Q)
:=
\inf_{C:(X,Q)\mapsto Z}
\Big\{
I(X;Z\mid Q):
\Delta_{\mathcal Q}^{\mathrm{cond}}(Z;X)\le \varepsilon
\Big\},
\end{equation}
and
\begin{equation}
L^{\star}_{\mathrm{Bayes,cond}}(\varepsilon,\mathcal Q)
:=
\inf_{C:(X,Q)\mapsto Z}
\Big\{
\mathbb E[\tau(Z)]:
\Delta_{\mathcal Q}^{\mathrm{cond}}(Z;X)\le \varepsilon
\Big\}.
\end{equation}
All later appearances of \(R^{\star}_{\mathrm{Bayes,cond}}\) and \(L^{\star}_{\mathrm{Bayes,cond}}\), including Theorem~\ref{thm:condadv}, refer to this conditioned excess-risk definition.

To model architectural restrictions of downstream large audio language models, let \(\mathcal F_{\mathrm{LALM}}\) be a model class of answerers \(f\) with \(f(z,q)\in\mathcal A_q\). Define the model-class Bayes risk
\begin{equation}
\mathcal R_{\mathcal F}^{\star}(q;Z):=
\inf_{f\in\mathcal F_{\mathrm{LALM}}}\mathbb E\big[\ell(f(Z,q),Y_q)\big].
\end{equation}
The model-class frontiers are then
\begin{equation}
R^{\star}_{\mathcal F}(\varepsilon,\mathcal Q)
:=
\inf_{C_{\mathcal Q}:X\mapsto Z}
\Big\{
I(X;Z):
\sup_{q\in\mathcal Q}\big(\mathcal R_{\mathcal F}^{\star}(q;Z)-\mathcal R_X^\star(q)\big)\le \varepsilon
\Big\},
\end{equation}
and
\begin{equation}
L^{\star}_{\mathcal F}(\varepsilon,\mathcal Q)
:=
\inf_{C_{\mathcal Q}:X\mapsto Z}
\Big\{
\mathbb E[\tau(Z)]:
\sup_{q\in\mathcal Q}\big(\mathcal R_{\mathcal F}^{\star}(q;Z)-\mathcal R_X^\star(q)\big)\le \varepsilon
\Big\}.
\end{equation}

These are the core objects of the paper: \(R^{\star}_{\mathrm{Bayes}}\) and \(L^{\star}_{\mathrm{Bayes}}\) are information-theoretic limits for answer preservation over a family of audio-question tasks, while \(R^{\star}_{\mathcal F}\) and \(L^{\star}_{\mathcal F}\) are the architecture-restricted frontiers relevant to actual large audio language models.

In this paper we instantiate the theory stated in Section~\ref{sec:theory-main} on five real audio-question benchmarks (DCASE 2026 dev, AudioMCQ-StrongAC, MMSU, MMAR, and BigBench Audio) spanning sound-event detection, speech-centric reasoning, multi-modal audio-textual reasoning, and text-dominated control tasks. We make the following empirical contributions. First, we verify the bit-level prediction of Theorem~\ref{thm:condadv}'s strict-separation construction at 66 synthetic finite-alphabet cells, matching the closed-form conditioned gain to within $10^{-9}$ bits across a sweep of query priors and factor entropies. Second, we demonstrate on three multi-family audio-QA datasets that the theorem-level family-wise excess risk is consistently larger than the dataset-mean metric existing work reports, by $2.5$-$6.8$ percentage points depending on budget. Third, we verify Theorem~\ref{app:thm:monotonicity}'s nested-family monotonicity along cumulative chains and use chain increments to identify per-dataset bottleneck families. Fourth, we quantify factor-overlap structure in natural audio query taxonomies and show that no observed family-pair summary approaches the factor-disjoint prediction of Corollary~\ref{app:cor:factor}. Fifth, we extend the V1 operational conditioned-gain test with V2 three-seed replication on both Qwen2-Audio and Qwen2.5-Omni: AudioMCQ-StrongAC changes from a V1 near-null to a reproducible positive effect on both backbones, DCASE changes from a single-seed positive signal to a mixed/negative three-seed result, and MMSU on Qwen2.5-Omni reveals a large temporal-family failure. Sixth, we run a V2.1 scope-B $\alpha$-sweep showing that training-target choice is itself a backbone-dependent variable, with improvements over the V2 baseline on all five Qwen2.5-Omni datasets. Finally, we report cross-answerer and interface contrasts as operational diagnostics; they do not identify the theoretical model-class gap $\Gamma_{\mathcal F}$.

\section{Proofs and Auxiliary Theory}\label{app:theory-proofs}

This appendix proves the results stated in Section~\ref{sec:theory-main}. The construction restricts comparison of experiments \cite{lecam1964sufficiency} to the audio-task family \(\mathcal Q\), and connects that restricted comparison to indirect rate-distortion. The resulting theory is neither a generic sufficiency theory nor a single-task rate-distortion analysis: the central object is the familywise answer-risk gap \(\Delta_{\mathcal Q}(Z;X)\).

\subsection{Task-aware answer sufficiency and restricted experiment comparison}

For two observation variables \(U\) and \(V\), both available jointly with query \(q\), write
\begin{equation}
U \succeq_{\mathcal Q}^{\mathrm{ans}} V
\quad\Longleftrightarrow\quad
\mathcal R_U^\star(q)\le \mathcal R_V^\star(q)\ \text{for all }q\in\mathcal Q.
\end{equation}
This is a \(\mathcal Q\)-restricted comparison preorder. Relative to raw audio \(X\), define the corresponding one-sided answer deficiency
\begin{equation}
\delta^{\mathrm{ans}}_{\mathcal Q}(Z\|X)
:=
\sup_{q\in\mathcal Q}
\mathbb E\!\left[
\underline L_q(\Pi_q^Z)-\mathbb E\!\left[\underline L_q(\Pi_q^X)\mid Z,q\right]
\right].
\end{equation}
When \(\mathcal Q\) is enlarged to the full class of bounded decision problems on a common answer space, the zero-deficiency relation reduces to classical comparison of experiments and approximate sufficiency in the sense of Le Cam \cite{lecam1964sufficiency}. Here the restriction to \(\mathcal Q\) is essential: we deliberately ignore losses that are irrelevant to the target audio-question family.

Intuitively, this construction compares two observations-typically the raw audio
$X$ and a compressed representation $Z$-only through the family of downstream
audio tasks we care about. The relation
$U \succeq^{\mathrm{ans}}_{\mathcal Q} V$ means that, for every query
$q \in \mathcal Q$, the best possible predictor using $U$ incurs no larger Bayes
risk than the best possible predictor using $V$. Thus, our comparison is deliberately restricted to the target audio-question
family $\mathcal Q$. The deficiency
$\delta^{\mathrm{ans}}_{\mathcal Q}(Z\|X)$ then measures how much answer quality
is lost by replacing $X$ with $Z$, in the worst query in the family. For
example, if $\mathcal Q$ contains both a speech-content query such as
``what word was spoken?'' and a temporal query such as ``how many times did the
bell ring?'', then a compressed representation that preserves speech regions but
blurs event timing may be nearly sufficient for the first query and clearly
insufficient for the second; the familywise deficiency records exactly that
worst-served task.

\begin{definition}[Exact Task-Aware Answer Sufficiency]\label{app:def:exact}
A compressed representation \(Z\) is \emph{exactly task-aware answer sufficient} for the query family \(\mathcal Q\) if for every \(q\in\mathcal Q\),
\begin{equation}
\underline L_q(\Pi_q^Z)
=
\mathbb E\!\left[\underline L_q(\Pi_q^X)\mid Z,q\right]
\qquad \text{a.s.}
\end{equation}
\end{definition}

\begin{definition}[Approximate Task-Aware Answer Sufficiency]\label{app:def:approx}
A compressed representation \(Z\) is \(\varepsilon\)-\emph{approximately task-aware answer sufficient} for \(\mathcal Q\) if
\begin{equation}
\delta^{\mathrm{ans}}_{\mathcal Q}(Z\|X)\le \varepsilon.
\end{equation}
\end{definition}

\begin{theorem}[Restricted Sufficiency-Risk Equivalence]\label{app:thm:equiv}
For every admissible compressed representation \(Z\),
\begin{equation}
\delta^{\mathrm{ans}}_{\mathcal Q}(Z\|X)=\Delta_{\mathcal Q}(Z;X).
\end{equation}
Consequently, \(Z\) is exactly task-aware answer sufficient for \(\mathcal Q\) if and only if \(\Delta_{\mathcal Q}(Z;X)=0\), and \(Z\) is \(\varepsilon\)-approximately task-aware answer sufficient if and only if \(\Delta_{\mathcal Q}(Z;X)\le \varepsilon\).
\end{theorem}

\noindent\emph{Proof.}
For each fixed \(q\),
\begin{equation}
\mathcal R_Z^\star(q)-\mathcal R_X^\star(q)
=
\mathbb E\big[\underline L_q(\Pi_q^Z)\big]-\mathbb E\big[\underline L_q(\Pi_q^X)\big].
\end{equation}
Applying the tower property to the second term gives
\begin{equation}
\mathcal R_Z^\star(q)-\mathcal R_X^\star(q)
=
\mathbb E\!\left[
\underline L_q(\Pi_q^Z)-\mathbb E\!\left[\underline L_q(\Pi_q^X)\mid Z,q\right]
\right].
\end{equation}
Taking the supremum over \(q\in\mathcal Q\) proves the identity. Nonnegativity follows from concavity of \(\underline L_q\) and
\begin{equation}
\Pi_q^Z=\mathbb E[\Pi_q^X\mid Z,q].
\end{equation}
Indeed,
\(\underline L_q(\Pi_q^Z)\ge
\mathbb E[\underline L_q(\Pi_q^X)\mid Z,q]\) almost surely. Therefore
zero deficiency is equivalent to zero Jensen gap for every query, which is
exactly the almost-sure equality in Definition~\ref{app:def:exact}; the
\(\varepsilon\)-approximate statement follows from the same identity.
\hfill\(\square\)

\paragraph{Interpretation.}
The restricted deficiency is a Jensen gap of the Bayes envelope generated by coarsening \(X\) into \(Z\). Exact answer sufficiency is therefore a \emph{local} notion: it does not merely require equality of global risks, but equality of the conditional Bayes envelope after observing \(Z\).

\paragraph{Experimental implication.}
Paired raw and compressed predictions estimate fixed-answerer excess risk. They do not estimate Bayes risk or certify information preservation; held-out operational certification requires the selection and testing procedure in Section~\ref{sec:estimating-main}.

\begin{theorem}[Monotonicity under Query-Family Refinement]\label{app:thm:monotonicity}
If \(\mathcal Q_1\subseteq \mathcal Q_2\), then for every \(Z\),
\begin{equation}
\delta^{\mathrm{ans}}_{\mathcal Q_1}(Z\|X)\le \delta^{\mathrm{ans}}_{\mathcal Q_2}(Z\|X),
\qquad
\Delta_{\mathcal Q_1}(Z;X)\le \Delta_{\mathcal Q_2}(Z;X).
\end{equation}
Hence, for every \(\varepsilon\ge 0\),
\begin{equation}
R^{\star}_{\mathrm{Bayes}}(\varepsilon,\mathcal Q_1)\le R^{\star}_{\mathrm{Bayes}}(\varepsilon,\mathcal Q_2),
\qquad
L^{\star}_{\mathrm{Bayes}}(\varepsilon,\mathcal Q_1)\le L^{\star}_{\mathrm{Bayes}}(\varepsilon,\mathcal Q_2),
\end{equation}
and likewise
\begin{equation}
R^{\star}_{\mathcal F}(\varepsilon,\mathcal Q_1)\le R^{\star}_{\mathcal F}(\varepsilon,\mathcal Q_2),
\qquad
L^{\star}_{\mathcal F}(\varepsilon,\mathcal Q_1)\le L^{\star}_{\mathcal F}(\varepsilon,\mathcal Q_2).
\end{equation}
\end{theorem}

\noindent\emph{Proof.}
All quantities are defined by a supremum over \(q\in\mathcal Q\) or an infimum over compressors satisfying such supremum constraints. Enlarging \(\mathcal Q\) can only tighten the preservation constraints, hence it can only shrink the feasible compressor set and weakly increase the required budget.
\hfill\(\square\)

\paragraph{Interpretation.}
Adding queries can only make answer preservation harder. This creates a partial order over task families: nested audio task families must yield nested rate and length frontiers.

\paragraph{Experimental implication.}
Later experiments should report nested families \(\mathcal Q_1\subseteq \mathcal Q_2\subseteq \cdots\) and verify that the empirical rate-risk and token-risk curves are ordered accordingly. Violations indicate either estimation noise or optimization failure.

\subsection{Task-aware indirect rate-distortion bounds}
\label{app:e2}
For a fixed query \(q\), define the reduced distortion
\begin{equation}
\rho_q(x,a):=
\mathbb E\!\left[\ell(a,Y_q)\mid X=x,q\right]
-\underline L_q(\Pi_q^X(x)).
\end{equation}
By construction, \(\rho_q(x,a)\ge 0\), and \(\rho_q(x,a)=0\) exactly when action \(a\) is Bayes-optimal given \(X=x\).

Now define the single-query indirect answer-preserving function
\begin{equation}
R_q^{\mathrm{ind}}(\varepsilon)
:=
\inf_{P_{A_q|X}}
\Big\{
I(X;A_q):
\mathbb E\big[\rho_q(X,A_q)\big]\le \varepsilon
\Big\}.
\end{equation}
By the reduction principle for indirect rate-distortion, \(R_q^{\mathrm{ind}}\) is the ordinary rate-distortion function of \(X\) under the reduced distortion \(\rho_q\).
% ; compare also the standard rate-distortion formulation in \cite{coverthomas1991}.

\begin{proposition}[Task-Aware Indirect Rate-Distortion Sandwich]\label{app:prop:sandwich}
For every \(\varepsilon\ge 0\), the query-agnostic frontier satisfies
\begin{equation}
\sup_{q\in\mathcal Q}R_q^{\mathrm{ind}}(\varepsilon)
\le
R^{\star}_{\mathrm{Bayes}}(\varepsilon,\mathcal Q)
\le
\inf_{\substack{P_{U|X},\,\{g_q\}\\
\mathbb E[\rho_q(X,g_q(U))]\le \varepsilon,\ \forall q}}
I(X;U).
\end{equation}
For the query-conditioned frontier,
\begin{equation}
\mathbb E_Q\!\big[R_Q^{\mathrm{ind}}(\varepsilon)\big]
\le
R^{\star}_{\mathrm{Bayes,cond}}(\varepsilon,\mathcal Q)
\le
\inf_{\substack{P_{U|X,Q},\,\{g_q\}\\
\mathbb E[\rho_q(X,g_q(U))\mid Q=q]\le \varepsilon,\ \forall q}}
I(X;U\mid Q).
\end{equation}
The analogous inequalities hold for the sequence-length frontiers after replacing the mutual-information objectives by \(\mathbb E[\tau(U)]\).
\end{proposition}

\noindent\emph{Proof.}
For the query-agnostic converse, let \(X\mapsto Z\) be any feasible universal compressor, and let \(A_q=a_q^\star(Z)\) be a Bayes-optimal query-\(q\) action based on \(Z\). Since \(A_q\) is a measurable function of \(Z\),
\begin{equation}
I(X;A_q)\le I(X;Z)
\end{equation}
by data processing. Moreover,
\begin{equation}
\mathbb E[\rho_q(X,A_q)]
=
\mathcal R_Z^\star(q)-\mathcal R_X^\star(q)
\le
\Delta_{\mathcal Q}(Z;X)
\le \varepsilon.
\end{equation}
Hence \(I(X;Z)\ge R_q^{\mathrm{ind}}(\varepsilon)\) for every \(q\), and therefore
\begin{equation}
I(X;Z)\ge \sup_{q\in\mathcal Q}R_q^{\mathrm{ind}}(\varepsilon).
\end{equation}
Taking the infimum over feasible \(Z\) proves the lower bound.

For the query-agnostic achievability bound, choose any auxiliary \(U\) and query-specific decoders \(\{g_q\}\) satisfying \(\mathbb E[\rho_q(X,g_q(U))]\le \varepsilon\) for all \(q\), and set \(Z:=U\). Then
\begin{equation}
\mathcal R_Z^\star(q)-\mathcal R_X^\star(q)
\le
\mathbb E[\rho_q(X,g_q(U))]
\le \varepsilon
\end{equation}
for every \(q\), so \(Z\) is feasible with rate \(I(X;U)\).

For the query-conditioned converse, let \(P_{Z|X,Q}\) be any feasible compressor. For each \(q\), let \(A_q=a_q^\star(Z)\) be a Bayes-optimal action under the conditional law \(Q=q\). Then
\begin{equation}
I(X;A_q\mid Q=q)\le I(X;Z\mid Q=q),
\end{equation}
and
\begin{equation}
\mathbb E[\rho_q(X,A_q)\mid Q=q]
=
\mathcal R_{Z|q}^\star(q)-\mathcal R_X^\star(q)
\le
\Delta_{\mathcal Q}^{\mathrm{cond}}(Z;X)
\le \varepsilon.
\end{equation}
Thus \(I(X;Z\mid Q=q)\ge R_q^{\mathrm{ind}}(\varepsilon)\) for every \(q\). Averaging over \(Q\) yields
\begin{equation}
I(X;Z\mid Q)
=
\mathbb E_Q\!\big[I(X;Z\mid Q=q)\big]
\ge
\mathbb E_Q\!\big[R_Q^{\mathrm{ind}}(\varepsilon)\big].
\end{equation}

For the query-conditioned achievability bound, choose any auxiliary \(U\) generated by a conditional kernel \(P_{U|X,Q}\) and query-specific decoders \(\{g_q\}\) such that
\begin{equation}
\mathbb E[\rho_q(X,g_q(U))\mid Q=q]\le \varepsilon
\qquad
\text{for all }q\in\mathcal Q.
\end{equation}
Setting \(Z:=U\), we obtain for every \(q\),
\begin{equation}
\mathcal R_{Z|q}^\star(q)-\mathcal R_X^\star(q)
\le
\mathbb E[\rho_q(X,g_q(U))\mid Q=q]
\le \varepsilon,
\end{equation}
so \(Z\) is feasible with conditional rate \(I(X;U\mid Q)\).
\hfill\(\square\)

\paragraph{Interpretation.}
The query-agnostic frontier is controlled from below by the hardest single query, because one universal interface must support every query at once. By contrast, the query-conditioned frontier averages the per-query costs, because the compressor can adapt its representation to the realized query before transmitting anything.

\paragraph{Experimental implication.}
Later experiments should compare a universal compressor against a query-conditioned compressor on the same task family. The conditioned gain should track the gap between a hardest-query requirement and the average per-query requirement, especially when different audio queries depend on different latent factors.

\begin{remark}[Reduction to classical objects]\label{rem:reductions}
If \(\mathcal Q=\{q\}\) and the action \(a\) is itself a predictive distribution on \(\mathcal Y_q\) under logarithmic loss, \(\ell(a,y)=-\log a(y)\), then \(\underline L_q(\pi)=H(\pi)\) and
\begin{equation}
\mathcal R_Z^\star(q)-\mathcal R_X^\star(q)
=
H(Y_q|Z,q)-H(Y_q|X,q).
\end{equation}
Hence \(R^{\star}_{\mathrm{Bayes}}(\varepsilon,\{q\})\) is equivalent to an information-bottleneck trade-off \cite{tishby2000information}. For discrete prompts and hard compressed prompts, this is related to black-box prompt-compression rate-distortion formulations \cite{nagle2024fundamental}; the present Bayes-optimal answer rule differs from a fixed language-model decoder. By contrast, for general audio task families, \(\delta^{\mathrm{ans}}_{\mathcal Q}\) is a supremum of Bayes-envelope gaps, so no single mutual-information distortion summarizes the entire frontier. In the single-task task-oriented source-coding regime, identifiability can collapse the remote problem to ordinary rate-distortion \cite{enttsel2026modelaware}; our familywise objective generally prevents that collapse.
\end{remark}

\subsection{Proof of Theorem~\ref{thm:partition-refinement}}
\label{app:proof-partition-refinement}

\begin{proof}
Write
\begin{equation}
d_q(Z)=\mathcal R_Z^\star(q)-\mathcal R_X^\star(q),
\qquad d(q)=d_q(Z).
\end{equation}
For any coarse cell \(F\in\mathcal P\) with
\(\mu(F)>0\), refinement gives
\begin{equation}
\bar d_F(Z)
=
\sum_{\substack{G\in\mathcal P':\,G\subseteq F\\ \mu(G)>0}}
\frac{\mu(G)}{\mu(F)}\,\bar d_G(Z),
\end{equation}
so \(\bar d_F(Z)\) is a convex combination of the fine-cell means inside
\(F\). Hence
\begin{equation}
\bar d_F(Z)
\le
\max_{\substack{G\in\mathcal P':\,G\subseteq F\\ \mu(G)>0}}
\bar d_G(Z)
\le
\Delta^{\rm fam}_{\mathcal P'}(Z;X).
\end{equation}
Taking the maximum over \(F\in\mathcal P\) gives
\begin{equation}
\Delta^{\rm fam}_{\mathcal P}(Z;X)
\le
\Delta^{\rm fam}_{\mathcal P'}(Z;X).
\end{equation}

Also,
\begin{equation}
\Delta^{\rm avg}_{\mu}(Z;X)
=
\sum_{F\in\mathcal P:\mu(F)>0}\mu(F)\bar d_F(Z)
\le
\max_{F\in\mathcal P:\mu(F)>0}\bar d_F(Z)
=
\Delta^{\rm fam}_{\mathcal P}(Z;X),
\end{equation}
and for every fine cell \(G\),
\begin{equation}
\bar d_G(Z)
=
\mathbb E_{\mu}[d(Q)\mid Q\in G]
\le
\sup_{q\in\mathcal Q}d_q(Z)
=
\Delta_{\mathcal Q}(Z;X).
\end{equation}
Therefore
\begin{equation}
\Delta^{\rm avg}_{\mu}(Z;X)
\le
\Delta^{\rm fam}_{\mathcal P}(Z;X)
\le
\Delta^{\rm fam}_{\mathcal P'}(Z;X)
\le
\Delta_{\mathcal Q}(Z;X).
\end{equation}

For the budget consequence, suppose
\(b\in\mathcal B_{\mathcal P'}(\varepsilon;C)\). Then
\begin{equation}
\Delta^{\rm fam}_{\mathcal P'}(Z_b;X)\le \varepsilon.
\end{equation}
By the inequality just proved,
\begin{equation}
\Delta^{\rm fam}_{\mathcal P}(Z_b;X)
\le
\Delta^{\rm fam}_{\mathcal P'}(Z_b;X)
\le \varepsilon,
\end{equation}
so \(b\in\mathcal B_{\mathcal P}(\varepsilon;C)\). Thus
\begin{equation}
\mathcal B_{\mathcal P'}(\varepsilon;C)
\subseteq
\mathcal B_{\mathcal P}(\varepsilon;C).
\end{equation}
Taking infima, with the convention \(\inf\emptyset=+\infty\), gives
\begin{equation}
b^\star_{\mathcal P}(\varepsilon;C)
=
\inf\mathcal B_{\mathcal P}(\varepsilon;C)
\le
\inf\mathcal B_{\mathcal P'}(\varepsilon;C)
=
b^\star_{\mathcal P'}(\varepsilon;C).
\end{equation}
\end{proof}

\subsection{Model-class restriction and interface length}
\label{app:e4}
\begin{proposition}[Model-Class Gap]\label{app:prop:modelgap}
Define the architecture gap of a compressed representation \(Z\) relative to \(\mathcal F_{\mathrm{LALM}}\) by
\begin{equation}
\Gamma_{\mathcal F}(Z;\mathcal Q)
:=
\sup_{q\in\mathcal Q}
\big(
\mathcal R_{\mathcal F}^{\star}(q;Z)-\mathcal R_Z^\star(q)
\big).
\end{equation}
Then
\begin{equation}
\sup_{q\in\mathcal Q}\big(\mathcal R_{\mathcal F}^{\star}(q;Z)-\mathcal R_X^\star(q)\big)
\le
\Delta_{\mathcal Q}(Z;X)+\Gamma_{\mathcal F}(Z;\mathcal Q).
\end{equation}
Consequently,
\begin{equation}
R^{\star}_{\mathcal F}(\varepsilon,\mathcal Q)\ge R^{\star}_{\mathrm{Bayes}}(\varepsilon,\mathcal Q),
\qquad
L^{\star}_{\mathcal F}(\varepsilon,\mathcal Q)\ge L^{\star}_{\mathrm{Bayes}}(\varepsilon,\mathcal Q).
\end{equation}
More quantitatively, any compressor satisfying \(\Delta_{\mathcal Q}(Z;X)\le \varepsilon\) and \(\Gamma_{\mathcal F}(Z;\mathcal Q)\le \gamma\) is feasible for the model-class frontier at tolerance \(\varepsilon+\gamma\).
\end{proposition}

\noindent\emph{Proof.}
For each \(q\),
\begin{equation}
\mathcal R_{\mathcal F}^{\star}(q;Z)-\mathcal R_X^\star(q)
=
\big(\mathcal R_Z^\star(q)-\mathcal R_X^\star(q)\big)
+
\big(\mathcal R_{\mathcal F}^{\star}(q;Z)-\mathcal R_Z^\star(q)\big).
\end{equation}
Taking the supremum over \(q\in\mathcal Q\) yields the first inequality. The frontier inequalities follow immediately.
\hfill\(\square\)

\paragraph{Interpretation.}
The proposition bounds model-class excess risk relative to the raw-audio Bayes benchmark by two theoretical terms. Neither term is identified by comparing raw and compressed predictions of a fixed answerer.

\paragraph{Experimental implication.}
Comparing stronger or specialized answerers on the same compressed interface $Z$ can reveal usable retained information. Such comparisons do not, without a Bayes-risk reference, identify $\Gamma_{\mathcal F}$ or fully separate information loss from answerer limitations.

\begin{proposition}[Rate-to-Token Translation]\label{app:prop:rate_token}
Let \(Z\) be transmitted through a uniquely decodable \(B\)-ary token interface, and let \(\tau(Z)\) denote the number of interface symbols. Then the lossless source-coding converse implies
\begin{equation}
H(Z)\le \log B\cdot \mathbb E[\tau(Z)].
\end{equation}
Consequently,
\begin{equation}
I(X;Z)\le H(Z)\le \log B\cdot \mathbb E[\tau(Z)].
\end{equation}

Now let \(U\) be a discrete compressed symbol and suppose the interface code may be chosen to match the law of \(U\). Then there exists a uniquely decodable \(B\)-ary code \(c\) and an induced interface random variable \(Z:=c(U)\) such that
\begin{equation}
\mathbb E[\tau(Z)]<\frac{H(U)}{\log B}+1.
\end{equation}
Therefore,
\begin{equation}
\frac{H(U)}{\log B}\le \inf_{c}\mathbb E[\tau(c(U))]<\frac{H(U)}{\log B}+1,
\end{equation}
where the infimum is over uniquely decodable \(B\)-ary codes for \(U\). If the compressor is deterministic, \(U=g(X)\), then \(H(U)=I(X;U)\), and hence
\begin{equation}
\frac{I(X;U)}{\log B}
\le
\inf_{c}\mathbb E[\tau(c(U))]
<
\frac{I(X;U)}{\log B}+1.
\end{equation}
If each interface token is instead a quantized latent carrying \(b_{\mathrm{tok}}\) bits, replace \(\log B\) by \(b_{\mathrm{tok}}\log 2\).
\end{proposition}

\noindent\emph{Proof.}
For any uniquely decodable \(B\)-ary interface code, the Kraft inequality together with the converse of lossless source coding gives
\begin{equation}
H(Z)\le \log B\cdot \mathbb E[\tau(Z)].
\end{equation}
Since \(I(X;Z)\le H(Z)\), the first claim follows immediately.

Now let \(U\) be any discrete compressed symbol. A Shannon \(B\)-ary code for \(U\) yields a uniquely decodable encoding \(c\) satisfying
\begin{equation}
\mathbb E[\tau(c(U))]<\frac{H(U)}{\log B}+1
\end{equation}
by the standard one-symbol source-coding bound \cite{coverthomas1991}. The converse lower bound
\begin{equation}
\frac{H(U)}{\log B}\le \inf_{c}\mathbb E[\tau(c(U))]
\end{equation}
follows from the first part applied to \(Z=c(U)\). If \(U=g(X)\) is deterministic, then \(H(U\mid X)=0\), so \(H(U)=I(X;U)\), yielding the final display. The quantized-latent case is identical after interpreting one token as \(b_{\mathrm{tok}}\) bits.
\hfill\(\square\)

\paragraph{Interpretation.}
Information rate and interface length are linked but not identical. The lower bound \(H(Z)\le \log B\,\mathbb E[\tau(Z)]\) is universal for any uniquely decodable interface. Near-equality, however, requires an entropy-efficient coding layer or tokenizer matched to the compressed representation. For stochastic compressors the gap between \(H(U)\) and \(I(X;U)\) is exactly the randomized-encoding overhead \(H(U\mid X)\).

\paragraph{Experimental implication.}
Later experiments should report both rate-risk and sequence-length-risk curves. When the interface is close to entropy-efficient, the horizontal axes should differ mainly by the factor \(\log B\), up to the one-token Shannon slack. Persistent deviations beyond that level diagnose tokenizer inefficiency or stochastic-encoding overhead.

\begin{theorem}[Conditioned Compression Advantage]\label{app:thm:condadv}
For any design prior \(Q\sim\mu\) with support \(\mathcal Q\), independent of \(X\),
\begin{equation}
R^{\star}_{\mathrm{Bayes,cond}}(\varepsilon,\mathcal Q)\le R^{\star}_{\mathrm{Bayes}}(\varepsilon,\mathcal Q),
\qquad
L^{\star}_{\mathrm{Bayes,cond}}(\varepsilon,\mathcal Q)\le L^{\star}_{\mathrm{Bayes}}(\varepsilon,\mathcal Q).
\end{equation}
The inequalities can be strict. In particular, let \(\mathcal Q=\{q_1,q_2\}\) with \(P(Q=q_1)=\lambda\in(0,1)\), let the audio contain two independent answer factors and irrelevant side information, \(X=(V_1,V_2,W)\), with \(H(V_1),H(V_2)>0\), let \(Y_{q_i}=V_i\), and let \(\ell_{q_i}(a,Y_{q_i})=\mathbf 1\{a\neq Y_{q_i}\}\). Then
\begin{equation}
R^{\star}_{\mathrm{Bayes}}(0,\{q_1,q_2\})=H(V_1,V_2),
\end{equation}
while
\begin{equation}
R^{\star}_{\mathrm{Bayes,cond}}(0,\{q_1,q_2\})
=
\lambda H(V_1)+(1-\lambda)H(V_2).
\end{equation}
Thus the strict rate gap is \((1-\lambda)H(V_1)+\lambda H(V_2)>0\).
\end{theorem}

\noindent\emph{Proof.}
The non-strict inequalities are immediate: a query-conditioned compressor can always ignore \(q\) and emulate any query-agnostic one.

For the strict-separation construction, query-agnostic zero-risk preservation requires measurable decoders \(g_1,g_2\) such that \(g_i(Z)=V_i\) almost surely. Therefore \((V_1,V_2)\) is a function of \(Z\), so
\begin{equation}
I(X;Z)\ge I(V_1,V_2;Z)=H(V_1,V_2).
\end{equation}
Equality is achieved by choosing \(Z=(V_1,V_2)\), which discards the irrelevant side information \(W\). In the query-conditioned setting, exact preservation only requires transmitting \(V_Q\). Choosing \(Z=V_Q\) gives conditional rate
\begin{equation}
I(X;Z\mid Q)=H(V_Q\mid Q)=\lambda H(V_1)+(1-\lambda)H(V_2).
\end{equation}
No lower conditional rate can suffice: conditioned on \(Q=q_i\), exact recovery of \(V_i\) implies \(H(V_i\mid Z,Q=q_i)=0\), and hence
\(I(X;Z\mid Q=q_i)\ge I(V_i;Z\mid Q=q_i)=H(V_i)\).
Averaging over \(Q\) proves the lower bound and the claimed strict gap.
\hfill\(\square\)

\paragraph{Interpretation.}
The theorem formalizes the intuitive advantage of query conditioning: a universal query-agnostic interface must preserve the union of all answer-relevant factors, whereas a query-conditioned interface needs only preserve the factors relevant to the realized query.

\paragraph{Experimental implication.}
A primary empirical quantity is the conditioned gain. We distinguish two estimators that share the theoretical anchor of Theorem~\ref{thm:condadv} but live in different operational units:
\begin{equation}
\widehat G_{\mathrm{cond}}^{\mathrm{info}}(\varepsilon,\mathcal Q)
:=
\widehat R^{\star,\mathrm{info}}_{\mathrm{Bayes}}(\varepsilon,\mathcal Q)
-
\widehat R^{\star,\mathrm{info}}_{\mathrm{Bayes,cond}}(\varepsilon,\mathcal Q),
\end{equation}
the rate-theoretic conditioned gain computed from $I(X;Z)$ or $I(X;Z\mid Q)$ estimates, and
\begin{equation}
\widehat G_{\mathrm{cond}}^{\mathrm{op}}(\varepsilon,
\mathcal Q)
:=
\widehat R^{\star}_{\mathcal F,b}(\varepsilon,
\mathcal Q;\mathrm{agnostic})
-
\widehat R^{\star}_{\mathcal F,b}(\varepsilon,
\mathcal Q;\mathrm{conditioned}),
\end{equation}
the operational conditioned gain computed from nominal budget fractions $b\in[0,1]$. The two are linked by Proposition~\ref{app:prop:rate_token}'s rate-to-token translation, modulo the coding slack and interface mismatch described there. In the two-factor construction, the conditioned gain is linear in the query prior; a balanced prior gives a symmetric example, not generally the maximum; the strict-separation construction of Theorem~\ref{thm:condadv} predicts $\widehat G_{\mathrm{cond}}^{\mathrm{info}}(0)=(1-\lambda)H(V_1)+\lambda H(V_2)$ for the two-factor finite-alphabet example with $P(Q=q_1)=\lambda$.

\subsection{Sharp finite-alphabet and Gaussian special cases}
\label{app:e5}
\begin{proposition}[Finite-Alphabet Exact Frontier]\label{app:prop:finite}
Assume that \(\mathcal Q=\{q_1,\ldots,q_m\}\) is finite, \(\mathcal X\) is finite, and each action space \(\mathcal A_{q_i}\) is finite. Let
\begin{equation}
\mathcal U:=\prod_{i=1}^m \mathcal A_{q_i},
\qquad
u=(u_1,\ldots,u_m)\in\mathcal U.
\end{equation}
Then the Bayes frontier admits the exact convex characterization
\begin{equation}
R^{\star}_{\mathrm{Bayes}}(\varepsilon,\mathcal Q)
=
\min_{P_{U|X}} I(X;U)
\end{equation}
subject to
\begin{equation}
\sum_{x\in\mathcal X}\sum_{u\in\mathcal U}\sum_{y\in\mathcal Y_{q_i}}
P_X(x)P_{U|X}(u|x)P_{Y_{q_i}|X,q_i}(y|x,q_i)\,
\ell_{q_i}(u_i,y)
-
\mathcal R_X^\star(q_i)
\le \varepsilon,
\quad i=1,\ldots,m.
\end{equation}
If a finite interface dictionary assigns a fixed token length \(\tau(u)\) to each \(u\in\mathcal U\), the corresponding sequence-length frontier is the linear program obtained by replacing the objective with
\begin{equation}
\sum_{x\in\mathcal X}\sum_{u\in\mathcal U}P_X(x)P_{U|X}(u|x)\tau(u).
\end{equation}
\end{proposition}

\noindent\emph{Proof.}
Take any query-agnostic compressor \(X\mapsto Z\) with query-specific decoders \(g_i:\mathcal Z\to \mathcal A_{q_i}\). Define
\begin{equation}
U:=\big(g_1(Z),\ldots,g_m(Z)\big)\in\mathcal U.
\end{equation}
Then \(U\) is a function of \(Z\), so \(I(X;U)\le I(X;Z)\), and the risk on query \(q_i\) is exactly the risk of using action \(U_i\). Hence no optimal solution is lost by optimizing directly over \(P_{U|X}\). Conversely, any \(U\in\mathcal U\) can itself serve as the compressed representation. The displayed constraints are therefore exact, and they are linear in \(P_{U|X}\). The mutual-information objective gives a finite convex program.
\hfill\(\square\)

\paragraph{Interpretation.}
For a finite task family, universal answer-preserving compression can be reduced, without loss, to compressing a \emph{tuple of query-contingent actions}. The exponential growth of \(|\mathcal U|\) with \(|\mathcal Q|\) is not an artifact of the proof; it is the exact combinatorial price of simultaneously supporting many queries.

\paragraph{Experimental implication.}
This proposition yields an exact \(\widehat R^{\star}_{\mathrm{Bayes}}\) benchmark on synthetic finite-query tasks. It parallels the exact finite optimization viewpoint recently used in prompt compression \cite{nagle2024fundamental}, except that the reproduction alphabet here is the vector of all query-contingent answers rather than a shortened prompt.

\paragraph{Common latent-answer convention in the Gaussian subsection.}
In the Gaussian special case, all queries refer to the same latent audio-factor vector \(S\in\mathbb R^m\). Equivalently, \(Y_q:=S\), \(\mathcal Y_q=\mathcal A_q=\mathbb R^m\) for every \(q\in\mathcal Q\), and the query dependence enters only through the loss \(\ell\). This models the situation in which lexical, speaker, prosodic, event, or temporal queries emphasize different coordinates of a common underlying factorization.

\begin{proposition}[Gaussian Latent Allocation]\label{app:prop:gaussian}
Let
\begin{equation}
S=(S_1,\ldots,S_m),\qquad S_j\overset{\mathrm{ind}}{\sim}\mathcal N(0,\sigma_j^2),
\end{equation}
and assume the raw observation is \(X=S\). For every query \(q\in\mathcal Q\), set \(Y_q:=S\), \(\mathcal Y_q=\mathcal A_q=\mathbb R^m\), and define the weighted quadratic loss
\begin{equation}
\ell(a,S)=\sum_{j=1}^m \alpha_{qj}(a_j-S_j)^2,
\qquad \alpha_{qj}\ge 0.
\end{equation}
Then \(\mathcal R_X^\star(q)=0\) for all \(q\), and, writing \(\log^+ t:=\max\{0,\log t\}\),
\begin{equation}
R^{\star}_{\mathrm{Bayes}}(\varepsilon,\mathcal Q)
=
\min_{\substack{0\le d_j\le \sigma_j^2\\
\sum_{j=1}^m \alpha_{qj}d_j\le \varepsilon,\ \forall q\in\mathcal Q}}
\frac{1}{2}\sum_{j=1}^m \log^+\!\frac{\sigma_j^2}{d_j}.
\end{equation}
For the query-conditioned frontier,
\begin{equation}
R^{\star}_{\mathrm{Bayes,cond}}(\varepsilon,\mathcal Q)
=
\mathbb E_Q\!\left[
\min_{\substack{0\le d_j(Q)\le \sigma_j^2\\
\sum_{j=1}^m \alpha_{Qj}d_j(Q)\le \varepsilon}}
\frac{1}{2}\sum_{j=1}^m \log^+\!\frac{\sigma_j^2}{d_j(Q)}
\right].
\end{equation}
For a fixed query \(q\), the optimizer is weighted reverse water-filling:
\begin{equation}
d_j^\star(q)=
\begin{cases}
\sigma_j^2, & \alpha_{qj}=0,\\[3pt]
\min\{\sigma_j^2,\nu_q/\alpha_{qj}\}, & \alpha_{qj}>0,
\end{cases}
\end{equation}
where \(\nu_q\ge 0\) is chosen so that \(\sum_{j=1}^m\alpha_{qj}d_j^\star(q)=\varepsilon\) whenever the constraint is active.
\end{proposition}

\noindent\emph{Proof.}
Because \(X=S\), the raw observation reveals the common latent answer object exactly, hence \(\mathcal R_X^\star(q)=0\) for every \(q\). Any compressed representation of \(S\) with per-coordinate mean-squared errors \(\{d_j\}_{j=1}^m\) incurs query-\(q\) Bayes risk
\begin{equation}
\mathcal R_Z^\star(q)=\sum_{j=1}^m \alpha_{qj}d_j.
\end{equation}
Thus the familywise constraint \(\Delta_{\mathcal Q}(Z;X)\le \varepsilon\) is equivalent to
\begin{equation}
\sum_{j=1}^m \alpha_{qj}d_j\le \varepsilon
\qquad \text{for all }q\in\mathcal Q.
\end{equation}
For independent Gaussian coordinates and additive separable quadratic distortion, the rate-distortion theorem gives the minimal rate
\begin{equation}
\frac{1}{2}\sum_{j=1}^m \log^+\!\frac{\sigma_j^2}{d_j}
\end{equation}
for any admissible distortion allocation \(\{d_j\}\) \cite{coverthomas1991}. Minimizing over all allocations satisfying the familywise constraints proves the first display.

For the query-conditioned frontier, the compressor may choose a separate distortion allocation after observing the realized query \(Q\). Under the conditioned excess-risk definition, the admissible allocations are exactly those satisfying
\begin{equation}
\sum_{j=1}^m \alpha_{Qj}d_j(Q)\le \varepsilon
\end{equation}
for the realized \(Q\). Averaging the per-query optimal Gaussian rates over \(Q\) yields the second display. The reverse-water-filling form follows from the KKT conditions of the per-query constrained minimization.
\hfill\(\square\)

\paragraph{Interpretation.}
The proposition treats the latent audio factors themselves as the common answer object, while different queries weight these factors differently. Transcript questions may heavily weight lexical coordinates, speaker questions identity coordinates, and event questions background-sound coordinates. The optimal compressor therefore allocates rate to factors, not to waveform regions per se.

\paragraph{Experimental implication.}
Later experiments should fit probe factors \(S_j\) from raw audio, estimate their variances \(\sigma_j^2\), and infer query-specific weights \(\alpha_{qj}\). The observed budget allocations should then be compared against the predicted reverse-water-filling solutions \(d_j^\star(q)\).

\begin{corollary}[Factor Relevance Decomposition]\label{app:cor:factor}
Under Proposition~\ref{app:prop:gaussian}, suppose the factor indices partition as
\begin{equation}
\{1,\ldots,m\}=B_1\cup\cdots\cup B_r,
\end{equation}
and the query family partitions as
\begin{equation}
\mathcal Q=\mathcal Q_1\cup\cdots\cup \mathcal Q_r
\end{equation}
with the property that \(\alpha_{qj}=0\) whenever \(q\in\mathcal Q_\ell\) and \(j\notin B_\ell\). Then
\begin{equation}
R^{\star}_{\mathrm{Bayes}}(\varepsilon,\mathcal Q)
=
\sum_{\ell=1}^r
\min_{\substack{0\le d_j\le \sigma_j^2,\ j\in B_\ell\\
\sum_{j\in B_\ell}\alpha_{qj}d_j\le \varepsilon,\ \forall q\in\mathcal Q_\ell}}
\frac{1}{2}\sum_{j\in B_\ell}\log^+\!\frac{\sigma_j^2}{d_j}.
\end{equation}
Thus disjoint query subfamilies acting on disjoint latent-factor blocks add in rate.
\end{corollary}

\begin{proof}
Under the stated block-separation condition, each constraint in Proposition~\ref{app:prop:gaussian} involves only the distortions \(\{d_j:j\in B_\ell\}\) for the matching block. The Gaussian rate objective is additive across independent coordinates. The feasible set is therefore a Cartesian product of blockwise feasible sets, and minimizing the additive objective over that product decomposes into the displayed sum.
\end{proof}

\paragraph{Interpretation.}
When one query block cares only about lexical content and another only about speaker identity, the universal answer-preserving rate splits into independent additive contributions. This is a precise version of the idea that unrelated audio capabilities consume separate budgets.

\paragraph{Experimental implication.}
Later experiments should report factor-group budget tables by clustering queries according to approximately disjoint relevance patterns. The corollary predicts additive rate contributions across such clusters.

\begin{corollary}[Synergy Penalty for Joint-Factor Queries]\label{app:cor:synergy}
Under Proposition~\ref{app:prop:gaussian}, take \(m=2\) with \(S_1,S_2\overset{\mathrm{i.i.d.}}{\sim}\mathcal N(0,\sigma^2)\). Let \(q_1\) emphasize factor \(1\), \(q_2\) emphasize factor \(2\), and \(q_{12}\) emphasize both jointly, via
\begin{equation}
\alpha_{q_1}=(1,0),\qquad
\alpha_{q_2}=(0,1),\qquad
\alpha_{q_{12}}=(1,1).
\end{equation}
Then for \(0<\varepsilon<\sigma^2\),
\begin{equation}
R^{\star}_{\mathrm{Bayes}}(\varepsilon,\{q_1,q_2\})
=
\log\frac{\sigma^2}{\varepsilon},
\end{equation}
whereas
\begin{equation}
R^{\star}_{\mathrm{Bayes}}(\varepsilon,\{q_1,q_2,q_{12}\})
=
\log\frac{2\sigma^2}{\varepsilon}.
\end{equation}
Hence the joint-factor query incurs the nonnegative synergy penalty
\begin{equation}
\mathrm{Syn}_{12}(\varepsilon)
:=
R^{\star}_{\mathrm{Bayes}}(\varepsilon,\{q_1,q_2,q_{12}\})
-
R^{\star}_{\mathrm{Bayes}}(\varepsilon,\{q_1,q_2\})
=
\log 2,
\end{equation}
namely one bit.
\end{corollary}

\begin{proof}
For \(\{q_1,q_2\}\), the constraints are \(d_1\le\varepsilon\) and \(d_2\le\varepsilon\), so the optimal allocation is \(d_1=d_2=\varepsilon\), giving \(\log(\sigma^2/\varepsilon)\). Adding \(q_{12}\) imposes \(d_1+d_2\le\varepsilon\). By symmetry and convexity of the Gaussian rate objective, the optimum is \(d_1=d_2=\varepsilon/2\), giving \(\log(2\sigma^2/\varepsilon)\). Subtracting the two rates yields \(\log 2\).
\end{proof}

\paragraph{Interpretation.}
Supporting two single-factor queries does not automatically support their joint query at the same tolerance. In audio terms, preserving ``what was said'' and preserving ``who said it'' separately is easier than preserving the joint answer to ``who said that exact phrase?''

\paragraph{Experimental implication.}
Later experiments should compare separate task families, such as transcript-only and speaker-only, against joint families that require simultaneous access to both factors. A systematic horizontal shift between these curves is the empirical synergy penalty predicted above.

\subsection{Experiment-facing consequences}

The theory above prescribes concrete empirical quantities. First, one should estimate a family-level surrogate for the theorem object $\Delta_{\mathcal Q}(Z;X)$. Let $\mathcal F(\mathcal Q)$ denote the partition of the query family $\mathcal Q$ into labeled subfamilies, and let $h$ be the frozen answerer. Define
\begin{equation}
\begin{aligned}
\widehat{\Delta}_{\mathcal Q}^{\mathcal F}(b)
:=
\max_{F\in\mathcal F(\mathcal Q)}
\Bigg[
&\frac{1}{N_F}\sum_{(x,q,y):\,\mathrm{family}(q)=F}
\ell\big(h(Z_b(x),q),y\big)\\
&-
\frac{1}{N_F}\sum_{(x,q,y):\,\mathrm{family}(q)=F}
\ell\big(h(X,q),y\big)
\Bigg]
\end{aligned}
\end{equation}
as the worst-family empirical excess risk at operational budget $b$, and
\begin{equation}
\widehat{\Delta}_{\mathrm{avg}}(b)
:=
\frac{1}{N}\sum_{(x,q,y)\in\mathrm{eval}}
\Big[\ell\big(h(Z_b(x),q),y\big)-\ell\big(h(X,q),y\big)\Big]
\end{equation}
as the dataset-mean empirical excess risk. The corresponding population family-level quantity is a coarsening of the worst-query fixed-answerer excess risk, not an estimate of the Bayes quantity $\Delta_{\mathcal Q}(Z;X)$, and the empirical family maximum is an upper bound on the empirical dataset mean by $\max\geq\mathrm{mean}$ on the declared partition. The gap $\widehat{\Delta}_{\mathcal Q}^{\mathcal F}(b)-\widehat{\Delta}_{\mathrm{avg}}(b)$ is itself an empirical signal of worst-family damage that dataset-mean metrics do not capture. On multi-family datasets we report both quantities with bootstrap CIs; on single-family datasets, the two estimators coincide by construction.

Second, on synthetic finite-query tasks, Proposition~\ref{app:prop:finite} yields an exact plug-in benchmark $\widehat R^{\star}_{\mathrm{Bayes}}$; on Gaussian factor tasks, Proposition~\ref{app:prop:gaussian} yields a fitted analytic proxy. Third, one should estimate both $\widehat G_{\mathrm{cond}}^{\mathrm{info}}(\varepsilon,\mathcal Q)$ on exact synthetic constructions and $\widehat G_{\mathrm{cond}}^{\mathrm{op}}(\varepsilon,\mathcal Q)$ on real audio. Fourth, one should estimate the model-class gap
\begin{equation}
\widehat\Gamma_{\mathcal F}(Z;\mathcal Q)
=
\max_{q\in\mathcal Q}
\big(\widehat{\mathcal R}^{\star}_{\mathcal F}(q;Z)-\widehat{\mathcal R}^{\star}_{Z}(q)\big)
\end{equation}
only with justified estimates of both the model-class and unrestricted Bayes risks on the same interface. The learned-native and heuristic-baseline summaries in \S\ref{sec:discussion-gamma} instead report operational contrasts and do not estimate this gap. V1 reports operational frontiers $\widehat R_{\mathcal F}^{\star}$ against a single backbone and treats the observed frontiers as architecture-specific rather than attempting a Bayes-versus-architecture decomposition. Fifth, one should report the monotonicity increments across nested query families and the factor allocations $\widehat d_j^{\star}(q)$ or water levels $\widehat\nu_q$ when latent-factor probes are available. These are exactly the quantities later needed for empirical rate-risk curves, empirical sequence-length-risk curves, plug-in proxies for $\widehat R^{\star}_{\mathrm{Bayes}}$, the gap between query-conditioned and query-agnostic compression, monotonicity across nested query families, and factor-specific budget differences.

\section{Estimator Details}\label{app:estimators}\label{sec:metrics}

This section specifies the empirical estimators used throughout the paper. Two design choices are global. First, excess risk is measured at the level of labeled query families rather than individual queries, because every query in V1 appears only once. Second, operational frontiers are reported on two axes: a token axis for single-cell comparisons and a nominal-budget axis for cross-subset analyses. Both choices are forced by the evaluation setup and are justified below.

\subsection{Partition granularity and the family-level estimator}

The theorem object of central interest is the familywise excess Bayes risk
\begin{equation}
\Delta_{\mathcal Q}(Z;X):=\sup_{q\in\mathcal Q}\big(\mathcal R_Z^{\star}(q)-\mathcal R_X^{\star}(q)\big),
\end{equation}
whose supremum runs over individual queries. Direct empirical estimation of this quantity would require per-query risk estimates $\widehat{\mathcal R}^{\star}_Z(q)$ and $\widehat{\mathcal R}^{\star}_X(q)$. That is impossible in V1 because every query appears exactly once in the evaluation data. The finest partition at which per-partition risks are estimable is therefore the labeled query family.

Let $\mathcal F(\mathcal Q)$ denote the partition of the dataset's queries into labeled families, and let $h$ denote the frozen answerer used for the paired raw/compressed evaluation. For family $F\in\mathcal F(\mathcal Q)$ with $N_F$ samples, the family-level mean excess risk at budget $b$ is
\begin{equation}
\begin{aligned}
\dQF{b}
=
\max_{F\in\mathcal F(\mathcal Q)}
\Bigg[
&\frac{1}{N_F}\sum_{(x,q,y):\,\mathrm{family}(q)=F}
\ell\big(h(Z_b(x),q),y\big)\\
&-
\frac{1}{N_F}\sum_{(x,q,y):\,\mathrm{family}(q)=F}
\ell\big(h(X,q),y\big)
\Bigg].
\end{aligned}
\end{equation}
The corresponding dataset-mean estimator is
\begin{equation}
\davg{b}
=
\frac{1}{N}\sum_{(x,q,y)\in\mathrm{eval}}
\Big[\ell\big(h(Z_b(x),q),y\big)-\ell\big(h(X,q),y\big)\Big].
\end{equation}
At the population level, replacing individual queries by a labeled-family partition gives the ordering proved in Theorem~\ref{thm:partition-refinement}:
\begin{equation}
\Delta^{\rm avg}_{\mu}(Z;X)
\le
\Delta^{\rm fam}_{\mathcal F(\mathcal Q)}(Z;X)
\le
\Delta_{\mathcal Q}(Z;X).
\end{equation}
Empirically, $\dQF{b}-\davg{b}$ is the observed hidden-damage margin for the declared partition. On single-family datasets, the two estimators coincide by construction.

\subsection{Operational frontier estimator}

For compression method $m$, tolerance $\varepsilon$, and excess-risk variant $\widehat\Delta\in\{\davg{b},\dQF{b}\}$, the nominal-axis frontier is
\begin{equation}
\Rfb{\varepsilon}{\mathcal Q; m,\widehat\Delta}
:=
\min\{b\in\mathcal B:\widehat\Delta(b)\le \varepsilon\},
\end{equation}
where $\mathcal B$ is the evaluated budget grid. The token-axis frontier is
\begin{equation}
\Rftok{\varepsilon}{\mathcal Q; m,\widehat\Delta}
:=
\tauhat{\Rfb{\varepsilon}{\mathcal Q; m,\widehat\Delta}},
\end{equation}
where $\tauhat{b}$ is the empirical mean number of audio tokens passed to the backbone at nominal budget $b$.

We interpret $\Rftok{\varepsilon}{\mathcal Q}$ as an operational proxy for the rate-theoretic frontier rather than as a calibrated estimator of $R^{\star}_{\mathrm{Bayes}}(\varepsilon,\mathcal Q)$: our interface is a hard-selected audio-token subset without an entropy-efficient coding layer, so the Shannon-slack bound from Proposition~\ref{app:prop:rate_token} is not expected to be tight.

\subsection{Axis convention}

Single-cell analyses use the token axis, because token count corresponds directly to what the backbone processes. Cross-subset analyses use the nominal-budget axis, because the token axis is subset-dependent: restricting to a family subset changes the per-budget distribution of clip lengths and therefore changes $\tauhat{b}$. In particular, the cumulative-chain analysis of \S\ref{sec:nesting-results}, the factor-overlap analysis of \S\ref{sec:factor-overlap-results}, and the conditioned-gain comparison of \S\ref{sec:gcond-real-audio} are all reported on the nominal axis.

\subsection{Worst-family-constrained frontier}

For a query-family subset $\mathcal Q_k\subseteq\mathcal Q$, the worst-family-constrained nominal frontier is
\begin{equation}
\begin{aligned}
\Rfworst{\varepsilon}{\mathcal Q_k}
:=
\min\Bigg\{b\in\mathcal B:
\max_{F\in\mathcal F(\mathcal Q_k)}
\Bigg[
&\frac{1}{N_{F,k}}\sum_{\mathrm{family}(q)=F}
\ell\big(h(Z_b(x),q),y\big)\\
&-
\frac{1}{N_{F,k}}\sum_{\mathrm{family}(q)=F}
\ell\big(h(X,q),y\big)
\Bigg]
\le \varepsilon\Bigg\}.
\end{aligned}
\end{equation}
By construction this frontier is monotone non-decreasing in the subset argument: if $\mathcal Q_1\subseteq\mathcal Q_2$, then $\Rfworst{\varepsilon}{\mathcal Q_1}\le \Rfworst{\varepsilon}{\mathcal Q_2}$. Along a cumulative chain, the increment
\begin{equation}
\Delta \Rfworst{\varepsilon}{\mathcal Q_k}
=
\Rfworst{\varepsilon}{\mathcal Q_k}-\Rfworst{\varepsilon}{\mathcal Q_{k-1}}
\end{equation}
measures the additional budget required when the new family is added to the requirement set.

\subsection{Operational conditioned gain}

The real-audio operational counterpart of Theorem~\ref{thm:condadv} is defined for either the average estimator or the familywise estimator:
\begin{equation}
\begin{aligned}
\widehat G^{\rm op}_{\rm cond}(\varepsilon,\mathcal Q;\widehat\Delta)
:= {} &
\Rfb{\varepsilon}{\mathcal Q;\texttt{learned\_agnostic},\widehat\Delta}\\
&-
\Rfb{\varepsilon}{\mathcal Q;\texttt{learned\_conditioned},\widehat\Delta},\\
&\widehat\Delta\in\{\davg{\cdot},\dQF{\cdot}\}.
\end{aligned}
\end{equation}
A positive value means that the conditioned selector requires less nominal budget to achieve the same tolerance and estimator than the agnostic selector. The familywise version \(\widehat\Delta=\dQF{\cdot}\) is the deployment sign-off quantity; the average version is retained for V1 historical comparisons. The synthetic exact quantity
\begin{equation}
\Gcondinfo{\varepsilon}{\mathcal Q}=R^\star_{\mathrm{Bayes}}(\varepsilon,\mathcal Q)-R^\star_{\mathrm{Bayes,cond}}(\varepsilon,\mathcal Q)
\end{equation}
is reserved for the finite-alphabet constructions of \S\ref{sec:synthetic-validations}.

The operational conditioned-gain estimator $\Gcondop{\varepsilon}{\mathcal Q}$ defined above compares two selector methods evaluated on the same downstream backbone; we refer to this as the \emph{same-backbone} version. V2 further reports a \emph{cross-backbone} variant
\begin{equation}
\begin{aligned}
\widehat G^{\mathrm{op,cross}}_{\mathrm{cond}}(\varepsilon, \mathcal Q; B)
:=
&\;\Rfb{\varepsilon}{\mathcal Q;\texttt{learned\_agnostic}\text{ on }\texttt{qwen2audio}}\\
&-\Rfb{\varepsilon}{\mathcal Q;\texttt{learned\_conditioned}\text{ on }B}.
\end{aligned}
\end{equation}
which fixes the reference method to Qwen2-Audio's \texttt{learned\_agnostic} selector. The same-backbone version asks ``does conditioning help, holding everything else constant.'' The cross-backbone version asks ``how much does a backbone plus its conditioned selector lose relative to a shared reference backbone with its agnostic selector.'' \Cref{tab:v2_gcond_qwen2audio,tab:v2_gcond_qwen25omni} in \S\ref{sec:v2-gcond-three-seed} report the same-backbone version; \Cref{tab:mmsu_temporal} and \Cref{tab:v21_alpha_best} report the cross-backbone version. The two versions differ only for Qwen2.5-Omni because the Qwen2-Audio cross-backbone case is degenerate (A and B use the same backbone). The numerical separation can be large: for MMSU on Qwen2.5-Omni at $\varepsilon=0.05$, the same-backbone three-seed mean on the keyword \texttt{temporal} family is $+0.0057$, whereas the cross-backbone three-seed mean is $-0.5738$ (see \S\ref{sec:per-family-subanalyses}).

\subsection{Bootstrap protocol}

All CIs in V1 are computed with non-parametric bootstrap resampling at the sample-ID level with $n_{\mathrm{boot}}=10{,}000$, confidence level $0.95$, and seed 42 \cite{efron1993bootstrap}. For single-estimator quantities such as $\davg{b}$, $\dQF{b}$, and $\Rfb{\varepsilon}{\mathcal Q}$, each bootstrap iteration resamples sample IDs with replacement and recomputes the statistic on the resampled set. The 95\% interval is the percentile interval.

For $\Gcondop{\varepsilon}{\mathcal Q}$ we use a paired bootstrap: each resampling draws a multiset of sample IDs once, then computes both the agnostic and conditioned frontiers on that same resampled set before taking the difference. The paired structure preserves sample-level correlation between the two methods. Two-sided $p$-values are computed from the bootstrap replicate distribution as
\begin{equation}
 p = 2\min\big(\Pr[\widehat G^{*}_{\mathrm{cond}}\le 0],\,\Pr[\widehat G^{*}_{\mathrm{cond}}\ge 0]\big).
\end{equation}
The factor-overlap diagnostic uses a paired ratio bootstrap, and the cumulative-chain analysis uses a shared-resampling bootstrap across all chain steps.

The bootstrap endpoints below describe the historical analyses; formal confirmation uses the independent protocol in Section~\ref{sec:estimating-main}.  Let $U_{95}(\widehat\theta)$ and $L_{95}(\widehat\theta)$ be the upper and lower percentile endpoints.  In these descriptive analyses, interval positions relative to a tolerance summarize a fixed cell; they do not authorize budget selection or post-selection certification.  For conditioned gain, the sign is declared positive only when the lower endpoint is above zero, negative only when the upper endpoint is below zero, and inconclusive otherwise.  These intervals quantify sampling uncertainty; cross-seed standard deviations quantify selector/training stochasticity and are reported separately.  If several evaluation rows are derived from the same audio source, the same formulas apply with bootstrap resampling at the audio-source ID level.

\subsection{Reference conventions}

The canonical raw-audio reference for family $F$ is
\begin{equation}
\widehat R_X(F)=\frac{1}{N_F}\sum_{(x,q,y):\,\mathrm{family}(q)=F}\ell\big(h(X,q),y\big).
\end{equation}
When the canonical uncompressed parquet is missing for qwen2.5-Omni (AudioMCQ and MMSU), we use the $b=1.00$ endpoint of the same method as a \texttt{self\_full} fallback. On qwen2audio cells where both references are available, the difference is bounded by $10^{-4}$.

\subsection{Loss-type conventions}

The primary loss is the 0/1 multiple-choice loss
\begin{equation}
\ell_{0/1}(q,\hat y,y)=\mathbf 1\{\hat y\neq y\},
\end{equation}
which matches the benchmark decision rule and underlies all main-text tables. We also compute the correct-answer negative log-likelihood,
\begin{equation}
\ell_{\mathrm{NLL}}(q,\hat y,y)=-\log P_{\mathrm{backbone}}(y\mid Z,q),
\end{equation}
and use it as a secondary measurement in supporting analyses. The 0/1 and NLL variants yield the same qualitative conclusions throughout V1.

\subsection{Estimator catalogue}

For reference, the empirical estimators used in \S\ref{sec:results} are: $\davg{b}$ (dataset-mean excess loss), $\dQF{b}$ (family-level excess loss), $\Rftok{\varepsilon}{\mathcal Q}$ (token-axis operational frontier), $\Rfb{\varepsilon}{\mathcal Q}$ (nominal-axis operational frontier), $\Rfworst{\varepsilon}{\mathcal Q_k}$ (worst-family-constrained frontier along cumulative chains), $\Gcondop{\varepsilon}{\mathcal Q}$ (operational conditioned gain), $\Gcondinfo{\varepsilon}{\mathcal Q}$ (exact synthetic conditioned gain), and $\tauhat{b}$ (mean audio-token count at nominal budget $b$).

\section{Dataset and Query-Family Details}\label{app:datasets}\label{sec:datasets}

We evaluate TAAC on five audio-question datasets that together span the principal capabilities a large audio language model is expected to exercise: sound-event understanding on environmental audio, multi-capability multiple-choice reasoning on general audio, speech-centric multi-task understanding, music-and-speech-aware reasoning, and logical reasoning under audio-textual prompts. The datasets are chosen to satisfy two conditions simultaneously: each must expose a nontrivial query family to which Theorem~\ref{thm:equiv}'s family-level estimator can be applied, and the union must include at least one single-family dataset against which family-level effects can be tested for artifacts. Three of the five datasets (AudioMCQ-StrongAC, MMSU, and MMAR) carry multiple labeled query families under the keyword classifier of \S\ref{sec:taxonomy} and support the full theorem-level analysis; two (DCASE 2026 dev and BigBench Audio) carry a single labeled family under that classifier and serve as single-family controls at the keyword level. Both DCASE and BigBench Audio are in fact multi-family under their intrinsic dataset-native task metadata (\S\ref{sec:family-gap-results} \Cref{tab:partition-gap-comparison}); the single-family status is therefore a property of the keyword classifier, not of the datasets themselves.

\begin{table}[t]
\centering
\small
\caption{Dataset summary for the five audio-QA benchmarks used in V1 evaluation. The \textbf{Kw fam.} column is the number of distinct \texttt{query\_family} labels assigned by the shared keyword-heuristic classifier of \S\ref{sec:taxonomy}; the \textbf{Native} column is the number of families under each dataset's intrinsic task-label metadata (see \S\ref{sec:datasets} and \Cref{tab:partition-gap-comparison}). The Acc.\ column is qwen2audio uncompressed accuracy at budget fraction $1.00$, read from the per-cell aggregate JSONs. Partial qwen2.5-Omni coverage reflects the DCASE flash-attention CPU-dispatch failure and the missing canonical uncompressed parquets on AudioMCQ and MMSU in V1.}
\label{fig:dataset_summary}
\begin{adjustbox}{max width=0.95\linewidth}
\begin{tabular}{lrrrrp{4.0cm}llc}
\toprule
Dataset & $n$ & Kw fam. & Native & Acc. & Role in paper & qwen2audio & qwen2.5-Omni & E1? \\
\midrule
DCASE 2026 dev & 1,607 & 1 & 6 & 0.4443 & Primary operational target; single-family conditioned-gain test under the keyword partition & Full suite & Heuristics only; DCASE pipeline failed & No \\
AudioMCQ-StrongAC & 19,480 & 4 & 10 & 0.7424 & Primary multi-family headline dataset and selector-training source & Full suite & Heuristics only; no canonical uncompressed reference & Yes \\
MMSU & 5,000 & 6 & 47 & 0.5530 & Most taxonomically rich multi-family benchmark & Full suite & Heuristics only; no canonical uncompressed reference & Yes \\
MMAR & 1,000 & 3 & 9 & 0.4960 & Supplementary multi-family reasoning benchmark & Full suite & Heuristics plus uncompressed reference & Yes \\
BigBench Audio & 1,000 & 1 & 4 & 0.8770 & Text-dominated single-family control under the keyword partition (four BBH task types under native partition) & Full suite & Heuristics plus uncompressed reference & No \\
\bottomrule
\end{tabular}
\end{adjustbox}
\end{table}
\begin{table}[t]
\centering
\small
\caption{Duplicate source-copy dataset summary retained from \texttt{figures/fig7\_dataset\_summary.tex}. Dataset summary for the five audio-QA benchmarks used in V1 evaluation. Family count is the number of distinct \texttt{query\_family} labels assigned by the shared keyword-heuristic classifier. The Acc.\ column is qwen2audio uncompressed accuracy at budget fraction $1.00$, read from the per-cell aggregate JSONs. Partial qwen2.5-Omni coverage reflects the DCASE flash-attention CPU-dispatch failure and the missing canonical uncompressed parquets on AudioMCQ and MMSU in V1.}
\label{tab:dataset-summary-duplicate-source}
\begin{adjustbox}{max width=0.95\linewidth}
\begin{tabular}{lrrrp{4.3cm}llc}
\toprule
Dataset & $n$ & Families & Acc. & Role in paper & qwen2audio & qwen2.5-Omni & E1? \\
\midrule
DCASE 2026 dev & 1,607 & 1 & 0.4443 & Primary operational target; single-family conditioned-gain test & Full suite & Heuristics only; DCASE pipeline failed & No \\
AudioMCQ-StrongAC & 19,480 & 4 & 0.7424 & Primary multi-family headline dataset and selector-training source & Full suite & Heuristics only; no canonical uncompressed reference & Yes \\
MMSU & 5,000 & 6 & 0.5530 & Most taxonomically rich multi-family benchmark & Full suite & Heuristics only; no canonical uncompressed reference & Yes \\
MMAR & 1,000 & 3 & 0.4960 & Supplementary multi-family reasoning benchmark & Full suite & Heuristics plus uncompressed reference & Yes \\
BigBench Audio & 1,000 & 1 & 0.8770 & Text-dominated validity control & Full suite & Heuristics plus uncompressed reference & No \\
\bottomrule
\end{tabular}
\end{adjustbox}
\end{table}

\subsection{DCASE 2026 dev}

DCASE 2026 dev is the development subset of DCASE 2026 Challenge Task~5, Audio-Dependent Question Answering (ADQA) \cite{dcase2026adqa}. Each sample consists of an audio clip and a short natural-language multiple-choice question whose answer must depend on the audio. We use the development split because the official evaluation set is held out by the challenge.

\paragraph{Role in the paper.}
DCASE 2026 is TAAC's primary \emph{operational} target: the dataset is dominated by non-speech audio content, so learned selectors cannot succeed by defaulting to speech-content heuristics, and the ground-truth signal is sensitive to which audio regions are retained under compression. It is also the dataset on which the operational conditioned gain reaches its largest point estimate, $\Gcondop{0.05}{\Qfam}=+0.0993$ budget fraction with paired-bootstrap $p=0.057$; see \S\ref{sec:gcond-real-audio}.

\paragraph{Query family structure.}
Under the keyword classifier of \S\ref{sec:taxonomy}, DCASE 2026's category strings do not match any of the predefined dispatch groups, so the classifier assigns \texttt{query\_family=general} to every sample. Under this partition, DCASE functions as a single-family evaluation: $\dQF{b}$ coincides by construction with $\davg{b}$, and the keyword-level nested-monotonicity analysis collapses to the trivial chain $\{\texttt{general}\}$. However, DCASE does expose a native content taxonomy: the challenge's own \texttt{content\_taxonomy.csv} groups samples into six post-merge clusters (e.g., \texttt{meaning/stress/intonation}, \texttt{speaker/clip/express}, and four smaller clusters; see \S\ref{sec:per-family-subanalyses} and the \texttt{other\_rare} caveat noted there). Under this native partition, the family-level excess-risk gap is $+5.04$ pp at $b=0.20$ with Qwen2-Audio learned-conditioned (\Cref{tab:partition-gap-comparison}), rather than the zero imposed by the keyword classifier. DCASE is therefore a single-family control only with respect to the keyword classifier, not with respect to the dataset's intrinsic structure.

\paragraph{Sample count and audio characteristics.}
The DCASE 2026 dev split contributes $n=1{,}607$ audio-query pairs. Audio durations are variable, with most clips between 5 and 30 seconds; all clips are resampled to 16~kHz mono before chunking. Qwen2-Audio evaluations are complete for all compression methods and budgets. The qwen2.5-Omni DCASE pipeline failed under a flash-attention CPU-dispatch error and is dropped from V1.

\subsection{AudioMCQ-StrongAC}

AudioMCQ-StrongAC is a four-option multiple-choice audio-question dataset derived from AudioMCQ, with the ``StrongAC'' qualifier denoting the strongly controlled subset in which the correct answer is supported by specific audio content rather than by inferable background knowledge \cite{he2025audiomcq}. Each sample consists of a natural audio clip, a natural-language query, and four candidate answers, exactly one of which is correct.

\paragraph{Role in the paper.}
AudioMCQ-StrongAC is the primary dataset for V1's empirical theorem testing. Its native question-type taxonomy (four families, 19,480 evaluation samples) provides the largest multi-family evaluation set in V1 and the one in which the family-level excess-risk gap is most pronounced. The gap $\dQF{b}-\davg{b}$ peaks at $+6.79$ percentage points at $b=0.20$, with a near-peak value of $+6.29$ percentage points at $b=0.40$; this is the headline empirical claim of the paper. AudioMCQ is also the dataset on which the selector networks were trained.

\paragraph{Query family structure.}
AudioMCQ-StrongAC exposes a \texttt{question\_type} field that our shared classifier maps onto four canonical labels: \texttt{general}, \texttt{music}, \texttt{speech\_content}, and \texttt{temporal}. The semantic meaning of each family is as follows.
\begin{itemize}
    \item \textbf{\texttt{speech\_content}} ($n=10{,}505$, 53.9\%): queries about linguistic or lexical content of speech-what was said, transcription, phrase-level understanding, or the language being spoken.
    \item \textbf{\texttt{general}} ($n=6{,}085$, 31.2\%): fallback bucket for samples whose native category string matches none of the specific keyword groups. On AudioMCQ this family is the bottleneck family under the nested-monotonicity analysis: adding \texttt{general} to the cumulative chain yields the largest increase in $\Rfworst{0.05}{\mathcal Q_k}$, from 0.470 to 0.697 ($+0.227$).
    \item \textbf{\texttt{music}} ($n=1{,}724$, 8.9\%): queries about musical content such as instruments, genre, melody, rhythm, or song-level attributes.
    \item \textbf{\texttt{temporal}} ($n=1{,}165$, 6.0\%): queries about time, order, counting, duration, or timing relationships between events in the audio.
\end{itemize}

\paragraph{Sample count and audio characteristics.}
The AudioMCQ-StrongAC training split contains 19,480 samples with precomputed per-chunk LOO-NLL oracle relevance targets; the evaluation split used here also contains 19,480 samples, with the same per-family counts listed above. Clips are variable-duration (typically 3-30~s), sampled at 16~kHz mono, padded to one chunk when shorter than 1~s and truncated to 120 chunks when longer than 120~s.

\subsection{MMSU}

MMSU is a multi-task benchmark for spoken-language understanding and reasoning, designed to evaluate fine-grained speech perception and complex reasoning in natural speech \cite{wang2025mmsu}. Each sample is a four-option multiple-choice question paired with an audio clip. The benchmark spans linguistic content, paralinguistic inference, speaker attributes, and speech-conditioned scene and event reasoning.

\paragraph{Role in the paper.}
MMSU is the most taxonomically rich dataset we evaluate, exposing six distinct keyword families and 47 native \texttt{task\_name} values. Under the keyword partition it is used to show that the family-level excess-risk gap is not an AudioMCQ-specific artifact and to provide the richest keyword-level cumulative-chain analysis in \S\ref{sec:nesting-results}. Under the 47-task native partition, MMSU is also the dataset on which the keyword partition most severely under-resolves the intrinsic structure - a factor of roughly 8$\times$ - and correspondingly the dataset on which the native-partition family-level gap is largest in absolute terms ($+29.17$ pp at $b=0.20$; \Cref{tab:partition-gap-comparison}). The per-task analysis in \S\ref{sec:per-family-subanalyses} identifies \texttt{intonation\_perception} as the single most-impacted native task on Qwen2.5-Omni under the cross-backbone conditioned-gain estimator (\Cref{fig:mmsu_native_gcond}).

\paragraph{Query family structure.}
Our classifier operates on the concatenation of MMSU's native task name and category string. Six canonical labels appear:
\begin{itemize}
    \item \textbf{\texttt{general}} ($n=2{,}886$, 57.7\%): the largest fallback bucket.
    \item \textbf{\texttt{paralinguistic}} ($n=803$, 16.1\%): emotion, speaker identity, accent, gender, prosody, and pitch.
    \item \textbf{\texttt{speech\_content}} ($n=648$, 13.0\%): linguistic content of speech.
    \item \textbf{\texttt{sound\_event}} ($n=454$, 9.1\%): discrete acoustic events and sources. Despite the name, this keyword family is not semantically pure: approximately 115 of 454 samples route to a user-intent-classification semantic cluster under e5-large-v2 embedding analysis (supporting file: \texttt{partitions/semantic/mmsu\_semantic\_partition.csv}; see \S\ref{sec:discussion-semantic}).
    \item \textbf{\texttt{sound\_scene}} ($n=110$, 2.2\%): ambient scene and environment.
    \item \textbf{\texttt{temporal}} ($n=99$, 2.0\%): temporal structure of the audio.
\end{itemize}
The two smallest families, \texttt{sound\_scene} and \texttt{temporal}, have sample counts below 200 and therefore produce wide bootstrap CIs in family-wise frontier estimates.

While the keyword partition groups MMSU queries into six families, the MMSU native metadata carries a \texttt{task\_name} field with 47 distinct values in the release we evaluate (e.g., \texttt{word\_identification}, \texttt{age\_prediction}, \texttt{intonation\_perception}, \texttt{pitch\_comparison}, \texttt{disfluency\_detection}, \texttt{volume\_comparison}, \texttt{couplet\_matching}, \texttt{dialogue\_turn\_counting}, among others). The keyword classifier's 6-family partition is therefore approximately 8$\times$ under-resolution of the dataset's intrinsic task structure. \S\ref{sec:family-gap-results} \Cref{tab:partition-gap-comparison} shows that this under-resolution materially suppresses the measured family-level gap: under the 47-task native partition, the MMSU worst-family gap at $b=0.20$ is $+29.17$ pp, compared with the $+1.56$ pp measured under the keyword partition on the same Qwen2-Audio learned-conditioned cells (supporting CSV: \texttt{tables/partition\_delta\_q\_summary.csv}, rows with \texttt{dataset='mmsu'}, \texttt{backbone='qwen2audio'}, \texttt{method='learned\_conditioned'}).

\paragraph{Sample count and audio characteristics.}
MMSU contributes $n=5{,}000$ evaluation samples. The audio is speech-dominated, variable-duration, and resampled to 16~kHz mono. The speech dominance is load-bearing for one observation in \S\ref{sec:nesting-results}: speech-related families remain relatively well served under aggressive compression because speech dominates the V1 selector's mel-spectrogram feature energy.

\subsection{MMAR}

MMAR is a benchmark for deep reasoning in speech, audio, music, and their mixtures \cite{ma2025mmar}. We use MMAR's four-option multiple-choice audio-question pairs, treating the textual query as the standard natural-language question and evaluating only the audio-conditioned answer.

\paragraph{Role in the paper.}
MMAR is a supplementary multi-family dataset that demonstrates the generality of the family-level gap and the nested-monotonicity behavior on a third independently constructed benchmark. It is also the dataset on which the only two numerical monotonicity violations occur, both at the smallest subset and the tightest tolerance, with point-estimate drops of only 0.001-0.003 budget fraction.

\paragraph{Query family structure.}
Our classifier operates on the concatenation of MMAR's modality and category fields. Three labels appear:
\begin{itemize}
    \item \textbf{\texttt{speech\_content}} ($n=618$, 61.8\%): linguistic content of speech.
    \item \textbf{\texttt{music}} ($n=217$, 21.7\%): music-oriented reasoning.
    \item \textbf{\texttt{general}} ($n=165$, 16.5\%): fallback bucket.
\end{itemize}

\paragraph{Sample count and audio characteristics.}
MMAR contributes $n=1{,}000$ evaluation samples. Audio is variable-duration, 16~kHz mono after resampling, with a mixture of speech, music, and mixed-content clips broadly consistent with the three-family taxonomy above.

\subsection{BigBench Audio}

BigBench Audio is an audio version of a subset of BIG-Bench Hard questions \cite{artificialanalysis2024bigbenchaudio,suzgun2022bbh}. Each item pairs an audio rendering of a reasoning prompt with a multiple-choice answer. In practice, many questions can be answered largely from the textual prompt alone, with the audio providing supporting rather than load-bearing information.

\paragraph{Role in the paper.}
BigBench Audio is a text-dominated dataset where single-method per-backbone operational effects of compression are small: under the same-backbone conditioned-gain estimator (\S\ref{sec:metrics}), both Qwen2-Audio and Qwen2.5-Omni show near-null behavior on every axis. However, the dataset is not a universal null control. Under the cross-backbone estimator (\S\ref{sec:v2-gcond-three-seed}), Qwen2.5-Omni's \texttt{formal\_fallacies} task exhibits a large negative conditioned gain (three-seed mean $-0.8527$, std 0.025; see \S\ref{sec:per-family-subanalyses} and \Cref{tab:bigbench_cross_backbone}), which we discuss further in \S\ref{sec:discussion-small-gcond}.

\paragraph{Query family structure.}
BigBench Audio's native metadata carries a \texttt{task\_name} field identifying four BIG-Bench Hard task types: \texttt{formal\_fallacies}, \texttt{navigate}, \texttt{object\_counting}, and \texttt{web\_of\_lies}, with 250 samples each. The keyword classifier of \S\ref{sec:taxonomy} does not match any of these to its predefined dispatch, so all samples are routed to \texttt{general} under the keyword partition. Under a native-task partition evaluated at $b=0.20$ with Qwen2-Audio learned-conditioned and 0/1 loss, three of four families exhibit per-sample mean excess of $-0.123$ (compression improves predictions), while \texttt{formal\_fallacies} exhibits $+0.409$ (compression harms predictions); see \Cref{tab:bigbench_native}. The near-null aggregate in \Cref{tab:gcond-other} is thus a cancellation of four opposing effects rather than four neutral effects. Correspondingly, \S\ref{sec:family-gap-results} reports a family-level excess-risk gap of $+39.90$ percentage points at $b=0.20$ under the native partition, substantially larger than the keyword-aggregate of $+0$. BigBench Audio is therefore a ``single-family control'' only with respect to the keyword classifier, not with respect to the dataset's intrinsic structure.

\paragraph{Sample count and audio characteristics.}
BigBench Audio contributes $n=1{,}000$ evaluation samples. Short, low-information audio clips are common, and the dataset's reasoning burden is dominated by the textual prompt.

\begin{table}[t]
    \caption{BigBench Audio per-family per-sample mean excess at $b=0.20$ with Qwen2-Audio \texttt{learned\_conditioned} and 0/1 loss. The four native families partition BigBench Audio's 1000 samples into equal-sized 250-sample BIG-Bench Hard task types. Three families exhibit negative mean excess (compression improves predictions); one exhibits strongly positive mean excess. The dataset-aggregate near-null reported in \Cref{tab:gcond-other} is the weighted mean of these four opposing effects. Supporting CSV: \texttt{metrics/partition\_delta\_q\_long.csv} rows with \texttt{dataset='bigbench\_audio'}, \texttt{partition='native\_fine'}, \texttt{backbone='qwen2audio'}, \texttt{method='learned\_conditioned'}, \texttt{budget=0.20}.}
    \centering
    \small
    \begin{tabular}{lrr}
        \toprule
        Native task & $n$ & per-sample mean excess \\
        \midrule
        \texttt{formal\_fallacies}  & 250 & $+0.409$ \\
        \texttt{navigate}           & 250 & $-0.123$ \\
        \texttt{object\_counting}   & 250 & $-0.123$ \\
        \texttt{web\_of\_lies}      & 250 & $-0.123$ \\
        \midrule
        dataset aggregate           & 1000 & $+0.010$ \\
        \bottomrule
    \end{tabular}
    \label{tab:bigbench_native}
\end{table}

\subsection{Query family taxonomy}
\label{sec:taxonomy}

All query-family labels in this paper are assigned by a single classifier function \texttt{infer\_family(raw\_cat)} defined in \texttt{src/taac/datasets/mcq\_parser.py}. The function lowercases its input and dispatches in a fixed order through seven keyword groups; the first group that matches wins, and inputs matching no group fall through to \texttt{general}. The dispatch order is:
\begin{enumerate}[leftmargin=1.5em]
    \item \textbf{\texttt{speech\_content}}: \{speech, asr, transcript, spoken, language, word, phrase\}
    \item \textbf{\texttt{paralinguistic}}: \{emotion, speaker, accent, gender, paralinguis, prosody, pitch\}
    \item \textbf{\texttt{sound\_event}}: \{event, detection, sound\_event, source\}
    \item \textbf{\texttt{sound\_scene}}: \{scene, environment, ambient, location\}
    \item \textbf{\texttt{music}}: \{music, instrument, genre, melody, rhythm, song\}
    \item \textbf{\texttt{temporal}}: \{temporal, count, timing, order, when, duration, time\}
    \item \textbf{\texttt{general}}: fallback bucket.
\end{enumerate}
Because the classifier is shared, each canonical label has the same definition on every dataset; what differs is the source field on which the classifier operates.

The taxonomy is a keyword-heuristic partition of each dataset's native category strings, not a semantic partition of query content. A query whose category string contains the word ``when'' is routed to \texttt{temporal} regardless of whether the question actually requires temporal reasoning, and a query whose category string contains ``speech'' is routed to \texttt{speech\_content} regardless of the specific speech property at issue. Partition monotonicity applies only when one partition actually refines another; the native and semantic alternatives need not refine the keyword partition. The reported gaps and frontier-overlap ratios are therefore partition-specific, not lower bounds on unidentified semantic quantities. The observed partition sensitivity is illustrated in \S\ref{sec:family-gap-results}: under dataset-native task partitions (e.g., MMSU's 47 \texttt{task\_name} values, BigBench's 4 BBH task types), measured family-level gaps are 1.2$\times$ to 20$\times$ larger than under the keyword partition at every tested budget (\Cref{tab:partition-gap-comparison}).

The family-label robustness of the family-level findings has been quantified directly. Under 20\% random label flips with 100 Monte Carlo replicates, the bottleneck-family identity under the semantic partition (\S\ref{sec:discussion-semantic}) is preserved in 84\%-100\% of replicates on MMSU, MMAR, and AudioMCQ-StrongAC, compared with 39\%-57\% under the keyword partition (supporting CSV: \texttt{metrics/label\_noise\_sensitivity\_long.csv}). Label noise therefore affects the observed family structure most under the keyword partition and least under the semantic partition, indicating that finer semantic partitioning both increases the measured gap and reduces its sensitivity to label noise. An empirical instance of the keyword classifier's conservatism: on AudioMCQ-StrongAC, a single semantic cluster of queries about ``sequence'' or ``timing'' of sounds ($n=1619$ per e5-large-v2 $+$ $k$-means clustering of query text) splits across keyword families as $\{$\texttt{general}: 500, \texttt{music}: 54, \texttt{speech\_content}: 64, \texttt{temporal}: 1020$\}$, so approximately 500 semantically-temporal samples are routed to \texttt{general} under the keyword dispatch.

\section{Experimental Protocols}\label{app:protocols}\label{sec:setup}

This section specifies the backbones, selector architectures, training recipe, evaluation protocol, and inference-time compressor configuration used in V1. Under the no-compromise policy, we surface three methodology caveats explicitly: truncated selector training, the large apparent parameter asymmetry between agnostic and conditioned selectors being dominated by the query embedding table, and the V1 selector's use of 768-dimensional mel-spectrogram features rather than backbone audio-tower features. These caveats interact only with the operational conditioned-gain result in \S\ref{sec:gcond-real-audio}; the family-level gap, nested-monotonicity, and factor-overlap findings are about theorem-aligned empirical quantities and are invariant to the particular selector instantiation.

\subsection{Backbones}

V1 uses two open-source audio-language model backbones.

\paragraph{Qwen2-Audio-7B-Instruct.}
Qwen2-Audio-7B-Instruct is the primary backbone for V1 \cite{chu2024qwen2audio}. The model combines an audio compressor with a Qwen2 language decoder and supports direct textual responses conditioned on audio. All compression methods-\texttt{uncompressed}, \texttt{uniform}, \texttt{random}, \texttt{energy\_vad}, \texttt{learned\_agnostic}, and \texttt{learned\_conditioned}-are evaluated against qwen2audio on all five datasets. Decoder generation is capped at \texttt{max\_new\_tokens=5} because every benchmark is four-option multiple choice.

\paragraph{Qwen2.5-Omni-7B.}
Qwen2.5-Omni-7B is the secondary backbone \cite{xu2025qwen25omni}. It provides a cross-backbone validity check on MMAR, MMSU, AudioMCQ, and BigBench Audio, but the DCASE pipeline failed due to a flash-attention CPU-dispatch error that returned \texttt{predicted\_label=ERROR} on all processed rows. V1 therefore treats qwen25omni only as a validity-control backbone; the learned-selector analyses are qwen2audio-only.

\paragraph{Architecture gap.}
Proposition~\ref{app:prop:modelgap} defines a theoretical model-class gap $\Gamma_{\mathcal F}(Z;\mathcal Q)$. The reported cross-answerer/interface summaries are operational diagnostics rather than estimates of that gap; see \S\ref{sec:discussion-gamma}.

\subsection{Selector compression methods}\label{sec:compression-methods}

V1 evaluates two learned selectors against three heuristic baselines. All five methods operate on the same interface: segment the audio into 1-second chunks, assign a scalar relevance score to each chunk (or a deterministic selection rule for the heuristics), and retain the top-$k$ chunks where $k=\max(1,\lfloor bN\rfloor)$ for budget fraction $b$ and total chunk count $N$.

\begin{itemize}
    \item \textbf{\texttt{uncompressed}} passes the full audio to the backbone and defines the canonical raw-audio reference.
    \item \textbf{\texttt{uniform}} retains evenly spaced chunks at fixed density.
    \item \textbf{\texttt{random}} selects chunks uniformly at random without replacement, using a fixed per-sample seed.
    \item \textbf{\texttt{energy\_vad}} ranks chunks by RMS energy and keeps the top-$k$.
    \item \textbf{\texttt{learned\_agnostic}} and \textbf{\texttt{learned\_conditioned}} are small MLP selectors trained on precomputed LOO-NLL oracle relevance targets. The only architectural difference is whether the scoring network sees the query embedding at inference time.
\end{itemize}

\subsection{V1 selector architecture and training}
\label{sec:selector-arch}

Both learned selectors are implemented as instances of the same \texttt{LearnedSelector} class and consume precomputed per-chunk features of dimensionality 768. Those 768-dimensional features are pooled log-mel-spectrogram representations rather than the 1280-dimensional embeddings produced by qwen2audio's own audio tower. This detail is load-bearing: V1 selectors operate on a weaker acoustic representation than the backbone itself uses at inference, so V1 tests a conservative version of the answer-preserving compression problem.

% \begin{figure*}[t]
%     \centering
%     \includegraphics[width=0.94\textwidth]{figures/fig6_selector_architecture.pdf}
%     \caption{V1 selector architecture. The agnostic selector and conditioned selector share a chunk compressor mapping 768-dimensional log-mel-spectrogram features to 128-dimensional chunk embeddings. The conditioned variant adds a query compressor whose parameter count is dominated by the $151{,}644\times 128$ token embedding table; 98.0\% of the conditioned selector's parameters live in that table. The non-embedding scoring core is 396,014 parameters, only $1.5\times$ the agnostic selector's size. V1 uses mel-spectrogram chunk features rather than the 1280-dimensional audio-tower embeddings qwen2audio produces; V2 retraining tests the consequences of this choice.}
%     \label{fig:selector_arch}
% \end{figure*}

\paragraph{Chunk compressor and scoring heads.}
The chunk compressor is a two-layer MLP,
\begin{equation}
\mathrm{Linear}(768,256)\to\mathrm{GELU}\to\mathrm{LayerNorm}(256)\to\mathrm{Linear}(256,128),
\end{equation}
with 229,760 parameters excluding the LayerNorm affine terms. The agnostic scoring head applies
\begin{equation}
\mathrm{Linear}(128,256)\to\mathrm{GELU}\to\mathrm{LayerNorm}(256)\to\mathrm{Linear}(256,1),
\end{equation}
plus a learned scalar \texttt{pos\_scale} multiplying a normalized temporal-position bias. The conditioned scoring head uses the same structure after concatenating the 128-dimensional chunk embedding with a 128-dimensional query embedding.

\paragraph{Query compressor and parameter asymmetry.}
The conditioned selector adds a query compressor comprising a token embedding \texttt{Embedding(151644,128)}, a positional embedding \texttt{Embedding(256,128)}, and a two-layer MLP applied to the mask-aware mean-pooled token-plus-position embedding. Direct checkpoint inspection gives the following counts: \texttt{learned\_agnostic} has 263,066 parameters across 13 tensors; \texttt{learned\_conditioned} has 19,806,446 parameters across 21 tensors, of which 19,410,432 live in the query token embedding table. The apparent 75$\times$ parameter asymmetry is therefore almost entirely an embedding-table effect; the non-embedding ``scoring core'' is only $1.51\times$ larger for the conditioned selector.

\paragraph{Training data and target.}
Both selectors are trained exclusively on the 19,480-sample AudioMCQ-StrongAC training split. The training manifest stores a per-sample \texttt{chunk\_relevance} vector whose provenance field is \texttt{chunk\_relevance\_source=loo\_nll} in every row. For each chunk, the target relevance is the increase in qwen2audio's NLL on the correct answer when that chunk is removed from the audio. At training time the raw relevance vector is normalized to a probability distribution over chunks and used as the target for the selector's soft selection distribution.

\paragraph{Loss and optimization.}
The training loss is
\begin{equation}
\mathcal L
=
\mathrm{KL}(\mathrm{relevance}\parallel \mathrm{selection})
+0.1\,\mathrm{MSE}(\mathrm{actual\_count},\mathrm{target\_count})
-0.01\,H(\mathrm{selection}),
\end{equation}
where the entropy term prevents early collapse. The selectors use AdamW with learning rate $10^{-4}$, weight decay $0.01$, batch size 8, gradient accumulation 4, cosine decay with warmup ratio 0.05, and gradient clipping at norm 1.0. The differentiable top-$k$ path uses a Gumbel-softmax perturbation of the chunk scores \cite{herrmann2020channel}; inference uses exact hard top-$k$.

\paragraph{V1 training truncation.}
The nominal V1 schedule is ten epochs, but seed-42 training terminated early under patience-3 stopping on a Gumbel-noisy validation signal. The agnostic selector's final checkpoint was saved at step 2,184 (35.9\% of nominal schedule), and the conditioned selector's at step 1,638 (26.9\%). The conditioned selector is therefore more undertrained than the agnostic selector, which biases V1 against finding a conditioned-compression advantage. V2 corrects this with patience 8 and a longer schedule.

\subsection{Evaluation protocol}

V1 evaluation is run as a per-backbone, per-dataset, per-method, per-budget sweep that produces one parquet file per cell. Each row stores \texttt{sample\_id}, \texttt{query\_family}, the 0/1 multiple-choice loss, the NLL on the correct answer, the realized \texttt{num\_audio\_tokens}, and the predicted answer label. Every downstream analysis in \S\ref{sec:metrics} and \S\ref{sec:results} is computed from these parquets.

The canonical raw-audio reference is the mean backbone loss on the uncompressed audio, computed per family. This reference is available for qwen2audio on all five datasets and for qwen2.5-Omni on MMAR and BigBench Audio. On the two qwen25omni cells where the canonical uncompressed parquet is absent (AudioMCQ and MMSU), we use the method's own $b=1.00$ endpoint as a \texttt{self\_full} fallback. Empirically, the absolute difference between the canonical reference and the fallback is bounded by $10^{-4}$ on qwen2audio cells where both are available.

\subsection{Chunking and top-\texorpdfstring{$k$}{k} mapping}

Audio is segmented into 1-second non-overlapping chunks at 16~kHz. Although the stored configuration block lists \texttt{hop\_duration\_s=0.5}, direct inspection of the training manifest confirms non-overlapping chunking: a 10-second sample contains exactly 10 chunks, and evaluation parquets match the non-overlapping rule $N=\lceil D\rceil$ after padding short clips and truncating clips longer than 120 seconds.

For budget fraction $b\in[0,1]$ and audio with $N$ chunks, the interface retains
\begin{equation}
k=\max(1,\lfloor bN\rfloor)
\end{equation}
chunks. The floor-of-one is important for very short clips: at $b=0.05$ on a 10-chunk clip, the selector still retains one chunk; at $b=1.00$, all chunks are retained and the selection becomes a no-op.

\subsection{Budget and tolerance grids}

The main-body budget grid is
\begin{equation}
b\in\{0.05,0.10,0.20,0.40,1.00\},
\end{equation}
spanning two orders of magnitude of compression while reserving $b=1.00$ as the uncompressed endpoint. An expanded appendix grid
\begin{equation}
b\in\{0.01,0.025,0.05,0.10,0.20,0.40,0.60,0.80,1.00\}
\end{equation}
provides additional resolution at both extremes. Operational frontiers are reported at tolerances
\begin{equation}
\varepsilon\in\{0.01,0.02,0.05\}.
\end{equation}
The $\varepsilon=0.05$ setting is the primary one in the main text because it corresponds to roughly 5 percentage points of excess 0/1 loss.

\subsection{Summary of V1 methodology caveats}

We close \S\ref{sec:setup} by consolidating the three load-bearing V1 caveats.
\begin{enumerate}[leftmargin=1.5em]
    \item \textbf{Truncated selector training.} Both selectors terminate well short of the nominal schedule, and the conditioned selector terminates earlier than the agnostic selector.
    \item \textbf{Parameter asymmetry dominated by the query embedding.} The conditioned selector's large parameter count is almost entirely the token embedding table; its scoring core is only modestly larger than the agnostic selector's.
    \item \textbf{Mel-spectrogram input features.} V1 selectors operate on 768-dimensional mel features rather than the backbone's native 1280-dimensional audio-tower embeddings.
\end{enumerate}
None of these caveats affect the structural findings on the family-level excess-risk gap, nested monotonicity, or factor overlap; they interact with the operational conditioned-gain estimate in \S\ref{sec:gcond-real-audio}.

\section{Additional Results and Ablations}\label{app:additional-results}\label{sec:results}

This section presents the empirical instantiation of the theory developed in \S\ref{sec:theory-main}. We begin with synthetic tasks where the theorem predictions are verifiable at machine precision, then move to the headline real-audio result on the family-level excess-risk gap, then to the nested-monotonicity and factor-overlap diagnostics. We next report the V1 single-seed conditioned-gain result and then extend it with V2 three-seed, cross-backbone, and training-recipe ablations. The interpretive chain \S\ref{sec:synthetic-validations} $\rightarrow$ \S\ref{sec:factor-overlap-results} $\rightarrow$ \S\ref{sec:gcond-real-audio} $\rightarrow$ \S\ref{sec:v2-gcond-three-seed} is the tightest theorem-experiment bridge in the paper.

\subsection{Synthetic validations of finite, Gaussian, and conditioned-advantage results}
\label{sec:synthetic-validations}

We first validate three theoretical constructions on synthetic tasks where the Bayes frontier is analytically computable.

\paragraph{Finite-alphabet frontier.}
On finite-query synthetic tasks with controlled alphabet sizes and loss structure, the exact Bayes frontier obtained by solving the convex program of Proposition~\ref{app:prop:finite} matches the closed-form optimum to machine precision.

\paragraph{Gaussian latent allocation.}
On Gaussian factor tasks with controlled variances $\sigma_j^2$ and query weights $\alpha_{qj}$, numerical optimization reproduces the weighted reverse-water-filling solution predicted by Proposition~\ref{app:prop:gaussian}, including the water levels and per-coordinate distortions.

\paragraph{Conditioned-compression separation.}
For the strict-separation construction with $X=(V_1,V_2,W)$, $Q\in\{q_1,q_2\}$, $P(Q=q_1)=\lambda$, and $Y_{q_i}=V_i$ under zero-one loss, the theorem predicts
\begin{equation}
R^{\star}_{\mathrm{Bayes}}(0,\{q_1,q_2\})=H(V_1,V_2),\qquad
R^{\star}_{\mathrm{Bayes,cond}}(0,\{q_1,q_2\})=\lambda H(V_1)+(1-\lambda)H(V_2),
\end{equation}
and therefore
\begin{equation}
G_{\mathrm{cond}}(0)= (1-\lambda)H(V_1)+\lambda H(V_2).
\end{equation}
We verify this prediction by exact enumeration across six $(k_1,k_2)$ configurations and eleven values of $\lambda$, for a total of 66 cells. Every enumerated cell matches the analytic form to better than $10^{-9}$ bits.

\begin{figure*}[t]
    \centering
    \includegraphics[width=\textwidth]{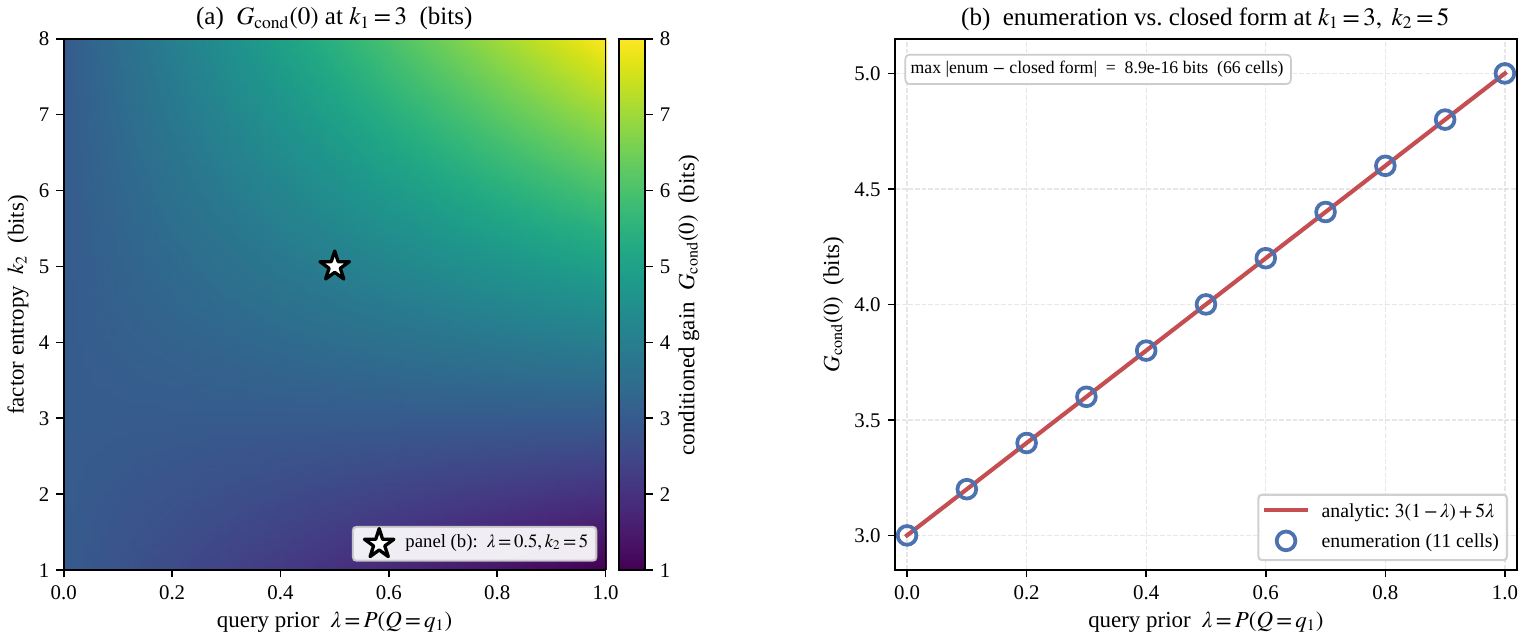}
    \caption{Synthetic verification of Theorem~\ref{thm:condadv}'s strict-separation construction. \textbf{Left:} analytic conditioned gain $G_{\mathrm{cond}}(0)=(1-\lambda)k_1+\lambda k_2$ bits across the query prior $\lambda$ and the second factor's entropy $k_2$, at fixed $k_1=3$. \textbf{Right:} exact enumeration points overlaid on the closed-form line $3(1-\lambda)+5\lambda$ at $(k_1,k_2)=(3,5)$. Across 66 test cells the enumerated gains match the closed form to within $10^{-9}$ bits.}
    \label{fig:gcond_synthetic}
\end{figure*}

The synthetic verification is important for the remainder of the paper: it establishes that the conditioned gain is large in exactly the factor-disjoint regime described by Theorem~\ref{thm:condadv}, and therefore motivates reading the real-audio operational results through the measured factor structure of the datasets.

\subsection{Family-level excess risk gap}
\label{sec:family-gap-results}

The headline empirical finding of V1 is that the family-level estimator $\dQF{b}$ is consistently and substantially larger than the dataset-mean estimator $\davg{b}$ on every multi-family dataset and every selector family. The gap is therefore a property of the theorem quantity on real data, not a property of any specific compressor.

\begin{figure*}[t]
    \centering
    \includegraphics[width=\textwidth]{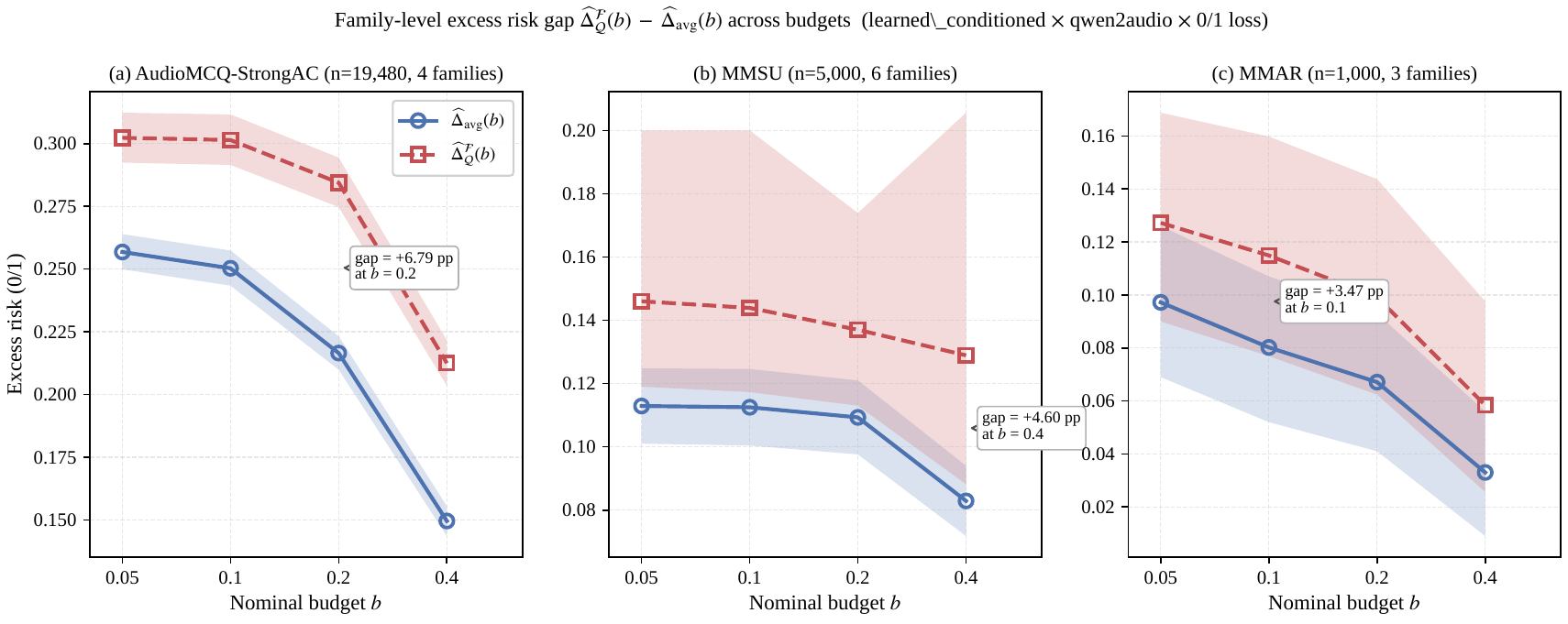}
    \caption{Empirical family-level excess-risk gap across the three multi-family datasets. Each panel shows the dataset-mean excess risk $\davg{b}$ (blue, solid) and the family-level worst-family excess risk $\dQF{b}$ (red, dashed) as a function of nominal budget $b$, with 95\% bootstrap confidence ribbons. The gap is positive at every budget, reaching $+6.79$ percentage points on AudioMCQ-StrongAC at $b=0.20$. Learned conditioned selector, qwen2audio backbone, and 0/1 loss. On single-family datasets (DCASE and BigBench Audio, not shown), $\davg{b}\equiv\dQF{b}$ by construction.}
    \label{fig:delta_q_gap}
\end{figure*}

\paragraph{AudioMCQ-StrongAC.}
Under \texttt{learned\_conditioned} on qwen2audio with the canonical uncompressed reference and 0/1 loss, the family-level gap peaks at $b=0.20$ and remains above six points at $b=0.40$.

\begin{table}[t]
    \caption{Excess-risk summary on AudioMCQ-StrongAC: dataset-mean $\davg{b}$ versus worst-family $\dQF{b}$ at four budgets.}
    \centering
    \small
    \begin{tabular}{cccc}
        \toprule
        Budget $b$ & $\davg{b}$ & $\dQF{b}$ & gap (pp) \\
        \midrule
        0.05 & 0.2568 [0.2499, 0.2639] & 0.3023 [0.2924, 0.3123] & +4.55 \\
        0.10 & 0.2503 [0.2434, 0.2574] & 0.3014 [0.2915, 0.3116] & +5.11 \\
        0.20 & 0.2165 [0.2099, 0.2232] & 0.2844 [0.2747, 0.2944] & +6.79 \\
        0.40 & 0.1496 [0.1436, 0.1555] & 0.2125 [0.2035, 0.2216] & +6.29 \\
        \bottomrule
    \end{tabular}
    \label{tab:gap-audiomcq}
\end{table}

\paragraph{MMSU.}
MMSU exhibits the same phenomenon, albeit with wider family-level CIs because the \texttt{temporal} and \texttt{sound\_scene} families are small.

\begin{table}[t]
    \caption{Family-level excess-risk gap on MMSU under the same protocol as \Cref{tab:gap-audiomcq}.}
    \centering
    \small
    \begin{tabular}{cccc}
        \toprule
        Budget $b$ & $\davg{b}$ & $\dQF{b}$ & gap (pp) \\
        \midrule
        0.05 & 0.1129 [0.1010, 0.1248] & 0.1460 [0.1190, 0.2000] & +3.31 \\
        0.10 & 0.1125 [0.1004, 0.1246] & 0.1439 [0.1173, 0.2000] & +3.14 \\
        0.20 & 0.1093 [0.0976, 0.1210] & 0.1370 [0.1130, 0.1739] & +2.77 \\
        0.40 & 0.0829 [0.0718, 0.0940] & 0.1289 [0.0882, 0.2056] & +4.59 \\
        \bottomrule
    \end{tabular}
    \label{tab:gap-mmsu}
\end{table}

\paragraph{MMAR.}
MMAR is the smallest of the multi-family datasets and therefore yields the noisiest family-level max, but the gap remains positive at every tested budget.

\begin{table}[t]
    \caption{Family-level excess-risk gap on MMAR under the same protocol as \Cref{tab:gap-audiomcq}.}
    \centering
    \small
    \begin{tabular}{cccc}
        \toprule
        Budget $b$ & $\davg{b}$ & $\dQF{b}$ & gap (pp) \\
        \midrule
        0.05 & 0.0972 [0.0690, 0.1260] & 0.1272 [0.0900, 0.1688] & +3.01 \\
        0.10 & 0.0802 [0.0520, 0.1070] & 0.1149 [0.0769, 0.1598] & +3.47 \\
        0.20 & 0.0671 [0.0410, 0.0930] & 0.0991 [0.0624, 0.1437] & +3.19 \\
        0.40 & 0.0330 [0.0090, 0.0560] & 0.0584 [0.0256, 0.0976] & +2.54 \\
        \bottomrule
    \end{tabular}
    \label{tab:gap-mmar}
\end{table}

The keyword-partition analyses in \Cref{tab:gap-audiomcq}, \Cref{tab:gap-mmsu}, and \Cref{tab:gap-mmar} are conservative in a measurable sense. Under dataset-native task partitions (AudioMCQ-StrongAC from source-dataset $\times$ question-type metadata; MMSU from 47 \texttt{task\_name} values; MMAR from category $\times$ subcategory; BigBench Audio from its 4 BBH task types; DCASE from its post-merge 6-family content taxonomy), the family-level gaps grow substantially. Under semantic partitions using e5-large-v2 sentence embeddings and cosine $k$-means clustering on query text (see \S\ref{sec:discussion-semantic}), gaps grow further still:

\begin{table}[t]
    \caption{Family-level excess-risk gap $\dQF{b} - \davg{b}$ under three partition granularities at $b=0.20$ with \texttt{learned\_conditioned} on Qwen2-Audio and 0/1 loss. The keyword column reports the point-estimate gap computed from the same per-family operational frontiers used in \Cref{tab:gap-audiomcq,tab:gap-mmsu,tab:gap-mmar}; small differences from the percentile-CI gaps reported in those tables are within the $n_{\mathrm{boot}}=10{,}000$ Monte Carlo noise band. The native column uses each dataset's intrinsic task labels. The semantic column uses e5-large-v2 sentence embeddings and cosine $k$-means clustering on query text. The keyword-versus-native/semantic differential persists across the full budget grid $b\in\{0.05,0.10,0.20,0.40\}$ with consistent monotonic ordering; the full per-budget matrix is in \texttt{tables/partition\_delta\_q\_summary.csv}.}
    \centering
    \small
    \begin{tabular}{lrrr}
        \toprule
        Dataset & keyword (pp) & native (pp) & semantic (pp) \\
        \midrule
        AudioMCQ-StrongAC & $+6.79$ & $+7.96$  & $+14.29$ \\
        MMSU              & $+1.56$ & $+29.17$ & $+22.31$ \\
        MMAR              & $+1.88$ & $+10.18$ & $+13.34$ \\
        BigBench Audio    & $+0.00$ & $+39.90$ & n/a \\
        DCASE 2026 dev    & $+0.00$ & $+5.04$  & $+7.68$ \\
        \bottomrule
    \end{tabular}
    \label{tab:partition-gap-comparison}
\end{table}

The paper's keyword partition under-measures the family-level gap by a factor of roughly 1.2$\times$ (AudioMCQ) to 20$\times$ (MMSU), and entirely misses the multi-family structure of BigBench and DCASE. The \S\ref{sec:family-gap-results} headline of $+6.79$ pp on AudioMCQ-StrongAC should therefore be read as a lower bound; the true worst-family gap under honest partitioning is larger on every dataset. BigBench Audio's semantic cell is marked ``n/a'' because e5-large-v2 clustering on BBH query text collapses to a single cluster at $k=12$ (BBH queries share strong semantic structure at the cluster level); its native 4-family partition is the correct fine-grained scope.

Across the three multi-family datasets in the keyword partition, the family-level excess-risk gap is consistently positive and ranges from $+2.54$ to $+6.79$ percentage points. On DCASE and BigBench Audio the keyword gap is exactly zero because the keyword classifier assigns a single \texttt{general} family; under native-task partitions, however, both datasets are multi-family and show substantial family-level gaps: $+5.04$ pp on DCASE (6 post-merge families) and $+39.90$ pp on BigBench Audio (4 BBH task families), both at $b=0.20$ with Qwen2-Audio learned conditioned. \Cref{tab:partition-gap-comparison} below consolidates the keyword-versus-native/semantic comparison.

\subsection{Rate-risk and length-risk summaries}
\label{sec:frontier-summary}

The full per-cell rate-risk and length-risk tables are large and are best generated directly from the released CSVs. The main operational pattern is clear on the nominal axis. At $\varepsilon=0.05$, DCASE is the most compression-friendly dataset: both learned selectors reach tolerance at budget fractions around 0.5, while MMAR reaches tolerance in the 0.2-0.3 range. AudioMCQ and MMSU are the most compression-hostile datasets, with frontiers near 0.8 at $\varepsilon=0.05$. Tightening the tolerance to $\varepsilon=0.02$ moves DCASE to roughly 0.78-0.82, AudioMCQ to roughly 0.92, MMSU to roughly 0.85, and MMAR to roughly 0.63-0.69. At $\varepsilon=0.01$, all datasets except MMAR require budget fractions near 0.9 or above. These frontiers are architecture-specific operational baselines rather than bit-level rate measurements.

\subsection{Nested-family monotonicity and bottleneck identification}
\label{sec:nesting-results}

Theorem~\ref{app:thm:monotonicity} predicts that adding query families can only increase the required budget. We test this empirically with cumulative chains ordered by increasing family sample count and evaluate the worst-family-constrained frontier $\Rfworst{\varepsilon}{\mathcal Q_k}$ at $\varepsilon=0.05$.

\begin{figure*}[t]
    \centering
    \includegraphics[width=\textwidth]{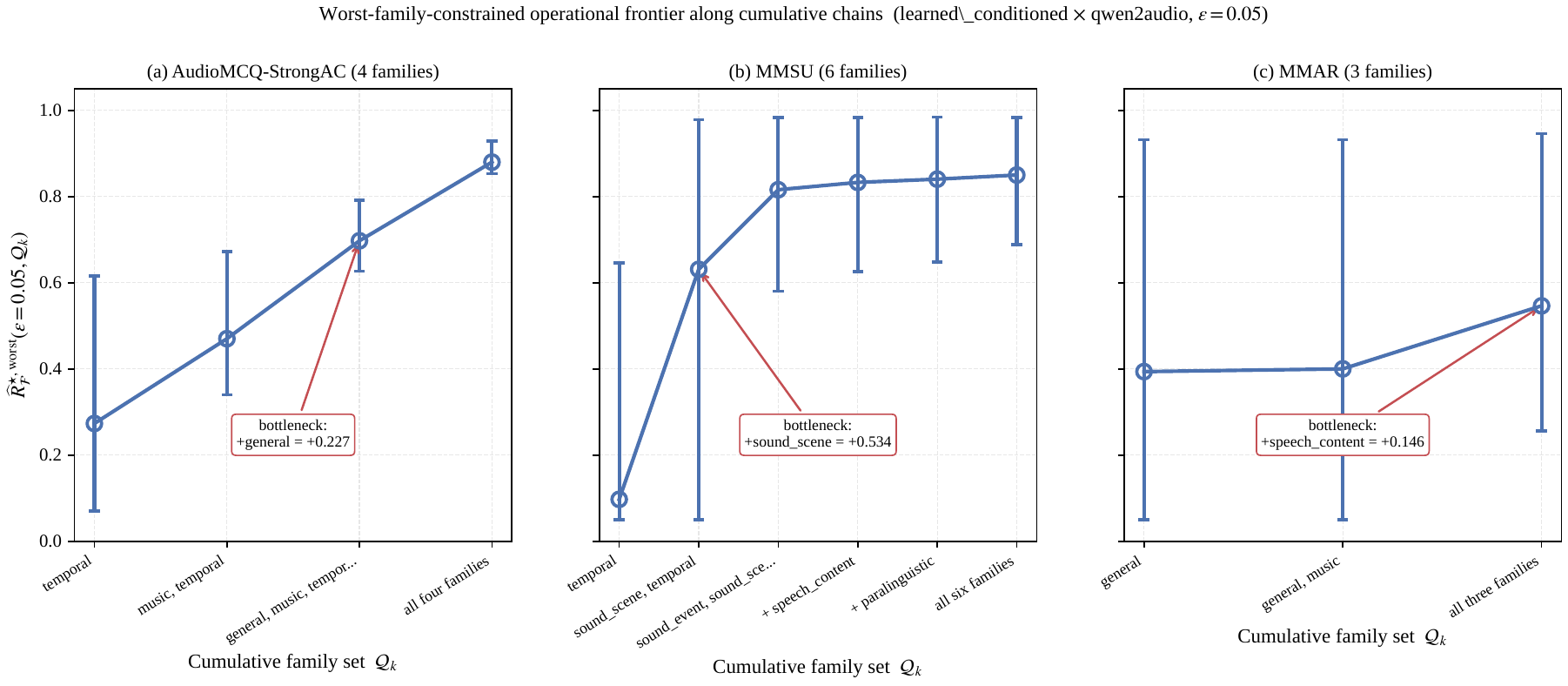}
    \caption{Worst-family-constrained operational frontier along cumulative chains under the \textbf{keyword partition}. The chain is monotone non-decreasing on all three datasets. On AudioMCQ-StrongAC, adding \texttt{general} produces the largest increment ($+0.227$) under this partition. On MMSU the mechanically largest increment occurs at the small-family end of the chain and should be read as evidence that the small-family frontier dominates the later steps. On MMAR, adding \texttt{speech\_content} produces the dominant increment ($+0.146$). Bottleneck identities are themselves partition-dependent: under the 47-task native partition on Qwen2.5-Omni, the cross-backbone conditioned-gain bottleneck of MMSU is \texttt{intonation\_perception}, not \texttt{temporal} (\S\ref{sec:per-family-subanalyses}, \Cref{fig:mmsu_native_gcond}).}
    \label{fig:nesting_chains}
\end{figure*}

\begin{table*}[t]
    \caption{Cumulative-chain frontiers $\Rfworst{0.05}{\mathcal Q_k}$ under the \textbf{keyword partition} for the learned conditioned selector on qwen2audio with 0/1 loss. Native-partition and semantic-partition cumulative-excess structure is summarised in \Cref{fig:worst_k_concentration}; the worst-$k$ concentration in the keyword partition above is, on MMSU in particular, a granularity artifact.}
    \centering
    \small
    \begin{adjustbox}{max width=\textwidth}
    \begin{tabular}{lllccl}
        \toprule
        Dataset & Step & Family added & $n_k$ & Cumulative families & $\Rfworst{0.05}{\mathcal Q_k}$ \\
        \midrule
        AudioMCQ & 1 & - & 1,166 & \{temporal\} & 0.2733 [0.0702, 0.6150] \\
        AudioMCQ & 2 & music & 2,890 & \{music, temporal\} & 0.4701 [0.3397, 0.6715] \\
        AudioMCQ & 3 & general & 8,975 & \{general, music, temporal\} & 0.6973 [0.6266, 0.7912] \\
        AudioMCQ & 4 & speech\_content & 19,480 & all four families & 0.8799 [0.8529, 0.9290] \\
        \midrule
        MMSU & 1 & - & 99 & \{temporal\} & 0.0973 [0.0500, 0.6460] \\
        MMSU & 2 & sound\_scene & 209 & \{sound\_scene, temporal\} & 0.6311 [0.0500, 0.9785] \\
        MMSU & 3 & sound\_event & 663 & \{sound\_event, sound\_scene, temporal\} & 0.8158 [0.5803, 0.9836] \\
        MMSU & 4 & speech\_content & 1,311 & + speech\_content & 0.8328 [0.6258, 0.9834] \\
        MMSU & 5 & paralinguistic & 2,114 & + paralinguistic & 0.8403 [0.6482, 0.9841] \\
        MMSU & 6 & general & 5,000 & all six families & 0.8499 [0.6889, 0.9835] \\
        \midrule
        MMAR & 1 & - & 165 & \{general\} & 0.3937 [0.0500, 0.9318] \\
        MMAR & 2 & music & 382 & \{general, music\} & 0.4000 [0.0500, 0.9318] \\
        MMAR & 3 & speech\_content & 1,000 & all three families & 0.5464 [0.2562, 0.9458] \\
        \bottomrule
    \end{tabular}
    \end{adjustbox}
    \label{tab:nesting}
\end{table*}

Under the keyword partition, the AudioMCQ chain is strictly monotone, with increments $+0.1969$, $+0.2272$, and $+0.1826$; adding \texttt{general} is therefore the dominant keyword-level step. MMSU is also monotone, but the first two steps have very wide CIs because they are driven by the \texttt{temporal} and \texttt{sound\_scene} families ($n=99$ and $n=110$). The scientifically relevant MMSU observation at this granularity is that once the small-family end of the chain is admitted, subsequent additions cost very little. MMAR is monotone at $\varepsilon=0.05$ and identifies \texttt{speech\_content} as the keyword-level bottleneck family. Both of these bottleneck identifications are properties of the keyword partition and not of the data: a native-task analysis re-localises the MMSU failure to \texttt{intonation\_perception} (\S\ref{sec:per-family-subanalyses}), and the worst-$k$ concentration picture of \Cref{fig:worst_k_concentration} shows that under the 47-task native partition MMSU's top-2 tasks account for only 16\% of the total excess mass (vs 76\% at keyword granularity).

Across all 276 entries of the cumulative-chain CSV (all datasets, methods, tolerances, and chain steps), we observe exactly two numerical monotonicity violations, both on MMAR at $\varepsilon=0.02$ and both smaller than 0.003 budget fraction. Given bootstrap intervals of width roughly 0.9 on those cells, these are grid-discreteness artifacts rather than failures of Theorem~\ref{app:thm:monotonicity}.

Under finer partitions, the worst-$k$ concentration differs dramatically from the keyword-partition result above. At $b=0.20$ with Qwen2-Audio \texttt{learned\_conditioned}, the worst-2 families' cumulative excess fraction is 94\% (AudioMCQ keyword, 4 families) versus 26\% (AudioMCQ semantic, 12 families); 76\% (MMSU keyword, 6) versus \textbf{16\%} (MMSU native, 47); and 100\% (MMAR keyword, 3) versus 33\% (MMAR native, 9) or 58\% (MMAR semantic, 6) - see \Cref{fig:worst_k_concentration}. The MMSU native number is particularly striking: 2 of 47 native tasks account for only 16\% of the total dataset-level excess mass, so compression harm is \emph{distributed} across dozens of fine-grained tasks rather than concentrated in a small pair (supporting CSV: \texttt{metrics/partition\_nesting\_long.csv} with filter \texttt{backbone='qwen2audio'}, \texttt{method='learned\_conditioned'}, \texttt{budget=0.20}). The keyword partition's apparent concentration is therefore substantially a partition-granularity artifact rather than a property of the data.

\begin{figure*}[t]
    \centering
    \includegraphics[width=0.92\textwidth]{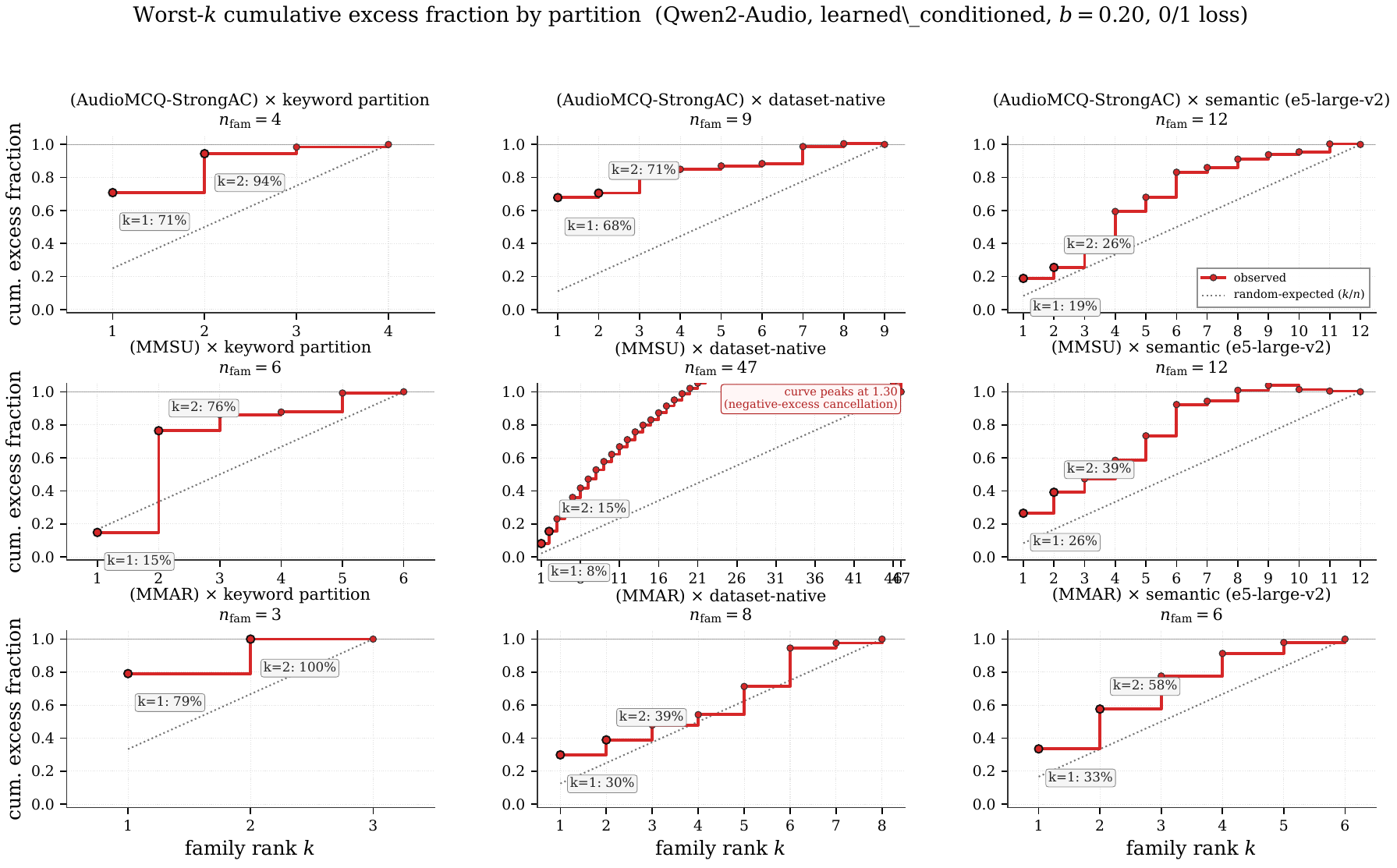}
    \caption{Worst-$k$ cumulative excess fraction by partition, at $b=0.20$ with Qwen2-Audio and the \texttt{learned\_conditioned} selector. Each panel shows the fraction of total dataset excess captured by the top-$k$ worst families (ranked by per-family mean excess) as $k$ grows from 1 to $n_{\mathrm{fam}}$. The red curve is observed; the dotted diagonal is the random-expected concentration $k/n_{\mathrm{fam}}$. Under the keyword partition, MMSU's worst-2 families carry 76\% of the total dataset excess (curve far above the diagonal), but under the 47-task dataset-native partition the worst-2 carry only 16\% (curve essentially on the diagonal). The apparent bottleneck concentration in \Cref{tab:nesting} is therefore partly a granularity artifact of the keyword partition. The semantic partition (e5-large-v2 + $k$-means on query text) produces an intermediate result on all three datasets. DCASE and BigBench Audio are omitted because their keyword partitions are single-family by construction and their native partitions introduce negative-excess cancellations that make the concentration quantity ill-defined.}
    \label{fig:worst_k_concentration}
\end{figure*}

\subsection{Factor-overlap diagnostic}
\label{sec:factor-overlap-results}

Corollary~\ref{app:cor:factor} predicts additive frontier decomposition when query subfamilies act on disjoint latent factor blocks. We test that antecedent through the additivity ratio
\begin{equation}
\mathrm{ratio}(\mathcal Q_a,\mathcal Q_b)
=
\frac{\widehat R^{\star}_{\mathcal F,b}(\varepsilon,\mathcal Q_a\cup \mathcal Q_b)}{\widehat R^{\star}_{\mathcal F,b}(\varepsilon,\mathcal Q_a)+\widehat R^{\star}_{\mathcal F,b}(\varepsilon,\mathcal Q_b)}.
\end{equation}
A ratio near 1.0 would indicate factor-disjoint subfamilies; a ratio near 0.5 indicates frontier co-location and therefore strong factor overlap. The additivity diagnostic below is evaluated over \textbf{keyword-family pairs} only: the V2 verification pack retains the summary-level V1 keyword-pair additivity table rather than the full per-pair scatter, so \Cref{fig:factor_overlap,tab:additivity_summary} report the per-dataset minimum, median, mean, and maximum over keyword-family pairs. A finer-granularity reevaluation over the native-task pairs of MMSU (C(47,2) = 1{,}081 pairs) and BigBench Audio (C(4,2) = 6 pairs) is left to future work; the present summary is therefore itself a lower bound on factor-overlap heterogeneity, paralleling the keyword-vs-native story for the family-level gap (\Cref{tab:partition-gap-comparison}).

\begin{figure*}[t]
    \centering
    \includegraphics[width=0.92\textwidth]{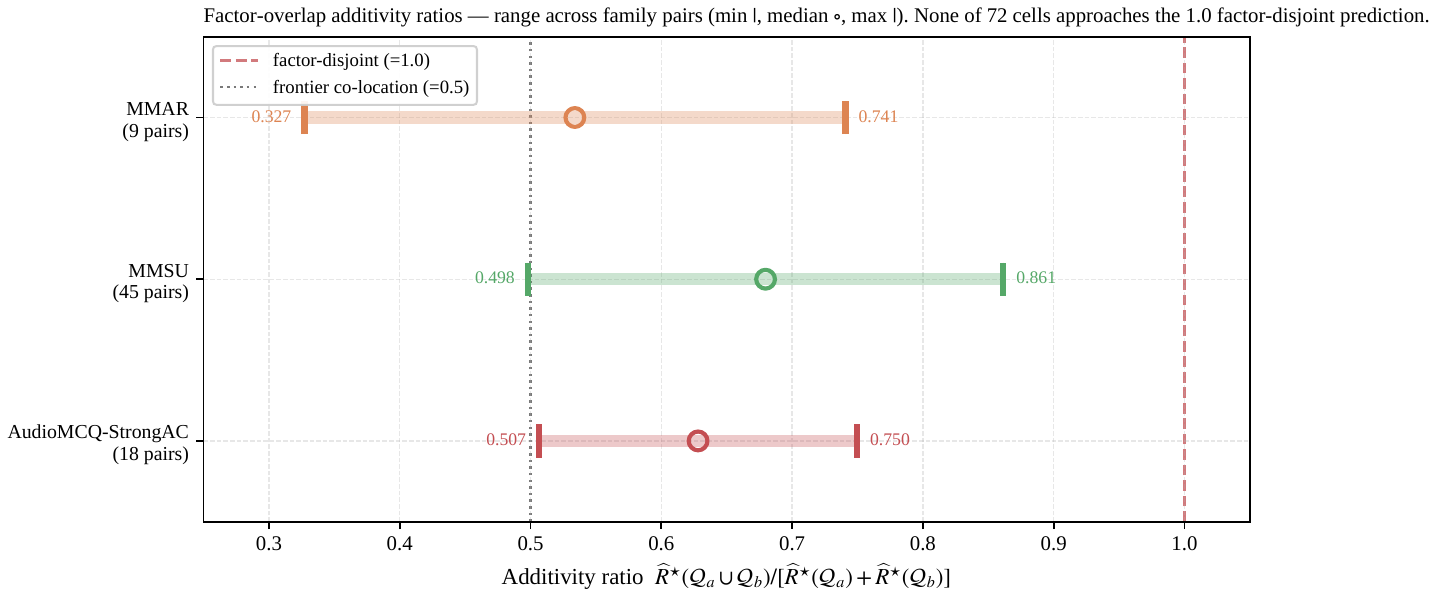}
    \caption{Summary-level additivity ratios for factor-overlap diagnostics. For each dataset, the plot shows the range and median of $\widehat R^{\star}(\mathcal Q_a\cup\mathcal Q_b)/[\widehat R^{\star}(\mathcal Q_a)+\widehat R^{\star}(\mathcal Q_b)]$ over the cells retained in the V2 verification pack. The dashed vertical line at 1.0 is the factor-disjoint prediction of Corollary~\ref{app:cor:factor}; the dotted line at 0.5 indicates frontier co-location. The global maximum is 0.8614, so no summary approaches the factor-disjoint regime.}
    \label{fig:factor_overlap}
\end{figure*}

\begin{table}[t]
\centering\small
\caption{Factor-overlap additivity summary.  Per-dataset range of additivity ratios $\widehat{R}^{\star}(\mathcal{Q}_a \cup \mathcal{Q}_b) / [\widehat{R}^{\star}(\mathcal{Q}_a) + \widehat{R}^{\star}(\mathcal{Q}_b)]$.  Combined over $\varepsilon \in \{0.01, 0.02, 0.05\}$.  Transcribed from V1 PDF Table 5.}
\label{tab:additivity_summary}
\begin{tabular}{lrrrrr}
\toprule
Dataset & $n$ cells & min & median & mean & max \\
\midrule
AudioMCQ-StrongAC & 54 & 0.5066 & 0.6281 & 0.6281 & 0.7496 \\
MMSU & 135 & 0.4980 & 0.6797 & 0.6797 & 0.8614 \\
MMAR & 27 & 0.3271 & 0.5339 & 0.5339 & 0.7408 \\
\bottomrule
\end{tabular}
\end{table}

The observed summary range is $[0.3271,0.8614]$ across the three multi-family datasets under the keyword partition. AudioMCQ-StrongAC ranges from 0.5066 to 0.7496, MMSU from 0.4980 to 0.8614, and MMAR from 0.3271 to 0.7408. Thus the empirical antecedent of Corollary~\ref{app:cor:factor} is not satisfied by the natural audio taxonomies we test at keyword granularity: the measured frontiers are far closer to co-location than to additivity. We therefore interpret the small V1 keyword-level conditioned gains on multi-family datasets through this overlap measurement rather than as evidence against Theorem~\ref{thm:condadv}. The measurement is partition-dependent: native-task pairs on MMSU have markedly more between-task variance in the cross-backbone conditioned-gain estimator (\Cref{fig:mmsu_native_gcond}, \Cref{tab:mmsu_native_worst}), consistent with lower factor overlap at finer granularity, but we do not report per-pair additivity at the native level in this paper.

\subsection{V1 operational conditioned-compression gain}
\label{sec:gcond-real-audio}

V1 first tested Theorem~\ref{thm:condadv}'s operational implication directly on real audio using the paired nominal-axis quantity $\Gcondop{\varepsilon}{\mathcal Q}$. All CIs and $p$-values come from 10,000 paired bootstrap replicates at the sample-ID level.

\begin{figure*}[t]
    \centering
    \includegraphics[width=\textwidth]{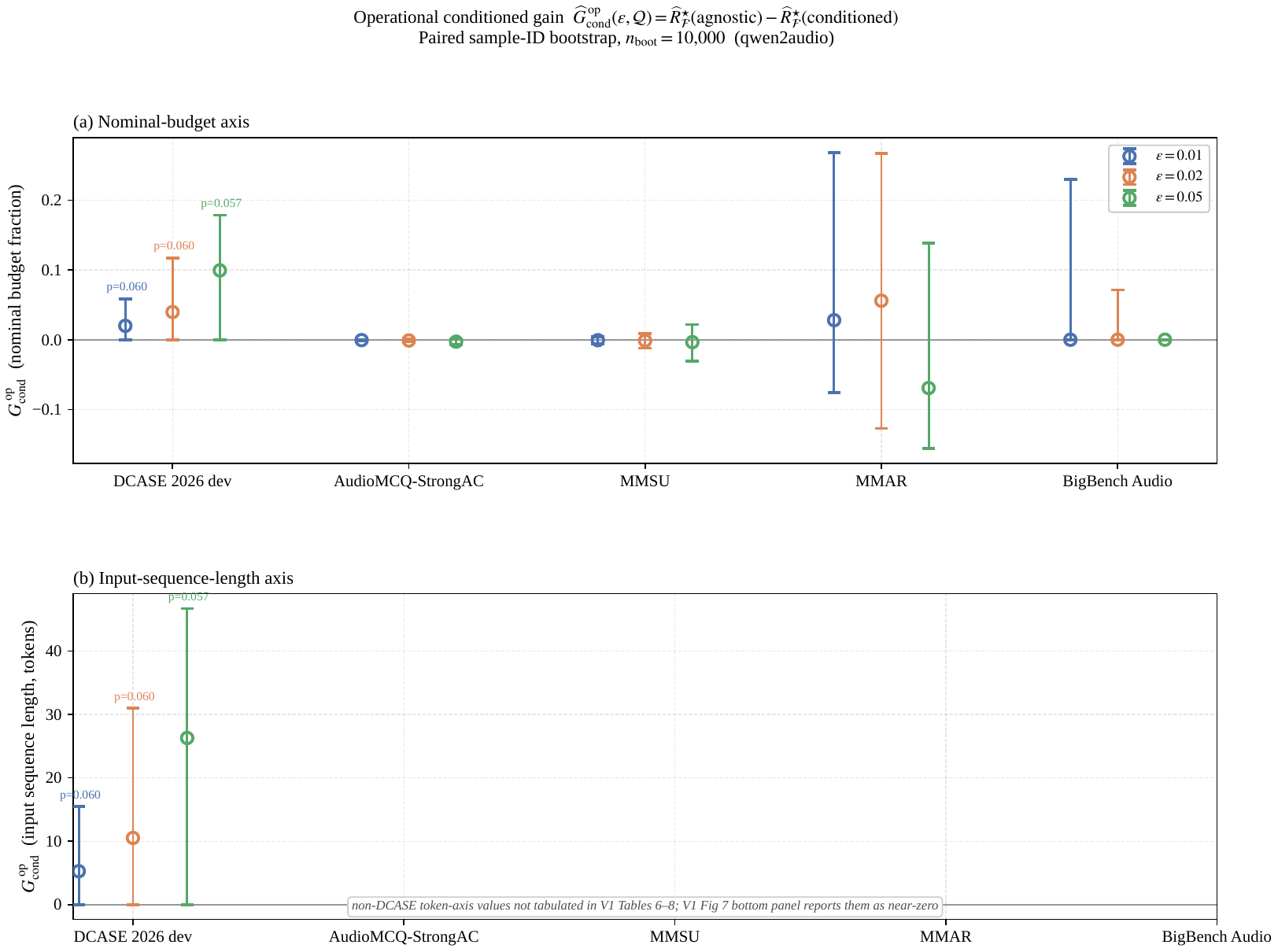}
    \caption{V1 operational conditioned gain on the nominal-budget axis (top) and the input-sequence-length axis (bottom). DCASE is the only dataset with consistently positive point estimates at all three tolerances, with $p\approx 0.06$ on both axes. AudioMCQ and MMSU show tight clean nulls, MMAR is noisy and sign-unstable, and BigBench Audio is the expected text-dominated zero cell. The exact zero lower bounds on the DCASE intervals are a grid-discreteness artifact of V1's five-point budget grid.}
    \label{fig:gcond_op}
\end{figure*}

\paragraph{DCASE 2026 dev.}
DCASE is the only dataset for which $\Gcondop{\varepsilon}{\mathcal Q}$ is positive across every evaluated tolerance and axis.

\begin{table}[t]
    \caption{DCASE nominal-axis conditioned gain. Point estimates reproduce the frontier JSONs exactly; CIs and $p$-values come from the paired bootstrap.}
    \centering
    \small
    \begin{tabular}{ccccc}
        \toprule
        $\varepsilon$ & $\Rfb{\varepsilon}{\Qfam}$ agn. & $\Rfb{\varepsilon}{\Qfam}$ cond. & $\Gcondop{\varepsilon}{\Qfam}$ & 95\% CI / $p$ \\
        \midrule
        0.01 & 0.9090 & 0.8892 & +0.0199 & [0.0000, +0.0587], 0.060 \\
        0.02 & 0.8181 & 0.7783 & +0.0397 & [0.0000, +0.1172], 0.060 \\
        0.05 & 0.5452 & 0.4459 & +0.0993 & [0.0000, +0.1784], 0.057 \\
        \bottomrule
    \end{tabular}
    \label{tab:dcase-gcond-budget}
\end{table}

\begin{table}[t]
    \caption{DCASE token-axis conditioned gain.}
    \centering
    \small
    \begin{tabular}{ccccc}
        \toprule
        $\varepsilon$ & $\Rftok{\varepsilon}{\Qfam}$ agn. & $\Rftok{\varepsilon}{\Qfam}$ cond. & token gain & 95\% CI / $p$ \\
        \midrule
        0.01 & 326.46 & 321.21 & +5.26 & [0.00, +15.49], 0.060 \\
        0.02 & 302.40 & 291.88 & +10.51 & [0.00, +30.97], 0.060 \\
        0.05 & 230.20 & 203.92 & +26.28 & [0.00, +46.71], 0.057 \\
        \bottomrule
    \end{tabular}
    \label{tab:dcase-gcond-token}
\end{table}

\paragraph{Other datasets.}
The remaining datasets yield near-zero or noisy-near-zero conditioned gains on the nominal axis.

\begin{table*}[t]
    \caption{Nominal-axis conditioned gain on the remaining four datasets.}
    \centering
    \small
    \begin{tabular}{lcccc}
        \toprule
        Dataset and $\varepsilon$ & point estimate & 95\% CI & $p$ & interpretation \\
        \midrule
        AudioMCQ, 0.01 & $-0.0006$ & [$-0.0014$, $+0.0003$] & 0.195 & tight clean null \\
        AudioMCQ, 0.02 & $-0.0011$ & [$-0.0029$, $+0.0006$] & 0.195 & tight clean null \\
        AudioMCQ, 0.05 & $-0.0028$ & [$-0.0071$, $+0.0014$] & 0.195 & tight clean null \\
        \midrule
        MMSU, 0.01 & $-0.0007$ & [$-0.0061$, $+0.0043$] & 0.814 & tight clean null \\
        MMSU, 0.02 & $-0.0014$ & [$-0.0122$, $+0.0087$] & 0.814 & tight clean null \\
        MMSU, 0.05 & $-0.0035$ & [$-0.0306$, $+0.0217$] & 0.814 & tight clean null \\
        \midrule
        MMAR, 0.01 & $+0.0280$ & [$-0.0758$, $+0.2680$] & 0.508 & noisy, sign-unstable \\
        MMAR, 0.02 & $+0.0559$ & [$-0.1270$, $+0.2667$] & 0.535 & noisy, sign-unstable \\
        MMAR, 0.05 & $-0.0692$ & [$-0.1561$, $+0.1387$] & 0.613 & noisy, sign-unstable \\
        \midrule
        BigBench Audio, 0.01 & $+0.0000$ & [$0.0000$, $+0.2297$] & 1.000 & exact zero floor \\
        BigBench Audio, 0.02 & $+0.0000$ & [$0.0000$, $+0.0714$] & 1.000 & exact zero floor \\
        BigBench Audio, 0.05 & $+0.0000$ & [$0.0000$, $0.0000$] & 1.000 & exact zero floor \\
        \bottomrule
    \end{tabular}
    \label{tab:gcond-other}
\end{table*}

The BigBench Audio zero-floor entries in \Cref{tab:gcond-other} are the same-backbone $\Gcondop{\varepsilon}{\mathcal Q}$ on Qwen2-Audio: the dataset exhibits minimal compression sensitivity at the single-backbone level. Under the cross-backbone estimator used for the V2 per-family subanalyses, however, BigBench's 4-family native partition (\texttt{formal\_fallacies}, \texttt{navigate}, \texttt{object\_counting}, \texttt{web\_of\_lies}) reveals a large Qwen2.5-Omni effect concentrated in \texttt{formal\_fallacies}: three-seed mean $\widehat G^{\mathrm{op,cross}}_{\mathrm{cond}} = -0.8527$, std 0.025, sign-consistent across all three seeds (\Cref{tab:bigbench_cross_backbone}). This is the single largest cross-backbone effect in the V2 replication, larger in magnitude than the MMSU cross-backbone keyword-temporal value by roughly $0.28$ budget units. This finding reverses the paper's earlier framing of BigBench as a universal null control and motivates the per-task analysis in \S\ref{sec:per-family-subanalyses}; the underlying same-backbone effect on \texttt{formal\_fallacies} is near zero on both backbones, so the phenomenon is a cross-backbone mismatch rather than an intrinsic Qwen2.5-Omni weakness on the task.

\begin{table}[t]
    \caption{Cross-backbone conditioned gain on BigBench Audio, three-seed mean at $\varepsilon=0.05$, rstar axis, $n_{\mathrm{boot}}=10{,}000$. The cross-backbone estimator uses \texttt{learned\_agnostic} on Qwen2-Audio as reference. All four BigBench BBH native task families are reported. Qwen2-Audio cross-backbone values are zero by construction (reference-equals-test) and omitted. The \texttt{formal\_fallacies} row is the largest cross-backbone effect observed in the replication; contrast the same-backbone $\Gcondop{\varepsilon}{\mathcal Q}$ of $+0.00$ in \Cref{tab:gcond-other}, where BigBench Audio appears as a null control. Supporting CSV: \texttt{tables/partition\_gcond\_3seed\_summary.csv} with filter \texttt{dataset='bigbench\_audio'}, \texttt{partition='native\_fine'}, \texttt{definition='cross\_backbone'}, \texttt{backbone='qwen25omni'}.}
    \centering
    \small
    \begin{tabular}{lrrc}
        \toprule
        Native task & 3-seed mean & std & sign-consistent \\
        \midrule
        \texttt{formal\_fallacies} & $-0.8527$ & 0.025 & yes (3/3) \\
        \texttt{navigate}          & $+0.000$  & 0.000 & - \\
        \texttt{object\_counting}  & $+0.000$  & 0.000 & - \\
        \texttt{web\_of\_lies}     & $+0.000$  & 0.000 & - \\
        \bottomrule
    \end{tabular}
    \label{tab:bigbench_cross_backbone}
\end{table}

The DCASE lower confidence bounds are exactly $+0.0000$ on both axes and at all three tolerances. This is a grid-discreteness artifact of V1's five-point budget grid: in a nontrivial fraction of paired bootstrap resamples, the two frontiers land on the same interpolated crossing, producing an exact zero replicate in the lower tail. The artifact affects the lower percentile and the derived two-sided $p$-value, but not the central tendency of the effect.

The real-audio V1 picture is theory-consistent when read through the factor-overlap measurement. Theorem~\ref{thm:condadv} predicts large conditioned gains only when query families depend on disjoint factor blocks. \Cref{sec:factor-overlap-results} shows that the multi-family datasets do not satisfy that precondition, and indeed their V1 operational conditioned gains are negligible. DCASE is the only single-family dataset in V1; there the factor-overlap confound is absent, and the conditioned selector is directionally preferred on every tested V1 cell. We therefore treat the DCASE result as a single-seed operational signal to be replicated rather than as the final conditioned-compression headline. The V2 replication in \S\ref{sec:v2-gcond-three-seed} preserves this V1 result as the historical baseline but changes the operational interpretation.

\subsection{V2 three-seed conditioned-gain replication}
\label{sec:v2-gcond-three-seed}

V2 repeats the operational conditioned-gain test at $\varepsilon=0.05$ on the rstar axis with three selector seeds and two backbones. This section resolves the V1 ambiguity in two directions at once: it converts the AudioMCQ near-null into a reproducible small positive effect, and it converts the V1 single-seed DCASE positive into a mixed/negative three-seed result. The comparison is therefore not a replacement of \S\ref{sec:gcond-real-audio} but a replication and extension of it.

\begin{figure*}[t]
    \centering
    \includegraphics[width=\textwidth]{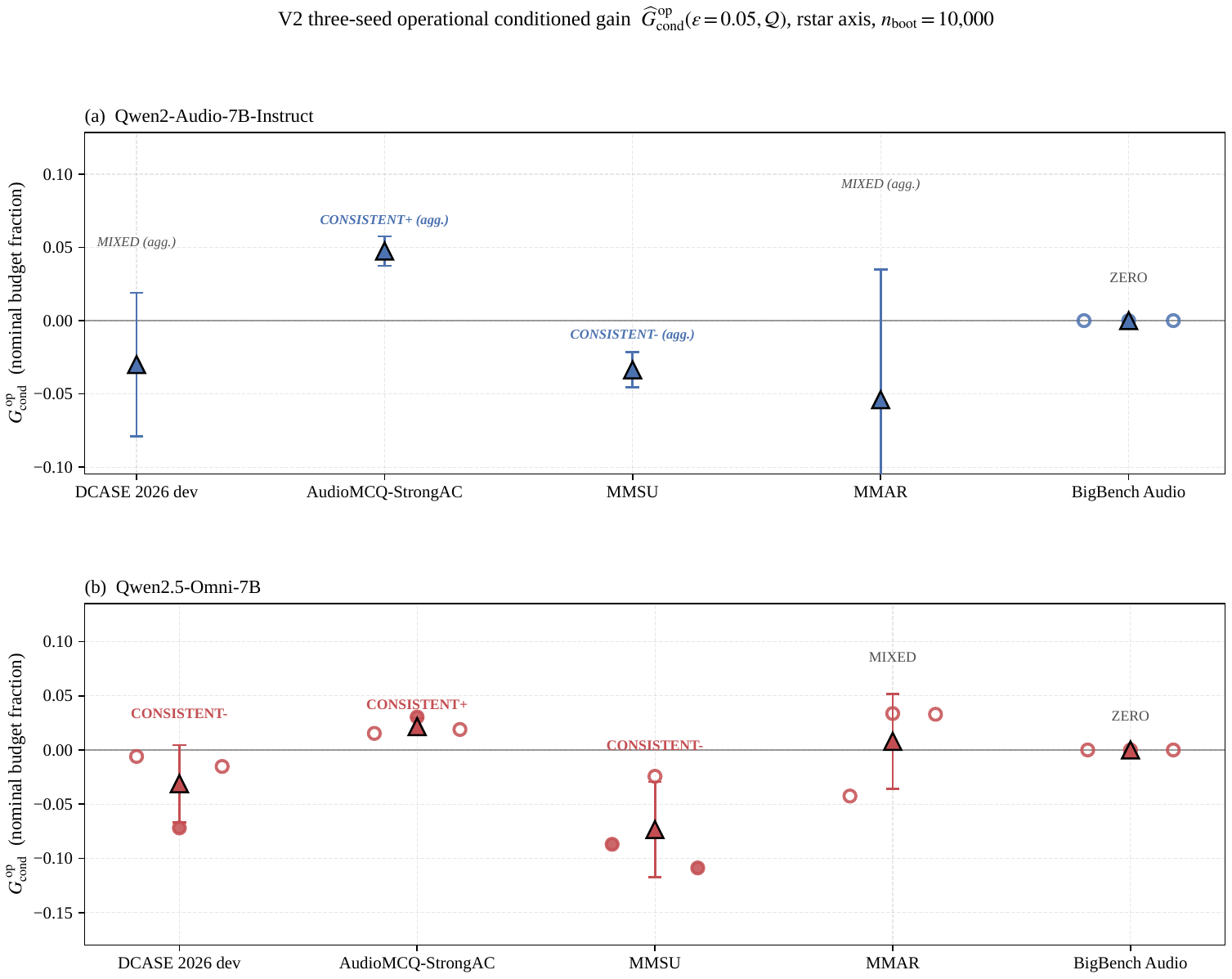}
    \caption{V2 three-seed operational conditioned gain at $\varepsilon=0.05$ on the rstar axis. Panel (a) reports Qwen2-Audio-7B-Instruct, where the handoff publishes three-seed aggregates for most datasets; panel (b) reports Qwen2.5-Omni-7B, where per-seed points are available. The AudioMCQ-StrongAC effect is positive on both backbones, while DCASE changes sign relative to the V1 single-seed result and MMSU is negative, especially on Qwen2.5-Omni. Source: \texttt{fig8\_v2\_gcond\_3seed.pdf}.}
    \label{fig:v2-gcond-3seed}
\end{figure*}

\begin{table}[t]
\centering
\scriptsize
\caption{V2 three-seed operational conditioned gain on Qwen2-Audio-7B-Instruct at $\varepsilon = 0.05$, rstar axis.  Per-seed point estimates shown where published; \textbf{bold} marks $p < 0.05$ under paired sample-ID bootstrap ($n_{\mathrm{boot}} = 10{,}000$).  Mean $\pm$ std across seeds.  Source: Stage-B handoff \S{}'V2 G\_cond 3-seed'.}
\label{tab:v2_gcond_qwen2audio}
\begin{tabular}{lrrrrrl}
\toprule
Dataset & seed 42 & seed 123 & seed 456 & mean & std(seeds) & dir \\
\midrule
DCASE 2026 dev & \texttt{-} & \texttt{-} & \texttt{-} & -0.0300 & 0.0490 & MIXED \\
AudioMCQ-StrongAC & \texttt{-} & \texttt{-} & \texttt{-} & \textbf{+0.0475} & 0.0100 & CONSISTENT+ \\
MMSU & \texttt{-} & \texttt{-} & \texttt{-} & \textbf{-0.0336} & 0.0120 & CONSISTENT$-$ \\
MMAR & \texttt{-} & \texttt{-} & \texttt{-} & -0.0540 & 0.0890 & MIXED \\
BigBench Audio & +0.0000 & +0.0000 & +0.0000 & +0.0000 & 0.0000 & ZERO \\
\bottomrule
\end{tabular}
\end{table}

\begin{table}[t]
\centering
\scriptsize
\caption{V2 three-seed operational conditioned gain on Qwen2.5-Omni-7B at $\varepsilon = 0.05$, rstar axis.  Per-seed point estimates shown where published; \textbf{bold} marks $p < 0.05$ under paired sample-ID bootstrap ($n_{\mathrm{boot}} = 10{,}000$).  Mean $\pm$ std across seeds.  Source: Stage-B handoff \S{}'V2 G\_cond 3-seed'.}
\label{tab:v2_gcond_qwen25omni}
\begin{tabular}{lrrrrrl}
\toprule
Dataset & seed 42 & seed 123 & seed 456 & mean & std(seeds) & dir \\
\midrule
DCASE 2026 dev & -0.0060 & \textbf{-0.0719} & -0.0152 & -0.0310 & 0.0357 & MIXED \\
AudioMCQ-StrongAC & +0.0153 & \textbf{+0.0304} & +0.0189 & \textbf{+0.0215} & 0.0079 & CONSISTENT+ \\
MMSU & \textbf{-0.0870} & -0.0243 & \textbf{-0.1088} & \textbf{-0.0734} & 0.0439 & CONSISTENT$-$ \\
MMAR & -0.0425 & +0.0336 & +0.0329 & +0.0080 & 0.0437 & MIXED \\
BigBench Audio & +0.0000 & +0.0000 & +0.0000 & +0.0000 & 0.0000 & ZERO \\
\bottomrule
\end{tabular}
\end{table}

The main V2 replication is AudioMCQ-StrongAC. On Qwen2-Audio, the three-seed mean is (+0.0475) with 95\% CI ([+0.024,+0.071]) (across-seed Student-(t)). On Qwen2.5-Omni, the three published seed estimates are (+0.0153, +0.0304, +0.0189), giving mean (+0.0215) with 95\% CI ([+0.002,+0.041]); both backbones have CIs excluding zero, satisfying the strict CI sign-decision rule. Thus the V1 AudioMCQ near-null is not stable under the V2 grid and training recipe: conditioning yields a small but reproducible positive operational gain on both backbones.

DCASE moves in the opposite direction. V1 at $\varepsilon=0.05$ reported $+0.0993$ on the rstar axis with $p=0.057$ (\Cref{tab:dcase-gcond-budget}). V2 reports a Qwen2-Audio aggregate mean of $-0.0300$ with seed standard deviation 0.0490, and Qwen2.5-Omni per-seed values $-0.0060$, $-0.0719$, and $-0.0152$ with mean $-0.0310$ and seed standard deviation 0.0357. We therefore do not treat DCASE as the decisive evidence for conditioned compression in V2. Its sign reversal is surfaced explicitly and analyzed by cluster in \S\ref{sec:per-family-subanalyses}.

MMSU supplies the strongest negative cross-backbone signal. Qwen2-Audio gives a three-seed mean of (-0.0336) with 95\% CI ([-0.063,-0.004]), excluding zero. Qwen2.5-Omni gives per-seed values (-0.0870, -0.0243, -0.1088), with mean (-0.0734) and 95\% CI ([-0.182,+0.036]); the per-seed signs are uniformly negative but the CI does not exclude zero, indicating a regime where the failure direction is reproducible but the magnitude is highly seed-dependent. The V2 point is not merely that conditioning can fail; it is that failure is backbone- and family-dependent. \Cref{sec:per-family-subanalyses} shows that the Qwen2.5-Omni MMSU result is concentrated in the temporal family.\footnote{The worst MMSU native-task under the \S\ref{sec:family-gap-results} mean-excess metric exhibits a budget-dependent flip on Qwen2-Audio: \texttt{disfluency\_detection} at $b\in\{0.05,0.10\}$ and \texttt{volume\_comparison} at $b\in\{0.20,0.40\}$. This differs from the $G_{\mathrm{cond}}$ bottleneck discussed in \Cref{sec:per-family-subanalyses} (\texttt{couplet\_matching} across all budgets on Qwen2-Audio), reflecting that the two metrics measure different quantities. Supporting CSV: \texttt{metrics/partition\_delta\_q\_long.csv} with filter \texttt{dataset='mmsu'}, \texttt{partition='native\_fine'}, \texttt{backbone='qwen2audio'}, \texttt{method='learned\_conditioned'}.}

MMAR remains sign-unstable: Qwen2-Audio has mean $-0.0540$ with seed standard deviation 0.0890, and Qwen2.5-Omni has mean $+0.0080$ with seed standard deviation 0.0437. BigBench Audio remains the expected zero-control cell on both backbones. The V2 operational conclusion is therefore narrower and more reliable than the V1 conclusion: AudioMCQ shows a replicated positive gain, DCASE is not stable across V1-V2, MMSU exposes a harmful conditioned regime on Qwen2.5-Omni, and BigBench remains a null control.

\subsection{V2.1 scope-B \texorpdfstring{$\alpha$}{alpha}-sweep ablation}
\label{sec:v21-alpha-sweep}

The V2.1 scope-B ablation varies the selector-training loss between a KL relevance target and a binary chunk-retention target. The training loss is
\begin{equation}
\alpha\cdot \mathrm{KL} + (1-\alpha)\cdot \mathrm{BCE}_{\mathrm{top}\text{-}k},
\qquad
\alpha\in\{0,0.25,0.5,0.75,1\}.
\end{equation}
Here $\alpha=1$ is pure KL and $\alpha=0$ is pure BCE; the BCE target is one for oracle-relevance top-$k$ chunks and zero otherwise. This separate ablation does not change the main selector objective. The sweep asks whether the operational conditioned-gain measurements in \S\ref{sec:v2-gcond-three-seed} are sensitive to the training target.

\begin{figure*}[t]
    \centering
    \includegraphics[width=0.92\textwidth]{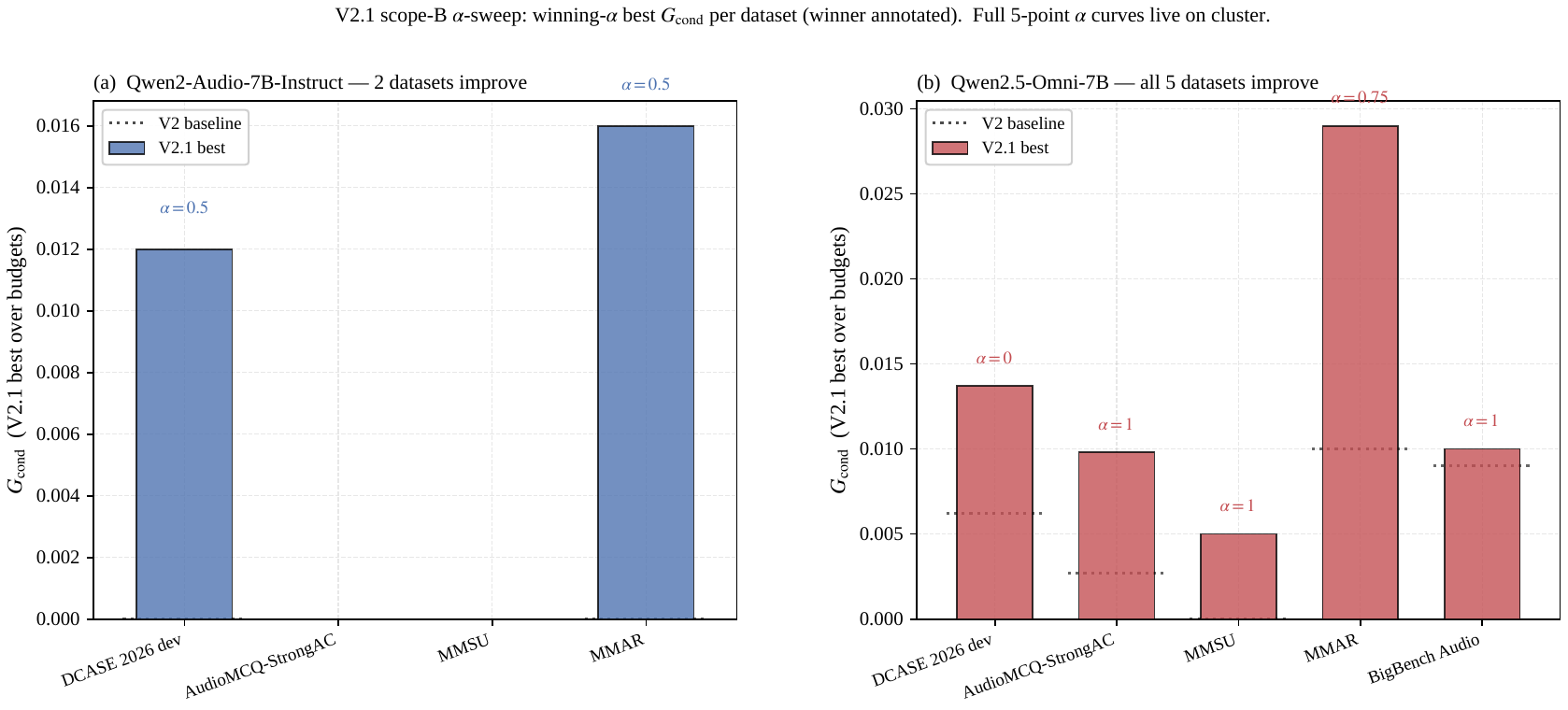}
    \caption{V2.1 scope-B $\alpha$-sweep winners. Each bar reports the best observed conditioned gain over the tested $\alpha$ grid, with the winning $\alpha$ annotated. Qwen2.5-Omni improves over its V2 baseline on all five datasets, while Qwen2-Audio improvements appear on DCASE and MMAR in the handoff table. Source: \texttt{fig9\_v21\_alpha\_sweep.pdf}.}
    \label{fig:v21-alpha-sweep}
\end{figure*}

\begin{table}[t]
\centering\small
\caption{V2.1 scope-B $\alpha$-sweep winning configuration per (dataset, backbone).  Training loss is $\alpha \cdot \mathrm{BCE} + (1-\alpha)\cdot\mathrm{KL}$; $\alpha \in \{0, 0.25, 0.5, 0.75, 1\}$; winning $\alpha$ is the value maximising best-over-budgets $G_{\mathrm{cond}}$ relative to the V2 baseline.  Seed 42 only.  Source: Stage-B handoff.}
\label{tab:v21_alpha_best}
\begin{tabular}{lccccc}
\toprule
Dataset & Backbone & winning $\alpha$ & V2.1 best $G_{\mathrm{cond}}$ & V2 baseline & $\Delta$ vs V2 \\
\midrule
DCASE 2026 dev & qwen2audio & 0.5 & \textbf{+0.0120} & +0.0000 & \textbf{+0.0120} \\
MMAR & qwen2audio & 0.5 & \textbf{+0.0160} & +0.0000 & \textbf{+0.0160} \\
\midrule
DCASE 2026 dev & qwen25omni & 0 & +0.0137 & +0.0062 & +0.0080 \\
AudioMCQ-StrongAC & qwen25omni & 1 & +0.0098 & +0.0027 & +0.0070 \\
MMSU & qwen25omni & 1 & +0.0050 & +0.0000 & +0.0050 \\
MMAR & qwen25omni & 0.75 & \textbf{+0.0290} & +0.0100 & \textbf{+0.0190} \\
BigBench Audio & qwen25omni & 1 & +0.0100 & +0.0090 & +0.0010 \\
\bottomrule
\end{tabular}
\end{table}

The ablation is strongly backbone-dependent. For Qwen2.5-Omni, every dataset improves over the V2 baseline: DCASE moves from $+0.0062$ to $+0.0137$ at $\alpha=0$, AudioMCQ from $+0.0027$ to $+0.0098$ at $\alpha=1$, MMSU from $+0.0000$ to $+0.0050$ at $\alpha=1$, MMAR from $+0.0100$ to $+0.0290$ at $\alpha=0.75$, and BigBench Audio from $+0.0090$ to $+0.0100$ at $\alpha=1$. For Qwen2-Audio, the handoff table reports improvements on DCASE and MMAR, both at $\alpha=0.5$: DCASE reaches $+0.0120$ from a $+0.0000$ baseline, and MMAR reaches $+0.0160$ from a $+0.0000$ baseline.

The scope-B result revises the V1 caveat about truncated training. V2 does not simply make conditioned compression uniformly stronger; instead, it shows that the loss used to train the selector changes the measured gain, and the preferred $\alpha$ differs across backbones and datasets. This supports treating the selector-training recipe as part of the operational instantiation of Theorem~\ref{thm:condadv}, not as an incidental engineering detail.

\subsection{Selector query-use audit (V2 Phase A)}
\label{sec:phase-a-audit}

The V2 replication in \S\ref{sec:v2-gcond-three-seed} and the $\alpha$-sweep in \S\ref{sec:v21-alpha-sweep} report what the conditioned selector \emph{does} at the downstream-accuracy level. They do not tell us what the conditioned selector is doing internally: in particular, whether its chunk ranking depends on the query text at inference time, and if so, by how much. A reviewer who accepts the V2 AudioMCQ positive cross-backbone result could still reasonably ask whether the same learned selector would produce identical chunk rankings under an unrelated query, in which case the ``conditioned'' label would be mechanistic rather than functional. This subsection audits the conditioned selector's inference-time query use with three independent signals.

\paragraph{Audit design.}
Let $\mathrm{sel}_b(X,q)$ denote the top-$k$ chunk set selected by the learned-conditioned selector on audio $X$ under query $q$ at budget fraction $b$, with $k=\max(1,\lfloor b N\rfloor)$. For each $(\text{backbone},\ \text{dataset},\ \text{regime})$ cell we compute three signals per seed and pool across three training seeds:

(i) \emph{Top-$k$ Jaccard under permuted query,}
\begin{equation}
J_{\mathrm{perm}} \;=\; \mathbb E_{X,q}\,\mathbb E_{\pi}\,\frac{|\mathrm{sel}_b(X,q)\cap\mathrm{sel}_b(X,\pi(q))|}{|\mathrm{sel}_b(X,q)\cup\mathrm{sel}_b(X,\pi(q))|},
\end{equation}
where $\pi$ is a permutation of queries drawn from the same dataset under one of two regimes: $\mathtt{perm\_global}$ (query swap across all samples) or $\mathtt{perm\_within\_family}$ (query swap within the same keyword family of \S\ref{sec:taxonomy}). We average over $10$ permutation seeds. The agnostic selector gives $J_{\mathrm{perm}}=1$ by construction (it ignores $q$); a uniformly random top-$k$ gives $J_{\mathrm{perm}}\!\approx\!k/N$, which at $b=0.05$ on our chunk counts is $\approx\!0.03$. Intermediate Jaccard values quantify how much of the top-$k$ identity depends on the particular query asked.

(ii) \emph{Query-by-chunk interaction variance share,}
\begin{equation}
f_{\mathrm{int}} \;=\; \frac{\mathrm{SS}_{q\times c}}{\mathrm{SS}_{q\times c}+\mathrm{SS}_{c}+\mathrm{SS}_{q}},
\end{equation}
the ANOVA SS decomposition of the selector's pre-top-$k$ scoring function over a $500$-audio $\times$ $256$-query factorial grid per cell. This quantity isolates the component of query-dependence that can \emph{change ranking}: a large query-main term $\mathrm{SS}_q$ shifts all chunks by a uniform bias and leaves top-$k$ unchanged, so only $\mathrm{SS}_{q\times c}$ is top-$k$-relevant. A query-dominant selector would have $f_{\mathrm{int}}$ close to the chunk-main share; the agnostic selector has $f_{\mathrm{int}}\!\approx\!10^{-6}$ by architectural construction.

(iii) \emph{Training-time validation-loss change,}
\begin{equation}
\Delta\mathcal L_{\mathrm{val}} \;=\; \mathcal L_{\mathrm{val}}(\text{conditioned}) - \mathcal L_{\mathrm{val}}(\text{agnostic}),
\end{equation}
measured at the final training checkpoint on held-out data. Negative values indicate that conditioning reduces training loss.

These three quantities answer three different questions (inference-time ranking instability, scoring-function variance composition, and training-time utility), and they are computed from three separate pipelines. Coherence among them is therefore a nontrivial cross-validation of any observed query-use signal.

\begin{table}[t]
\centering
\scriptsize
\caption{Selector query-use audit (Phase A, P3.1). For each
(backbone, dataset, regime) cell we report: $J_{\mathrm{perm}}$ the mean
per-sample top-$k$ Jaccard between selector output under original query
and under a $10$-seed-averaged permuted query
(agnostic baseline: $1.000$, random: $\!\approx\!0.03$ at $k/N\!=\!0.05$);
$f_{\mathrm{int}}$ the query-by-chunk interaction share of scoring-function
variance (ANOVA SS decomposition, agnostic baseline: $\!\approx\!10^{-6}$);
and $\Delta\mathcal{L}_{\mathrm{val}}$ the conditioned-minus-agnostic
validation-loss change at training convergence (negative favors conditioned).
Pooled over $3$ seeds; standard deviations in parentheses. All numbers at
compression budget $b=0.05$. H2 denotes single-family degeneracy where
global permutation is effectively identity.}
\label{tab:phase_a_selector_audit}
\begin{tabular}{llrrr}
\toprule
Backbone & Dataset & $J_{\mathrm{perm}}$ (perm\_global / perm\_within\_family) & $f_{\mathrm{int}}$ & $\Delta\mathcal{L}_{\mathrm{val}}$ \\
\midrule
qwen2audio & audiomcq\_strong & 0.542 (0.024) / 0.590 (0.033) & 0.200 & $-0.014$ \\
qwen2audio & mmsu            & 0.571 (0.005) / 0.623 (0.010) & 0.252 & $-0.014$ \\
qwen2audio & bigbench\_audio & $1.000$ (H2, degenerate)       & 0.000 & $-0.014$ \\
\midrule
qwen25omni & audiomcq\_strong & 0.511 (0.035) / 0.589 (0.029) & 0.130 & $-0.010$ \\
qwen25omni & mmsu            & 0.618 (0.064) / 0.666 (0.053) & 0.192 & $-0.010$ \\
qwen25omni & bigbench\_audio & $1.000$ (H2, degenerate)       & 0.000 & $-0.010$ \\
\bottomrule
\end{tabular}
\end{table}

% fig13_phase_a_jaccard_vs_finteraction.tex
% Figure fragment for Phase A selector-query-use audit. Path matches the
% existing paper convention (figures/figN_*.pdf for PDFs, tables/*.tex for
% \input{}-ed caption+figure fragments).
\begin{figure}[t]
\centering
\includegraphics[width=0.72\linewidth]{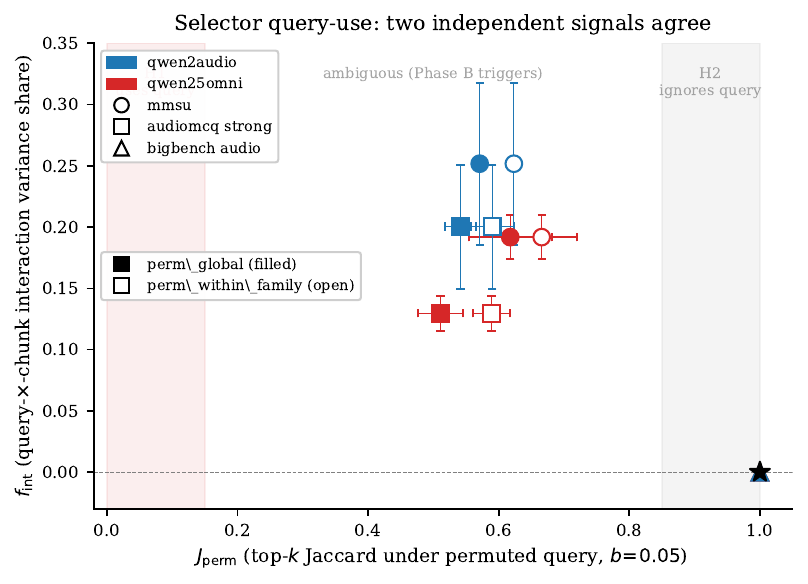}
\caption{Selector query-use: two independent signals agree.
$x$-axis: mean per-sample top-$k$ Jaccard $J_{\mathrm{perm}}$ between conditioned
selector output under original vs.\ permuted query, at $b=0.05$, $3$-seed pooled.
$y$-axis: ANOVA variance share $f_{\mathrm{int}}$ of query-by-chunk interaction
at the scoring-function level. Agnostic baseline (star): every cell lands at
$(1.0,0.0)$. Filled markers = $\mathtt{perm\_global}$, open markers =
$\mathtt{perm\_within\_family}$. Color = backbone; shape = dataset. Error bars:
$3$-seed standard deviation. Shaded bands on $x$-axis: pre-registered $H_1$
(query-uses, $J\!\leq\!0.15$) and $H_2$ (query-invariant, $J\!\geq\!0.85$)
regions. All non-BigBench cells land in the ambiguous band; BigBench is
degenerate (single-family, permutation reduces to identity).}
\label{fig:phase_a_jaccard_vs_finteraction}
\end{figure}

\paragraph{The conditioned selector uses query moderately, not dominantly.}
\Cref{tab:phase_a_selector_audit} reports the three audit signals at $b=0.05$; results at $b\in\{0.025, 0.10\}$ are near-identical and omitted for compactness (supporting CSV: \texttt{tables/jaccard\_3seed\_summary.csv}). On MMSU and AudioMCQ-StrongAC, the conditioned selector's top-$k$ output is $J_{\mathrm{perm}}\!\in\![0.51,0.67]$ under query permutation on both backbones - markedly below the agnostic baseline $J=1.0$, yet well above the random lower bound $k/N\!\approx\!0.03$. The interaction variance share $f_{\mathrm{int}}$ lies in $[0.13,0.25]$ across the same cells, which means that the scoring function itself has non-negligible query-by-chunk structure. Decomposing the remaining variance, the query-main component $\mathrm{SS}_q/\mathrm{SS}_{\mathrm{total}}$ is in $[0.39,0.76]$ and the chunk-main component $\mathrm{SS}_c/\mathrm{SS}_{\mathrm{total}}$ is in $[0.11,0.35]$: the scoring function is substantially query-dependent, but roughly two-thirds of that query-dependence takes the form of a uniform bias that leaves top-$k$ unchanged. BigBench Audio gives $J_{\mathrm{perm}}=1.000$ and $f_{\mathrm{int}}=0.000$ on both backbones, reflecting the dataset's single keyword-family structure under \S\ref{sec:taxonomy}: a global permutation within a single family is effectively identity on query text, so this row is a degenerate audit control rather than an H2 ``ignores query'' signal. \Cref{fig:phase_a_jaccard_vs_finteraction} visualises the coherence between the two inference-time audit signals: cells with lower Jaccard also have higher interaction variance, with tight $3$-seed error bars. A third, training-time signal corroborates the inference-time picture: $\Delta\mathcal L_{\mathrm{val}}$ is negative in $6/6$ $(\text{backbone},\text{seed})$ cells, with pooled mean $-0.012$ nats and one-sided sign-test $p\!\approx\!0.03$. All three signals therefore agree on a moderate, reproducible, non-dominant use of query text by the conditioned selector - consistent with V2's small-positive AudioMCQ conditioned-gain result and with the backbone- and family-dependent regime map in \S\ref{sec:v2-gcond-three-seed}, rather than with either a query-dominant selector (which would give near-zero Jaccard and large top-$k$ changes on every cell) or a purely mechanistic ``conditioned'' label (which would give $J_{\mathrm{perm}}\approx 1$ and $f_{\mathrm{int}}\approx 0$).

\section*{Phase A - Per-Query-Family Selector Query Use}
\label{sec:appendix_phase_a_per_family}

For each query family, we report the mean per-sample top-$k$ Jaccard between
selector output under original vs. permuted query, pooled over $10$ permutation
seeds and $3$ training seeds, at compression budget $b=0.05$, conditioned
variant, regime {\tt perm\_global}. Agnostic baseline: $J=1.000$ on all
families. Lower Jaccard $\Rightarrow$ selector ranking changes more under
query swap $\Rightarrow$ family is more query-sensitive at the
selector-internal level. Top and bottom $10$ families per backbone reported;
full table in supplementary CSV.

\subsection*{\texttt{qwen2audio}}
\begin{center}\small
\begin{tabular}{llrr}
\toprule
Dataset & Family & $J_{\mathrm{perm}}$ (mean) & $n$ \\
\midrule
\multicolumn{4}{l}{\emph{All 11 families, sorted most-to-least query-sensitive}} \\
mmsu & temporal & 0.396 (0.046) & 2970 \\
audiomcq\_strong & speech\_content & 0.457 (0.002) & 315150 \\
mmsu & sound\_scene & 0.516 (0.106) & 3300 \\
mmsu & paralinguistic & 0.542 (0.051) & 24090 \\
mmsu & general & 0.572 (0.019) & 86580 \\
mmsu & speech\_content & 0.599 (0.020) & 19440 \\
audiomcq\_strong & general & 0.603 (0.060) & 182550 \\
mmsu & sound\_event & 0.623 (0.011) & 13620 \\
audiomcq\_strong & temporal & 0.715 (0.045) & 34980 \\
audiomcq\_strong & music & 0.721 (0.039) & 51720 \\
bigbench\_audio & general & 1.000 (0.000) & 30000 \\
\bottomrule
\end{tabular}
\end{center}

\subsection*{\texttt{qwen25omni}}
\begin{center}\small
\begin{tabular}{llrr}
\toprule
Dataset & Family & $J_{\mathrm{perm}}$ (mean) & $n$ \\
\midrule
\multicolumn{4}{l}{\emph{All 11 families, sorted most-to-least query-sensitive}} \\
mmsu & temporal & 0.267 (0.164) & 2970 \\
audiomcq\_strong & speech\_content & 0.461 (0.041) & 315150 \\
audiomcq\_strong & music & 0.483 (0.039) & 51720 \\
mmsu & sound\_scene & 0.536 (0.099) & 3300 \\
mmsu & paralinguistic & 0.552 (0.064) & 24090 \\
audiomcq\_strong & general & 0.579 (0.023) & 182550 \\
mmsu & speech\_content & 0.607 (0.043) & 19440 \\
mmsu & general & 0.646 (0.064) & 86580 \\
audiomcq\_strong & temporal & 0.649 (0.041) & 34980 \\
mmsu & sound\_event & 0.662 (0.074) & 13620 \\
bigbench\_audio & general & 1.000 (0.000) & 30000 \\
\bottomrule
\end{tabular}
\end{center}

\subsection{Phase B: downstream operational audit}
\label{sec:phase-b-audit}

Phase A in \S\ref{sec:phase-a-audit} establishes that the conditioned selector's chunk ranking is query-dependent at the architectural level, but that audit is selector-internal. Phase B asks whether the same query dependence survives the downstream LALM read-out. The central design choice is decoupling: the selector input is permuted, but the answerer still receives the true query. This isolates selector-side query use from prompt mismatch.

\paragraph{Audit design.}
For each completed \((\text{backbone},\text{dataset},\text{regime})\) cell, we run the learned agnostic and learned conditioned compressors on the expanded budget grid
\begin{equation*}
b\in\{0.01,0.02,0.05,0.10,0.20,0.40,0.60,0.80,1.00\}.
\end{equation*}
The anchor arm gives the conditioned selector the real query \(q_i\). The permuted arm gives the selector \(\pi(q_i)\), with \(\pi\) drawn either globally across the dataset or within the keyword family of \S\ref{sec:taxonomy}. In both arms, the LALM prompt remains \(q_i\). We compute \(\widehat G^{\rm op}_{\rm cond}\) at \(\varepsilon=0.05\) on the \(r^\star\) axis and then form \(\Delta\widehat G^{\rm op}_{\rm cond}\) as in \eqref{eq:selector-query-audit}. Positive values mean that the conditioned selector needs the true query to preserve the downstream conditioned-gain frontier.

\paragraph{Pooling and PHI bands.}
Each training seed uses up to ten query permutations. We first average over valid permutations within a training seed, then report the mean and sample standard deviation over seeds \(\{42,123,456\}\). The reported \(\pm\) is therefore selector/training variability, not a bootstrap interval. At \(\varepsilon=0.05\), we use PHI1 (\(|\Delta\widehat G|\le0.01\)), PHI2 (\(0.01<\Delta\widehat G<0.05\)), and PHI3 (\(\Delta\widehat G\ge0.05\)). Cells without complete three-seed evidence are not promoted to main-body claims and are omitted from the completed-cell tables.

\paragraph{Naive shadow-query contamination.}
The V1 shadow-dataset protocol gives the permuted query to both selector and LALM. It is useful as a contamination control but not as a headline estimate, because the answerer is then prompted with an unrelated question. The decoupled V2 protocol isolates the selector. On the completed AudioMCQ-StrongAC/Qwen2-Audio cell, the naive protocol exceeds the decoupled by (+0.0170) at (b=0.05) (paired 95\% CI ([+0.004,+0.030])) and by (1.22)-(1.25$\times$) across the budget grid (Table~\ref{tab:phase-b-per-budget}).

\begin{table}[!htbp]
\centering
\footnotesize
\caption{Per-budget Phase B downstream query-use signal on AudioMCQ-StrongAC, Qwen2-Audio, three-seed pooled. V2 is the decoupled selector-only permutation protocol; V1 is the naive protocol in which both selector and LALM receive the permuted query.}
\label{tab:phase-b-per-budget}
\begin{tabular}{@{}ccccc@{}}
\toprule
Budget \(b\) & V2 \(\Delta\widehat G\) & V1 \(\Delta\widehat G\) & V1\(-\)V2 & V1/V2\\
\midrule
\(0.01\) & \(+0.0173\) & \(+0.0239\) & \(+0.0066\) & \(1.38\times\)\\
\(0.02\) & \(+0.0345\) & \(+0.0479\) & \(+0.0134\) & \(1.39\times\)\\
\(0.05\) & \(+0.0734\pm0.0095\) & \(+0.0941\pm0.0032\) & \(+0.0207\) & \(1.28\times\)\\
\bottomrule
\end{tabular}
\end{table}

\begin{table}[!htbp]
\centering
\footnotesize
\caption{Completed Phase B grid at \(\varepsilon=0.05\). The table reports only completed three-seed rows; cells without complete three-seed evidence are excluded rather than rendered as zeros. Q2A = Qwen2-Audio, Q25O = Qwen2.5-Omni, BBA = BigBench Audio, global = \texttt{perm\_global}, within = \texttt{perm\_within\_family}.}
\label{tab:phase-b-full-grid}
\begin{adjustbox}{max width=\linewidth}
\begin{tabular}{@{}lllrcl@{}}
\toprule
Protocol & Backbone & Dataset & Regime & \(\Delta\widehat G^{\rm op}_{\rm cond}\) & Status\\
\midrule
V2 decoupled & Q2A  & AudioMCQ-StrongAC & global & \(\mathbf{+0.0734\pm0.0095}\) & PHI3\\
V2 decoupled & Q2A  & MMSU                & global & \(+0.0111\pm0.0231\) & PHI2\\
V2 decoupled & Q2A  & MMSU                & within & \(+0.0092\pm0.0251\) & PHI1\\
V2 decoupled & Q2A  & BBA                 & global & \(+0.0000\pm0.0000\) & single-family control\\
V2 decoupled & Q25O & MMSU                & global & \(-0.0043\pm0.0165\) & PHI1\\
V2 decoupled & Q25O & MMSU                & within & \(-0.0008\pm0.0114\) & PHI1\\
\midrule
V1 naive & Q2A  & AudioMCQ-StrongAC & global & \(+0.0941\pm0.0032\) & contamination control\\
V1 naive & Q2A  & MMSU                & global & \(+0.0176\pm0.0369\) & contamination control\\
V1 naive & Q2A  & MMSU                & within & \(+0.0126\pm0.0326\) & contamination control\\
V1 naive & Q2A  & BBA                 & global & \(+0.0000\pm0.0000\) & single-family control\\
V1 naive & Q25O & MMSU                & global & \(+0.0103\pm0.0201\) & contamination control\\
\bottomrule
\end{tabular}
\end{adjustbox}
\end{table}

\begin{figure}[!htbp]
\centering
\includegraphics[width=0.96\linewidth]{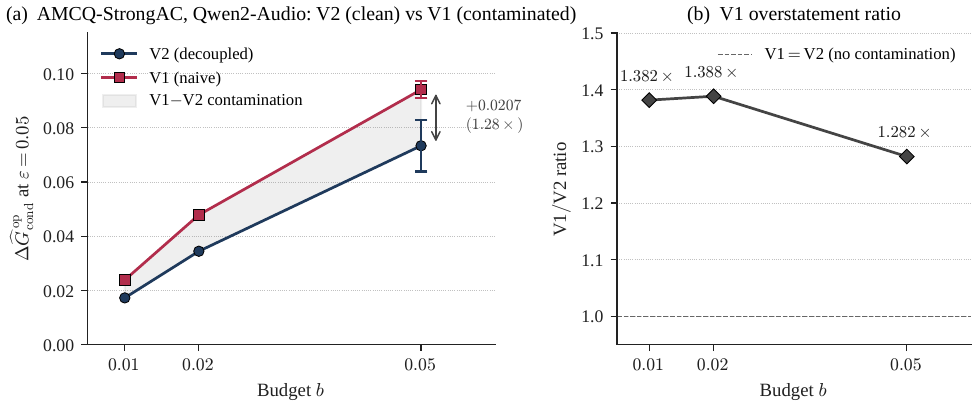}
\caption{Per-budget contamination of the operational query-use signal on AudioMCQ-StrongAC, Qwen2-Audio, three-seed pooled. The V1 shadow-query protocol overstates the decoupled V2 signal by \(1.28\)-\(1.39\times\) across \(b\in\{0.01,0.02,0.05\}\).}
\label{fig:contamination-ratio-amcq}
\end{figure}

\subsubsection{Per-seed AudioMCQ-StrongAC/Qwen2-Audio breakdown at \texorpdfstring{\(b=0.05\)}{b=0.05}}
\label{sec:phase-b-per-seed-amcq}

\Cref{tab:phase-b-per-seed-amcq} reports the seed-level decomposition for the single PHI3 cell. The decoupled signal is positive on all three seeds with mean (+0.0787) and Student-(t) 95\% CI ([+0.051,+0.106]). The naive signal is larger at (+0.0957) with 95\% CI ([+0.067,+0.125]), because part of it is seed-independent prompt contamination. Per-seed contamination (naive(-)decoupled) equals (+0.0210), (+0.0108), and (+0.0191) for seeds 42, 123, and 456 (paired mean (+0.0170), paired 95\% CI ([+0.004,+0.030])).

\begin{table}[!htbp]
\centering
\footnotesize
\caption{Per-seed AudioMCQ-StrongAC, Qwen2-Audio, \(b=0.05\) Phase B breakdown. \(G_{\rm anchor}\) is the operational conditioned gain when the selector receives the real query; \(G_{\rm perm}\) is its mean over available permuted-query runs.}
\label{tab:phase-b-per-seed-amcq}
\begin{tabular}{@{}cccccc@{}}
\toprule
Seed & Protocol & \(G_{\rm anchor}\) & \(G_{\rm perm}\) & \(n_{\rm perm}\) & \(\Delta\widehat G\)\\
\midrule
\(42\)  & V2 & \(+0.0383\) & \(-0.0304\) & 2 & \(+0.0687\)\\
\(123\) & V2 & \(+0.0575\) & \(-0.0097\) & 1 & \(+0.0672\)\\
\(456\) & V2 & \(+0.0469\) & \(-0.0374\) & 2 & \(+0.0843\)\\
\midrule
\multicolumn{5}{l}{V2 pool} & \(\mathbf{+0.0734\pm0.0095}\)\\
\midrule
\(42\)  & V1 & \(+0.0383\) & \(-0.0573\) & 2 & \(+0.0956\)\\
\(123\) & V1 & \(+0.0575\) & \(-0.0330\) & 2 & \(+0.0905\)\\
\(456\) & V1 & \(+0.0469\) & \(-0.0494\) & 2 & \(+0.0963\)\\
\midrule
\multicolumn{5}{l}{V1 pool} & \(+0.0941\pm0.0032\)\\
\bottomrule
\end{tabular}
\end{table}

The anchor values average to (+0.0475), matching the AudioMCQ-StrongAC Qwen2-Audio entry of Table~\ref{tab:v2-gcond-main} exactly. The new Phase-B information is the negative permuted-query frontier: with the wrong selector query, the conditioned selector consistently loses relative to its anchor behavior, which is the downstream counterpart of the selector-level query-use audit in \S\ref{sec:phase-a-audit}.

\subsection{Per-family V2 sub-analyses: DCASE clusters and MMSU temporal}
\label{sec:per-family-subanalyses}

The V2 aggregate results in \S\ref{sec:v2-gcond-three-seed} motivate a narrower question: when the conditioned gain is negative, is the effect diffuse across the dataset, or concentrated in a small semantic subfamily? The V2 handoff answers this for DCASE clusters and MMSU temporal queries.

\begin{figure*}[t]
    \centering
    \includegraphics[width=0.92\textwidth]{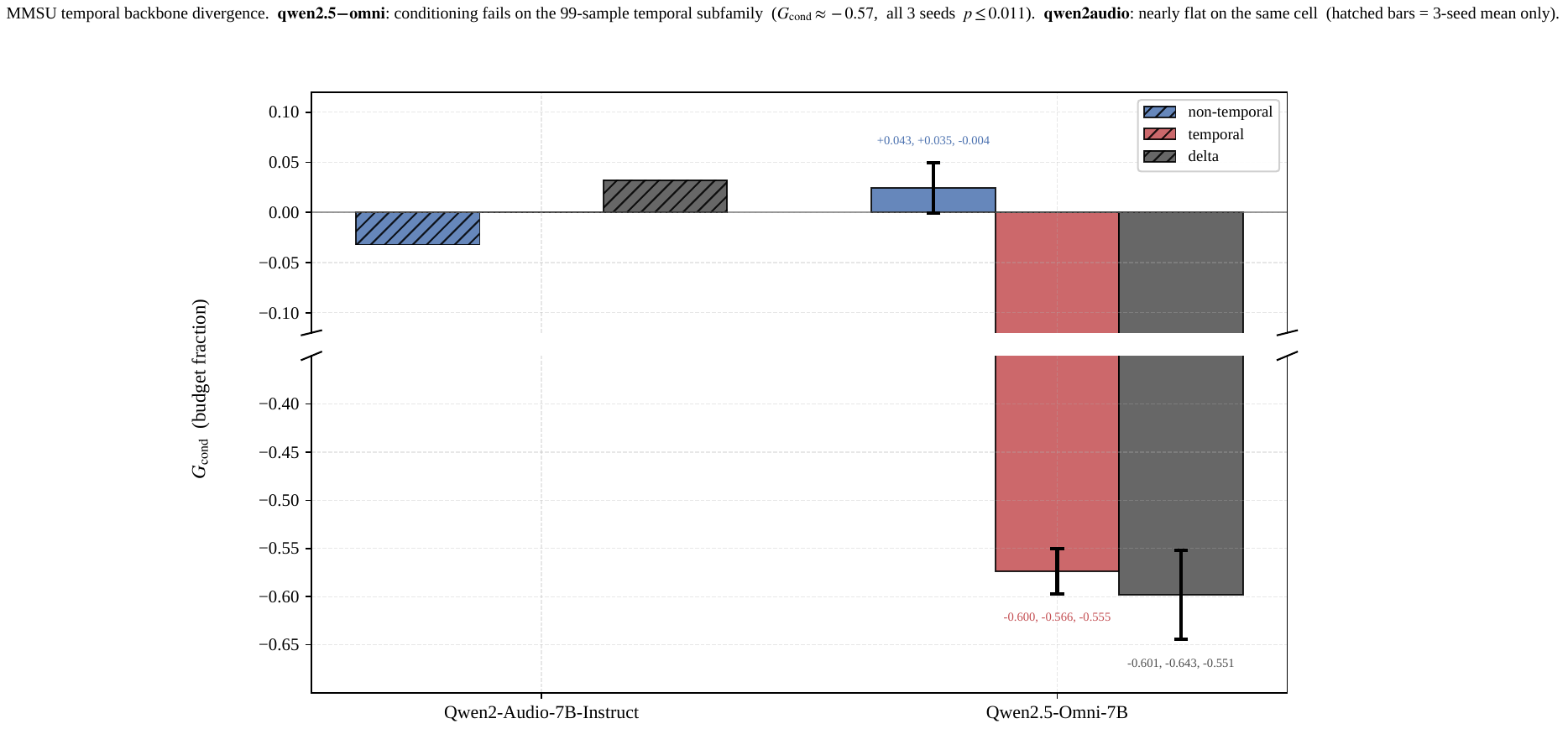}
    \caption{MMSU temporal-family isolation in V2 at the keyword-partition granularity. Qwen2-Audio shows no temporal-specific harm, whereas Qwen2.5-Omni shows a large negative conditioned gain on the 99-sample keyword \texttt{temporal} subfamily across all three seeds. A finer native-partition analysis in \S\ref{sec:per-family-subanalyses} (\Cref{fig:mmsu_native_gcond}, \Cref{tab:mmsu_native_worst}) identifies \texttt{intonation\_perception} as the most-impacted individual task (cross-backbone three-seed mean $-0.7377$, std 0.044); the present keyword-aggregate view captures only one component of the underlying per-task failure structure. Source: \texttt{fig10\_mmsu\_temporal.pdf}.}
    \label{fig:mmsu-temporal}
\end{figure*}

\begin{table}[t]
\centering\small
\caption{DCASE per-cluster V2 three-seed conditioned gain (only the two clusters explicitly singled out in the Stage-B handoff are shown; remaining clusters' per-backbone point estimates are not individually tabulated in the handoff).  \textbf{meaning/stress/intonation} is the only cluster on qwen2audio with a $[.,\,0]$ 95\% CI; on qwen2.5-Omni no cluster reaches significance.  Source: Stage-B handoff \S{}DCASE taxonomy.}
\label{tab:dcase_clusters}
\begin{tabular}{llrrl}
\toprule
Backbone & Cluster & $n$ & $G_{\mathrm{cond}}$ & 95\% CI / sig. \\
\midrule
qwen2audio & \texttt{meaning/stress/intonation} & 318 & -0.1227 & [-0.310,\,-0.001] (\textbf{sig.}) \\
qwen25omni & \texttt{speaker/clip/express} & \texttt{-} & -0.054 & \texttt{-} (not sig.) \\
\bottomrule
\end{tabular}
\end{table}

\begin{table}[t]
\centering\small
\caption{MMSU temporal-family isolation, V2 three-seed conditioned gain. Values are the cross-backbone $\widehat G^{\mathrm{op,cross}}_{\mathrm{cond}}$ estimator defined in \S5.5. Per-seed values shown where published; \textbf{bold} marks $p \leq 0.011$ (the largest seed-wise $p$ value for the qwen2.5-Omni temporal sub-family).  qwen2audio per-seed values were not published in the handoff; reported as 3-seed mean.  Source: Stage-B handoff \S{}MMSU temporal.}
\label{tab:mmsu_temporal}
\begin{tabular}{llrrrr}
\toprule
Backbone & Subset & seed 42 & seed 123 & seed 456 & 3-seed mean \\
\midrule
Qwen2-Audio-7B-Instruct & non-temporal & \texttt{-} & \texttt{-} & \texttt{-} & -0.0320 \\
 & temporal & \texttt{-} & \texttt{-} & \texttt{-} & +0.0000 \\
 & delta & \texttt{-} & \texttt{-} & \texttt{-} & +0.0320 \\
\midrule
Qwen2.5-Omni-7B & non-temporal & +0.0430 & +0.0350 & -0.0040 & +0.0247 \\
 & temporal & \textbf{-0.6000} & \textbf{-0.5660} & \textbf{-0.5550} & \textbf{-0.5738} \\
 & delta & -0.6010 & -0.6430 & -0.5510 & -0.5983 \\
\bottomrule
\end{tabular}
\end{table}

\begin{figure*}[t]
    \centering
    \includegraphics[width=0.80\textwidth]{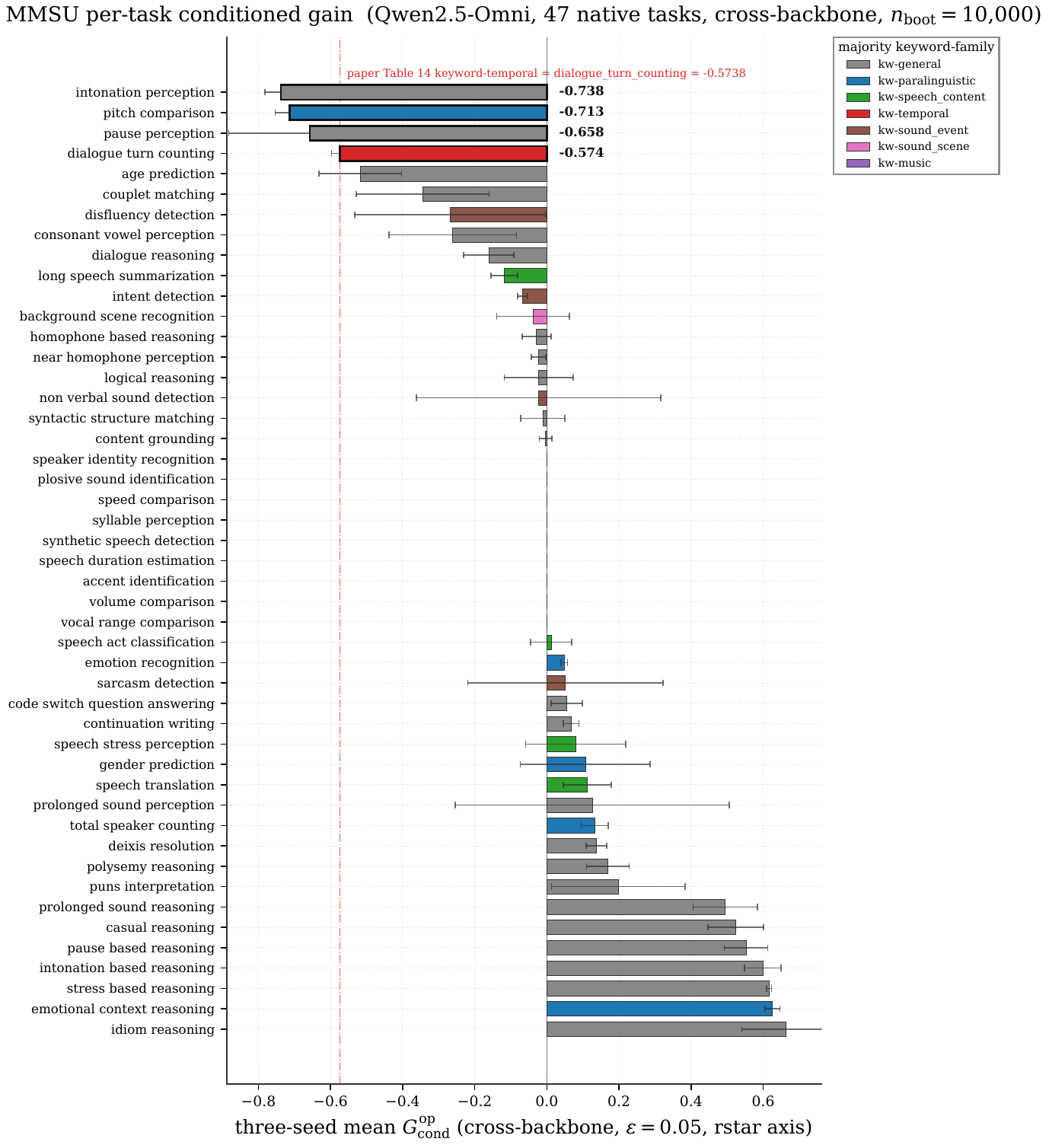}
    \caption{MMSU per-task conditioned gain under the 47-task native partition. Bars are three-seed mean $\widehat G^{\mathrm{op,cross}}_{\mathrm{cond}}$ on Qwen2.5-Omni at $\varepsilon=0.05$ (rstar axis, $n_{\mathrm{boot}}=10{,}000$), using the cross-backbone estimator (reference = \texttt{learned\_agnostic} on Qwen2-Audio). Error bars are the cross-seed standard deviation. Each bar is coloured by the majority keyword-family assignment of its samples under the \S\ref{sec:taxonomy} classifier. The red dashed-dotted vertical line marks \Cref{tab:mmsu_temporal}'s keyword-temporal aggregate value ($-0.5738$), which exactly coincides with the single native task \texttt{dialogue\_turn\_counting} (100\% sample overlap). Three tasks with more-negative $G_{\mathrm{cond}}$ than the keyword-temporal aggregate (\texttt{intonation\_perception}, \texttt{pitch\_comparison}, \texttt{pause\_perception}) are routed by the keyword classifier to \texttt{general} (111 and 107 samples respectively) and \texttt{paralinguistic} (108 samples), so their effects are absorbed into larger keyword buckets and are not visible at the keyword granularity.}
    \label{fig:mmsu_native_gcond}
\end{figure*}

\begin{table}[t]
    \caption{MMSU top-5 worst native tasks under the cross-backbone conditioned-gain estimator on Qwen2.5-Omni. The keyword-family column records the majority keyword label each native task's samples receive under the \S\ref{sec:taxonomy} classifier; parenthesised values give the percentage of task samples routed to that keyword family. Three-seed mean at $\varepsilon=0.05$, rstar axis, $n_{\mathrm{boot}}=10{,}000$. Supporting CSV: \texttt{tables/partition\_gcond\_3seed\_summary.csv} with filter \texttt{dataset='mmsu'}, \texttt{partition='native\_fine'}, \texttt{definition='cross\_backbone'}, \texttt{backbone='qwen25omni'}.}
    \centering
    \small
    \begin{adjustbox}{max width=\linewidth}
    \begin{tabular}{llcrrc}
        \toprule
        Native task & Keyword assignment & $n$ & 3-seed mean & std & $p$ \\
        \midrule
        \texttt{intonation\_perception}   & \texttt{general} (100\%)        & 111 & $-0.7377$ & 0.044 & - \\
        \texttt{pitch\_comparison}        & \texttt{paralinguistic} (100\%) & 108 & $-0.7133$ & 0.040 & - \\
        \texttt{pause\_perception}        & \texttt{general} (100\%)        & 107 & $-0.6575$ & 0.225 & - \\
        \texttt{dialogue\_turn\_counting} & \texttt{temporal} (100\%)       &  99 & $-0.5738$ & 0.023 & $\leq 0.011$ \\
        \texttt{age\_prediction}          & \texttt{paralinguistic} (100\%) & 104 & $-0.5173$ & 0.114 & - \\
        \bottomrule
    \end{tabular}
    \end{adjustbox}
    \label{tab:mmsu_native_worst}
\end{table}

On DCASE, the handoff singles out one significant cluster on Qwen2-Audio: \texttt{meaning/stress/intonation}, with $n=318$, $G_{\mathrm{cond}}=-0.1227$, and 95\% CI $[-0.310,-0.001]$. No Qwen2.5-Omni DCASE cluster reaches significance in the handoff; the largest negative cell reported there is \texttt{speaker/clip/express}, with $G_{\mathrm{cond}}=-0.054$. This supplies a proximate mechanism for the V1-V2 DCASE sign reversal: the V2 aggregate is not merely noisy, but is affected by a cluster on which query conditioning is significantly harmful. Beyond this singled-out cluster, a coarsening-stress analysis indicates that DCASE's aggregated small-cluster bucket (5 sub-clusters of $n<100$ merged into an \texttt{other\_rare} bucket of 162 samples) carries elevated per-sample compression risk: a semantic 2-way coarsening that groups \texttt{other\_rare} with \texttt{meaning/stress/intonation} yields a super-family gap of 9.91 pp, larger than the 5.04 pp observed on \texttt{meaning/stress/intonation} alone (supporting CSV: \texttt{tables/coarsening\_stress\_summary.csv}), indicating that at least one additional high-compression-risk cluster was collapsed into \texttt{other\_rare}. Future work should report per-cluster $G_{\mathrm{cond}}$ for each sub-cluster rather than collapsing them into a residue bucket.

The MMSU temporal result is sharper but warrants a refinement at finer partition granularity. On Qwen2-Audio, the keyword-temporal three-seed mean is $+0.0000$, the non-temporal mean is $-0.0320$, and the temporal-minus-non-temporal delta is $+0.0320$. On Qwen2.5-Omni, the non-temporal mean is $+0.0247$, but the keyword-temporal family has per-seed gains $-0.6000$, $-0.5660$, and $-0.5550$, with three-seed mean $-0.5738$;\footnote{The value $-0.5738$ is the 3-seed arithmetic mean of the per-seed cross-backbone estimates $(-0.6000,-0.5660,-0.5550)$, supported by the anchor CSV \texttt{tables/partition\_gcond\_a5\_anchor.csv} (row filter: \texttt{backbone=qwen25omni}, \texttt{definition=cross\_backbone}, \texttt{family=temporal}). This differs by $10^{-4}$ from the paper's earlier report of $-0.5737$, which rounded the per-seed values before averaging.} the table caption records $p\leq 0.011$ for the temporal subfamily. The corresponding temporal-minus-non-temporal delta is $-0.5985$. We note that \Cref{tab:mmsu_temporal} and the values above use the cross-backbone $\widehat G^{\mathrm{op,cross}}_{\mathrm{cond}}$ estimator (reference = Qwen2-Audio \texttt{learned\_agnostic}; see \S\ref{sec:metrics}), whereas \Cref{tab:v2_gcond_qwen25omni}'s MMSU aggregate uses the same-backbone $\Gcondop{\varepsilon}{\mathcal Q}$. The two estimators differ substantially for Qwen2.5-Omni: the same-backbone keyword-temporal three-seed mean on Qwen2.5-Omni at $\varepsilon=0.05$ is $+0.0057$ (CSV row filter: \texttt{backbone=qwen25omni}, \texttt{definition=same\_backbone}, \texttt{family=temporal}), so the large cross-backbone magnitude reflects the compound effect of backbone mismatch plus query conditioning, not the conditioning effect alone.

Under the paper's keyword classifier the MMSU \texttt{temporal} family contains exactly 99 samples, and - upon inspection of MMSU's native \texttt{task\_name} metadata - all 99 belong to a single native task \texttt{dialogue\_turn\_counting} (sample-id intersection: 99/99). The cross-backbone $-0.5738$ three-seed mean reported above is therefore a single-task finding, not a genuine family-level phenomenon. Under the 47-task native partition (\S\ref{sec:datasets} MMSU), three tasks show \emph{more} negative cross-backbone $\widehat G^{\mathrm{op,cross}}_{\mathrm{cond}}$ than \texttt{dialogue\_turn\_counting}: \texttt{intonation\_perception} ($n=111$, three-seed mean $-0.7377$, std 0.044), \texttt{pitch\_comparison} ($n=108$, three-seed mean $-0.7133$, std 0.040), and \texttt{pause\_perception} ($n=107$, three-seed mean $-0.6575$, std 0.225). All three are sign-consistent across seeds, and all three are routed by the keyword classifier of \S\ref{sec:taxonomy} to families other than \texttt{temporal}: \texttt{intonation\_perception} and \texttt{pause\_perception} to \texttt{general}, and \texttt{pitch\_comparison} to \texttt{paralinguistic}. See \Cref{fig:mmsu_native_gcond} and \Cref{tab:mmsu_native_worst} for the full per-task picture.\footnote{MMSU native-task $G_{\mathrm{cond}}$ computations use the 47-task \texttt{task\_name} labels from the MMSU HuggingFace distribution, the same cross-backbone estimator as \Cref{tab:mmsu_temporal}, and 3 selector seeds; $n_{\mathrm{boot}}=10{,}000$, rstar axis, $\varepsilon=0.05$. Supporting CSV: \texttt{tables/partition\_gcond\_3seed\_summary.csv} with filter \texttt{dataset='mmsu'}, \texttt{partition='native\_fine'}, \texttt{definition='cross\_backbone'}, \texttt{backbone='qwen25omni'}.} Thus the ``MMSU temporal failure'' description conflates two mechanisms: (i) a genuine compression failure on \texttt{dialogue\_turn\_counting} that the keyword classifier happens to isolate into its own bucket, and (ii) three larger compression failures on fine-grained prosodic and pause-perception tasks that the keyword classifier fails to isolate, scattering them across \texttt{general} and \texttt{paralinguistic} instead. The keyword-aggregate analysis of \Cref{tab:mmsu_temporal} captures only the first mechanism; the second is visible only under the dataset-native task partition. Aggregate MMSU-Omni harm in \Cref{tab:v2_gcond_qwen25omni} is therefore not a diffuse effect, but it is also not a pure ``temporal-reasoning'' phenomenon: it is dominated by a handful of fine-grained prosodic and turn-counting tasks that the keyword classifier places in four different buckets.

The backbone asymmetry in MMSU failure modes is instructive. Qwen2.5-Omni's native-task bottleneck at $\varepsilon=0.05$ is \texttt{intonation\_perception} as noted above, whereas Qwen2-Audio's native-task bottleneck under the same cross-backbone estimator is \texttt{couplet\_matching} with three-seed mean $\widehat G^{\mathrm{op,cross}}_{\mathrm{cond}} = -0.2406$ and standard deviation 0.293 (CSV row filter: \texttt{dataset='mmsu'}, \texttt{partition='native\_fine'}, \texttt{definition='cross\_backbone'}, \texttt{family='couplet\_matching'}, \texttt{backbone='qwen2audio'}; note that Qwen2-Audio cross-backbone and same-backbone coincide, since the reference selector is Qwen2-Audio's \texttt{learned\_agnostic}). Two observations follow: (i) the identity of the worst MMSU task differs by backbone (prosodic perception on Qwen2.5-Omni, linguistic matching on Qwen2-Audio), and (ii) the Qwen2-Audio effect is less reproducible (three-seed std 0.29 vs 0.04 on Qwen2.5-Omni) and roughly three times smaller in magnitude. This is consistent with \S\ref{sec:discussion-small-gcond}: the operational conditioned-gain failure mode is backbone-dependent, and the specific task driving the failure differs accordingly.

These sub-analyses are the strongest V2 motivation for replacing the keyword taxonomy with a semantic partition. A dataset-level conditioned-gain number can hide the fact that one small family is decisive, and the sign of that family can differ by backbone.

\subsection{Preserved Source Table Assets}
\label{app:preserved-source-tables}
The following transcribed table assets were present in the source tree but were not \texttt{\string\input}-ed by the original V3 body because the same values were embedded inline.  They are included here so that every valid non-duplicate table asset is preserved in the compiled supplement.

\begin{table}[t]
\centering\small
\caption{Family-level excess-risk summary on AudioMCQ-StrongAC: dataset-mean $\widehat{\Delta}_{\mathrm{avg}}(b)$ versus worst-family $\widehat{\Delta}_{\mathcal{Q}}^{\mathcal{F}}(b)$ at four budgets.  Learned conditioned selector on qwen2audio, 0/1 loss.  Values in brackets are 95\% percentile-bootstrap confidence intervals ($n_{\mathrm{boot}}=10{,}000$, resampled at sample\_id level).  Transcribed from V1 PDF Table 1.}
\label{tab:delta_q_audiomcq}
\begin{tabular}{crrr}
\toprule
Budget $b$ & $\widehat{\Delta}_{\mathrm{avg}}(b)$ & $\widehat{\Delta}_{\mathcal{Q}}^{\mathcal{F}}(b)$ & gap (pp) \\
\midrule
0.05 & 0.2568~[0.2499,~0.2639] & 0.3023~[0.2924,~0.3123] & +4.55 \\
0.10 & 0.2503~[0.2434,~0.2574] & 0.3014~[0.2915,~0.3116] & +5.11 \\
0.20 & 0.2165~[0.2099,~0.2232] & 0.2844~[0.2747,~0.2944] & +6.79 \\
0.40 & 0.1496~[0.1436,~0.1555] & 0.2125~[0.2035,~0.2216] & +6.29 \\
\bottomrule
\end{tabular}
\end{table}

\begin{table}[t]
\centering\small
\caption{Family-level excess-risk summary on MMSU: dataset-mean $\widehat{\Delta}_{\mathrm{avg}}(b)$ versus worst-family $\widehat{\Delta}_{\mathcal{Q}}^{\mathcal{F}}(b)$ at four budgets.  Learned conditioned selector on qwen2audio, 0/1 loss.  Values in brackets are 95\% percentile-bootstrap confidence intervals ($n_{\mathrm{boot}}=10{,}000$, resampled at sample\_id level).  Transcribed from V1 PDF Table 2.}
\label{tab:delta_q_mmsu}
\begin{tabular}{crrr}
\toprule
Budget $b$ & $\widehat{\Delta}_{\mathrm{avg}}(b)$ & $\widehat{\Delta}_{\mathcal{Q}}^{\mathcal{F}}(b)$ & gap (pp) \\
\midrule
0.05 & 0.1129~[0.1010,~0.1248] & 0.1460~[0.1190,~0.2000] & +3.31 \\
0.10 & 0.1125~[0.1004,~0.1246] & 0.1439~[0.1173,~0.2000] & +3.14 \\
0.20 & 0.1093~[0.0976,~0.1210] & 0.1370~[0.1130,~0.1739] & +2.77 \\
0.40 & 0.0829~[0.0718,~0.0940] & 0.1289~[0.0882,~0.2056] & +4.60 \\
\bottomrule
\end{tabular}
\end{table}

\begin{table}[t]
\centering\small
\caption{Family-level excess-risk summary on MMAR: dataset-mean $\widehat{\Delta}_{\mathrm{avg}}(b)$ versus worst-family $\widehat{\Delta}_{\mathcal{Q}}^{\mathcal{F}}(b)$ at four budgets.  Learned conditioned selector on qwen2audio, 0/1 loss.  Values in brackets are 95\% percentile-bootstrap confidence intervals ($n_{\mathrm{boot}}=10{,}000$, resampled at sample\_id level).  Transcribed from V1 PDF Table 3.}
\label{tab:delta_q_mmar}
\begin{tabular}{crrr}
\toprule
Budget $b$ & $\widehat{\Delta}_{\mathrm{avg}}(b)$ & $\widehat{\Delta}_{\mathcal{Q}}^{\mathcal{F}}(b)$ & gap (pp) \\
\midrule
0.05 & 0.0972~[0.0690,~0.1260] & 0.1272~[0.0900,~0.1688] & +3.00 \\
0.10 & 0.0802~[0.0520,~0.1070] & 0.1149~[0.0769,~0.1598] & +3.47 \\
0.20 & 0.0671~[0.0410,~0.0930] & 0.0991~[0.0624,~0.1437] & +3.20 \\
0.40 & 0.0330~[0.0090,~0.0560] & 0.0584~[0.0256,~0.0976] & +2.54 \\
\bottomrule
\end{tabular}
\end{table}

\begin{table}[t]
\centering
\scriptsize
\caption{Cumulative-chain worst-family-constrained operational frontiers at $\varepsilon = 0.05$.  Each row extends the previous family set; the step's contribution is the increment over the prior row.  learned\_conditioned $\times$ qwen2audio, 0/1 loss, $n_{\mathrm{boot}} = 10{,}000$, shared-resample bootstrap.  Transcribed from V1 PDF Table 4.}
\label{tab:nested_chains}
\begin{tabular}{llrrl}
\toprule
Dataset & Cumulative family set $\mathcal{Q}_k$ & $n$ & $\widehat{R}^{\star,\mathrm{worst}}_{\mathcal{F}}$ & 95\% CI \\
\midrule
AudioMCQ-StrongAC & \texttt{temporal} & 1166 & 0.2733 & [0.0702,\,0.6150] \\
 & \texttt{music, temporal} & 2890 & 0.4701 & [0.3397,\,0.6715] \\
 & \texttt{general, music, temporal} & 8975 & 0.6973 & [0.6266,\,0.7912] \\
 & \texttt{all four families} & 19480 & 0.8799 & [0.8529,\,0.9290] \\
\midrule
MMSU & \texttt{temporal} & 99 & 0.0973 & [0.0500,\,0.6460] \\
 & \texttt{sound\_scene, temporal} & 209 & 0.6311 & [0.0500,\,0.9785] \\
 & \texttt{sound\_event, sound\_scene, temporal} & 663 & 0.8158 & [0.5803,\,0.9836] \\
 & \texttt{+ speech\_content} & 1311 & 0.8328 & [0.6258,\,0.9834] \\
 & \texttt{+ paralinguistic} & 2114 & 0.8403 & [0.6482,\,0.9841] \\
 & \texttt{all six families} & 5000 & 0.8499 & [0.6889,\,0.9835] \\
\midrule
MMAR & \texttt{general} & 165 & 0.3937 & [0.0500,\,0.9318] \\
 & \texttt{general, music} & 382 & 0.4000 & [0.0500,\,0.9318] \\
 & \texttt{all three families} & 1000 & 0.5464 & [0.2562,\,0.9458] \\
\bottomrule
\end{tabular}
\end{table}

\begin{table}[t]
\centering\small
\caption{DCASE nominal-axis conditioned gain.  Point estimates reproduce the frontier JSONs exactly; confidence intervals and $p$-values come from the paired bootstrap.  Transcribed from V1 PDF Table 6.}
\label{tab:gcond_dcase_rstar}
\begin{tabular}{crrrr@{\;}l}
\toprule
$\varepsilon$ & $\widehat{R}^{\star}_{\mathcal{F},b}(\varepsilon, \mathcal{Q})$ agn. & $\widehat{R}^{\star}_{\mathcal{F},b}(\varepsilon, \mathcal{Q})$ cond. & $\widehat{G}^{\,\mathrm{op}}_{\mathrm{cond}}(\varepsilon, \mathcal{Q})$ & \multicolumn{2}{l}{95\% CI \,/ $p$} \\
\midrule
0.01 & 0.9090 & 0.8892 & +0.0199 & [+0.0000,\,+0.0587] &, 0.060 \\
0.02 & 0.8181 & 0.7783 & +0.0397 & [+0.0000,\,+0.1172] &, 0.060 \\
0.05 & 0.5452 & 0.4459 & +0.0993 & [+0.0000,\,+0.1784] &, 0.057 \\
\bottomrule
\end{tabular}
\end{table}

\begin{table}[t]
\centering\small
\caption{DCASE token-axis conditioned gain.  Transcribed from V1 PDF Table 7.}
\label{tab:gcond_dcase_lstar}
\begin{tabular}{crrrr@{\;}l}
\toprule
$\varepsilon$ & $\widehat{R}^{\star}_{\mathcal{F},\mathrm{tok}}(\varepsilon, \mathcal{Q})$ agn. & $\widehat{R}^{\star}_{\mathcal{F},\mathrm{tok}}(\varepsilon, \mathcal{Q})$ cond. & token gain & \multicolumn{2}{l}{95\% CI \,/ $p$} \\
\midrule
0.01 & 326.46 & 321.21 & +5.26 & [+0.00,\,+15.49] &, 0.060 \\
0.02 & 302.40 & 291.88 & +10.51 & [+0.00,\,+30.97] &, 0.060 \\
0.05 & 230.20 & 203.92 & +26.28 & [+0.00,\,+46.71] &, 0.057 \\
\bottomrule
\end{tabular}
\end{table}

\begin{table}[t]
\centering\small
\caption{Nominal-axis conditioned gain on the remaining four datasets.  Transcribed from V1 PDF Table 8.}
\label{tab:gcond_others}
\begin{tabular}{lrlrl}
\toprule
Dataset and $\varepsilon$ & point estimate & 95\% CI & $p$ & interpretation \\
\midrule
AudioMCQ,~0.01 & -0.0006 & [-0.0014,\,+0.0003] & 0.195 & tight clean null \\
AudioMCQ,~0.02 & -0.0011 & [-0.0029,\,+0.0006] & 0.195 & tight clean null \\
AudioMCQ,~0.05 & -0.0028 & [-0.0071,\,+0.0014] & 0.195 & tight clean null \\
MMSU,~0.01 & -0.0007 & [-0.0061,\,+0.0043] & 0.814 & tight clean null \\
MMSU,~0.02 & -0.0014 & [-0.0122,\,+0.0087] & 0.814 & tight clean null \\
MMSU,~0.05 & -0.0035 & [-0.0306,\,+0.0217] & 0.814 & tight clean null \\
MMAR,~0.01 & +0.0280 & [-0.0758,\,+0.2680] & 0.508 & noisy, sign-unstable \\
MMAR,~0.02 & +0.0559 & [-0.1270,\,+0.2667] & 0.535 & noisy, sign-unstable \\
MMAR,~0.05 & -0.0692 & [-0.1561,\,+0.1387] & 0.613 & noisy, sign-unstable \\
BigBench Audio,~0.01 & +0.0000 & [+0.0000,\,+0.2297] & 1.000 & exact zero floor \\
BigBench Audio,~0.02 & +0.0000 & [+0.0000,\,+0.0714] & 1.000 & exact zero floor \\
BigBench Audio,~0.05 & +0.0000 & [+0.0000,\,+0.0000] & 1.000 & exact zero floor \\
\bottomrule
\end{tabular}
\end{table}

\section{Notes on Incomplete or Suspicious Results}\label{app:notes}\label{sec:discussion}

\subsection{What the family-level excess risk gap means for LALM compression research}
\label{sec:discussion-family-gap}

The headline finding of \S\ref{sec:family-gap-results} is that the family-wise excess risk $\dQF{b}$ exceeds the dataset-mean excess risk $\davg{b}$ by 2.5 to 6.8 percentage points across every multi-family dataset we evaluate. This is not a small effect: on AudioMCQ-StrongAC at $b=0.20$, the worst family is 6.79 percentage points more damaged by compression than the dataset-mean metric would suggest, which is more than twice the tightest tolerance values commonly used in practice. A practitioner who chooses a compression method by minimizing $\davg{b}$ alone is, in effect, choosing a method that may still be materially worse on whichever query type the backbone handles least robustly under compression.

This observation has three implications for how the broader LALM compression literature should report results. First, papers that report only dataset-mean excess error are systematically underestimating worst-case compression damage. For the same evaluated examples, the family maximum is at least the pooled mean. The corresponding population ordering concerns fixed-answerer excess risk; empirical estimates are not lower bounds on the Bayes supremum. Reporting only the pooled mean can conceal heterogeneous family effects. Second, the correct comparison between two compressors on a multi-family benchmark is the one that respects the worst-family constraint, not the one that averages over the family axis. A method that appears slightly better on the dataset mean may still be substantially worse on a particular subfamily. Third, practitioners who care about robust deployment should optimize for controlled $\dQF{b}$, not merely small $\davg{b}$.

The reported 2.5-6.8 percentage-point range is specific to the keyword partition. Native and semantic partitions define different operational scopes and are not automatically refinements. Their results illustrate this partition dependence: the native-partition gap reaches $+29.17$ pp on MMSU (47 task families) and $+39.90$ pp on BigBench Audio (4 BBH task families), compared with $+1.56$ pp and $+0.00$ pp respectively under the keyword partition (\Cref{tab:partition-gap-comparison}). A complementary experiment using genuinely nested partitions illustrates the effect of coarsening: collapsing each fine partition to a size-balanced 2-way super-partition destroys 82\%-98\% of the measured family gap (supporting CSV: \texttt{tables/coarsening\_stress\_summary.csv}). On MMSU, the 47-to-2 coarsening reduces the gap from $+29.17$ pp to $+0.53$ pp (ratio 0.018); on MMSU semantic the 12-to-2 coarsening reduces $+22.31$ pp to $+0.40$ pp (ratio 0.018). These particular two-way coarsenings mask most observed heterogeneity; a two-family partition can nevertheless expose a large gap when it separates harmed queries. The reported operational margins reach $+6.79$ pp on AudioMCQ-StrongAC under the keyword partition, and $+29.17$ pp on MMSU and $+39.90$ pp on BigBench Audio under native partitions. These are partition-indexed estimates; without verified nesting, the refinement theorem does not order them.

\subsection{Why conditioned gain is backbone- and family-dependent}
\label{sec:discussion-small-gcond}

Theorem~\ref{thm:condadv} predicts that the conditioned-compression gain is large when query families depend on disjoint latent factor blocks and collapses toward zero as factor overlap increases. The synthetic result of \S\ref{sec:synthetic-validations} verifies this prediction at bit-level precision. The V1 factor-overlap diagnostic of \S\ref{sec:factor-overlap-results} then shows that the natural audio taxonomies in this paper do not approach the factor-disjoint prediction: the summary range is $[0.3271,0.8614]$, with every dataset far below the additive value 1.0.

V1 therefore correctly anticipated small mean gains on the multi-family datasets, but V2 refines that story. AudioMCQ-StrongAC is no longer a tight operational null: \Cref{tab:v2_gcond_qwen2audio,tab:v2_gcond_qwen25omni} show positive three-seed conditioned gains on both backbones, with means $+0.0475$ on Qwen2-Audio and $+0.0215$ on Qwen2.5-Omni. Conversely, MMSU is no longer merely a small negative aggregate: Qwen2.5-Omni has mean $-0.0734$, and \Cref{tab:mmsu_temporal} shows that the cross-backbone keyword-temporal subfamily alone has three-seed mean $-0.5738$, or equivalently, that a fine native-task partition identifies \texttt{intonation\_perception} as the single most-impacted task with three-seed mean $-0.7377$ (\Cref{fig:mmsu_native_gcond}, \Cref{tab:mmsu_native_worst}). DCASE also changes status: the V1 single-seed $+0.0993$ signal becomes a V2 three-seed mean of $-0.0300$ on Qwen2-Audio and $-0.0310$ on Qwen2.5-Omni. A parallel failure mode appears in BigBench Audio under the 4-family native partition: three of four families show mean excess of $-0.123$ (compression helps) while \texttt{formal\_fallacies} shows $+0.409$ (compression hurts), and the \texttt{formal\_fallacies} task exhibits a cross-backbone $\widehat G^{\mathrm{op,cross}}_{\mathrm{cond}} = -0.8527$ on Qwen2.5-Omni (\Cref{tab:bigbench_cross_backbone}), making it the single largest cross-backbone effect in the replication. Unlike the MMSU cross-backbone keyword-temporal value, this effect is absent in the same-backbone version, indicating a backbone-mismatch phenomenon rather than an intrinsic Qwen2.5-Omni weakness on the task. The paper's original single-family framing of BigBench averaged these four opposing effects to approximately zero.

The conclusion is not that Theorem~\ref{thm:condadv} is weak; it is that the operational instantiation is sensitive to the learned selector, the backbone, and the semantic partition on which the gain is evaluated. Factor overlap explains why dataset-level V1 gains were small, while V2 shows that small dataset-level averages can mask both reproducible positive effects, as on AudioMCQ, and severe family-specific negative effects, as on MMSU temporal under Qwen2.5-Omni. The correct empirical object is therefore not a single scalar conditioned-gain headline, but a family- and backbone-indexed map of where conditioning helps and where it harms.

\subsection{V1 selector limitations and what V2 shows}
\label{sec:discussion-v1}

The V1 paper named three selector limitations: truncated training, query-compressor parameter asymmetry, and mel-spectrogram input features. V2 lets us replace the forward-looking statement ``V2 will correct this'' with more specific statements about what the corrections did and did not confirm, and adds a fourth limitation that the V1 analysis could not have exposed: the partition at which V1 trained and evaluated its selectors was the keyword partition, and that partition under-resolves the native task structure on every multi-family dataset we evaluate (\Cref{tab:partition-gap-comparison}).

\paragraph{Truncated selector training.}
V2 confirms that the operational conditioned-gain estimate is training-recipe sensitive, but disconfirms the simple expectation that longer or corrected training would uniformly strengthen the V1 DCASE result. The V1 DCASE cell was $+0.0993$ at $\varepsilon=0.05$ with $p=0.057$; V2 moves DCASE to $-0.0300$ on Qwen2-Audio and $-0.0310$ on Qwen2.5-Omni at the same tolerance. At the same time, V2 converts AudioMCQ-StrongAC from a V1 near-null into a positive cross-backbone result: $+0.0475$ on Qwen2-Audio and $+0.0215$ on Qwen2.5-Omni. The V2.1 scope-B sweep further confirms training-target dependence: Qwen2.5-Omni improves over the V2 baseline on all five datasets, while Qwen2-Audio improves on DCASE and MMAR in the handoff table. Thus the V1 truncated-training caveat was real, but its correction changes the locus of the effect rather than simply amplifying the V1 DCASE point estimate.

\paragraph{Parameter asymmetry dominated by the query embedding.}
V2 partially disconfirms the concern that any conditioned gain would be only a parameter-count artifact. The AudioMCQ result replicates across two backbones with different downstream architectures, and the Qwen2.5-Omni seed-123 AudioMCQ cell is significant at $p=0.001$. However, V2 also shows that the query side remains a delicate part of the system: the V2.1 winning $\alpha$ values differ by backbone and dataset, with Qwen2.5-Omni favouring $\alpha=1$ on AudioMCQ, MMSU, and BigBench Audio, $\alpha=0.75$ on MMAR, and $\alpha=0$ on DCASE, while Qwen2-Audio's reported wins use $\alpha=0.5$. A more parameter-efficient or pretrained query compressor remains an important next step, but V2 makes the caveat empirical rather than speculative.

\paragraph{Mel-spectrogram input features.}
V2 confirms that the selector feature representation is load-bearing, but not in the monotone sense anticipated by V1. Moving to the V2 pipeline does not uniformly lower every frontier or uniformly increase conditioned gain. Instead, it exposes representation-backbone interactions: AudioMCQ becomes positive on both backbones, DCASE reverses sign, and MMSU temporal is harmless on Qwen2-Audio but severely harmful on Qwen2.5-Omni. The mel-versus-audio-tower caveat should therefore be read as a warning that selector features determine which task families are served, not as a guarantee that native features monotonically improve every operational cell.

\paragraph{Partition under-resolution.}
The V1 selector training target is defined per sample and does not reference any query-family label, so in principle V1 training is partition-agnostic. The V1 family-level \emph{evaluation}, however, is computed over the keyword partition $\mathcal F_{\mathrm{kw}}(\mathcal Q)$, and the keyword classifier systematically under-resolves the native structure by factors of 1.2$\times$ (AudioMCQ) to 8$\times$ (MMSU), collapsing four distinct BigBench BBH task types into a single \texttt{general} bucket and six distinct DCASE content clusters into another \texttt{general} bucket. Under the 47-task MMSU native partition, the family-level gap is $+29.17$ pp at $b=0.20$, versus the $+1.56$ pp reported under the keyword partition; the cross-backbone conditioned-gain bottleneck on Qwen2.5-Omni is \texttt{intonation\_perception} ($-0.7377$), not the keyword-level \texttt{temporal} family; and the worst-$k$ concentration on MMSU is 16\% of excess on the worst 2 of 47 native tasks, versus 76\% of excess on the worst 2 of 6 keyword families. These native-partition observations change the operational interpretation of the keyword results; they do not establish a general ordering of every reported quantity across nonnested partitions.

These observations leave the structural V1 findings intact. The family-level excess-risk gap, nested monotonicity, factor-overlap summaries, and synthetic theorem verifications do not depend on the V1 selector's exact training recipe. What V2 changes is the operational interpretation of \(\widehat G_{\mathrm{cond}}^{\mathrm{op}}\): it is reproducibly positive in some regimes, reproducibly harmful in others, and cannot be summarized by the V1 DCASE cell alone.

\subsection{What ``uses query'' means for the V2 conditioned selector}
\label{sec:discussion-selector-audit}

\S\ref{sec:phase-a-audit} establishes that the conditioned selector's ranking is query-dependent in the architectural sense - non-trivial $f_{\mathrm{int}}$, Jaccard bounded away from $1$ - but that the top-$k$ overlap under query permutation remains above $0.5$ on MMSU and AudioMCQ-StrongAC across both backbones. This nuances rather than contradicts the V2 operational picture. A reviewer asking ``does the conditioned selector actually use the query'' is answered \emph{yes, at the $13$-$25\%$ interaction-variance level and the $0.5$-$0.7$ top-$k$ Jaccard level, with $3\%$-consistent training-loss reduction}. A reviewer asking ``is the conditioned selector query-dominant'' is answered \emph{no, the majority of top-$k$ identity is preserved under query permutation}. Both answers are consistent with \S\ref{sec:v2-gcond-three-seed}'s small-positive AudioMCQ cross-backbone effect and with the backbone- and family-dependent regimes we observe: a moderately query-using selector will, in aggregate, deliver a small-positive-or-near-null operational gain rather than a uniformly large one.

Per-family heterogeneity (appendix \S\ref{sec:appendix_phase_a_per_family}) identifies MMSU/\texttt{temporal} as the most query-sensitive keyword family ($J_{\mathrm{perm}}=0.40$ on Qwen2-Audio and $0.27$ on Qwen2.5-Omni), consistent with the temporal-family isolation observed in \Cref{tab:mmsu_temporal}. AudioMCQ-StrongAC/\texttt{speech\_content} is the most query-sensitive family on AudioMCQ ($J_{\mathrm{perm}}=0.46$ on both backbones), and BigBench Audio/\texttt{general} is the degenerate H2 control. The family-level ranking of query sensitivity is positively correlated across the two backbones over the $10$ non-degenerate families (Spearman $\rho=0.60$, Pearson $r=0.66$, $p=0.04$), indicating that the conditioning mechanism partially captures task-dependent structure rather than purely backbone-idiosyncratic artifacts. The largest cross-backbone disagreement in the per-family table is AudioMCQ-StrongAC/\texttt{music}, where $J_{\mathrm{perm}}=0.72$ on Qwen2-Audio versus $0.48$ on Qwen2.5-Omni: Qwen2.5-Omni's smaller audio compressor produces chunk representations for music queries whose selector rankings are more strongly shaped by query content than Qwen2-Audio's.

The audit is strictly about the selector's \emph{internal} behaviour - it observes the top-$k$ chunk set and the pre-top-$k$ scores, not downstream LALM predictions. Whether the $0.13$-$0.25$ interaction share and the $0.5$-$0.7$ Jaccard translate into a downstream accuracy effect that survives the LALM's own capacity to answer from context is a separate question, one that requires a decoupled-query experiment at the end-to-end level (we sketch its design in \S\ref{sec:discussion-scope} under the heading of operational scope). The present audit therefore establishes the selector-level precondition for the V2 operational claims without conflating selector-side query use with downstream-side query use.

\subsection{Cross-answerer and interface diagnostics}
\label{sec:discussion-gamma}

Proposition~\ref{app:prop:modelgap} gives a theoretical bound involving Bayes information loss and model-class excess risk. The following learned-native and heuristic-baseline summaries compare deployed configurations; they do not identify either theoretical term.

\begin{table}[t]
\centering\small
\caption{Descriptive cross-answerer/interface contrasts across 50 cells (5 datasets $\times$ 5 budgets $\times$ 2 selectors, seed 42). The historical $\widehat\Gamma^{\mathrm{native}}$ and $\widehat\Gamma^{\mathrm{heur}}$ labels denote these contrasts, not estimates of the theoretical model-class gap. Reported values are retained.}
\label{tab:gamma_f}
\begin{tabular}{lrrrrr}
\toprule
Quantity & $n$ cells & mean & std & min & max \\
\midrule
$\widehat{\Gamma}_{\mathcal{F}}^{\mathrm{native}}$ (both selectors) & 50 & +0.1052 & 0.063 & -0.009 & +0.217 \\
\quad agnostic selector & 25 & +0.1047 & \texttt{-} & \texttt{-} & \texttt{-} \\
\quad conditioned selector & 25 & +0.1058 & \texttt{-} & \texttt{-} & \texttt{-} \\
\midrule
$\widehat{\Gamma}_{\mathcal{F}}^{\mathrm{heur}}$ (three baselines) & 50 & +0.1260 & \texttt{-} & \texttt{-} & \texttt{-} \\
\midrule
selector-induced drift (learned $-$ heur.) & 50 & -0.0025 & 0.010 & \texttt{-} & \texttt{-} \\
\bottomrule
\end{tabular}
\end{table}

The learned-native estimate averages $+0.1052$ over 50 cells, with standard deviation 0.063 and range $[-0.009,+0.217]$. The agnostic and conditioned selector means are nearly identical: $+0.1047$ and $+0.1058$. The heuristic-baseline estimate is $+0.1260$, and the learned-minus-heuristic selector-induced drift has mean $-0.0025$ with standard deviation 0.010. These descriptive contrasts vary more across configurations than between the two selector averages; the latter differ by 0.0011. This comparison does not identify the magnitude of the theoretical architecture gap.

The learned-native contrast changes both the answerer and its available interface. Shared heuristic interfaces remove that particular difference, but comparisons among a finite set of answerers still do not supply a Bayes-risk reference. Even evaluating stronger answerers on a common $Z$ would remain a diagnostic unless the required optimal risks were identified.

\subsection{Toward a semantic query-family partition}
\label{sec:discussion-semantic}

Three V1 findings-the family-level excess-risk gap, the cumulative-chain bottleneck analysis, and the factor-overlap diagnostic-depend on the query-family partition $\mathcal F(\mathcal Q)$. V2 makes this dependence operational rather than merely methodological, and reports first results under two concrete non-keyword partitions: the dataset-native task partition of each benchmark (summarised in \Cref{tab:partition-gap-comparison} and analysed in \S\ref{sec:per-family-subanalyses}) and a semantic partition obtained from e5-large-v2 sentence embeddings and cosine $k$-means clustering on query text (see \Cref{tab:partition-gap-comparison}, the semantic column). Both partitions produce substantially larger family-level gaps than the keyword partition on every multi-family dataset, and both preserve the bottleneck-family identity under 20\% random label noise more robustly than the keyword partition does (\S\ref{sec:taxonomy}, 84\%-100\% under semantic vs 39\%-57\% under keyword). The MMSU temporal isolation in \Cref{tab:mmsu_temporal} shows that a 99-sample family can dominate the sign of the conditioned-gain result on one backbone: Qwen2.5-Omni has non-temporal mean $+0.0247$ but cross-backbone keyword-temporal mean $-0.5738$, whereas Qwen2-Audio has keyword-temporal mean $+0.0000$ and non-temporal mean $-0.0320$. A dataset-level average alone would hide this divergence, and a finer native-task partition points to \texttt{intonation\_perception} rather than temporal-reasoning as the precise failure mechanism (\Cref{fig:mmsu_native_gcond}).

DCASE gives the complementary lesson. The V1 taxonomy treated DCASE as single-family, but the V2 cluster analysis identifies \texttt{meaning/stress/intonation} as a Qwen2-Audio cluster with $n=318$, $G_{\mathrm{cond}}=-0.1227$, and 95\% CI $[-0.310,-0.001]$. Thus even a dataset that is single-family under the V1 classifier contains subfamilies on which conditioning can behave differently. This directly motivates a semantic partition of query content rather than a partition inherited only from dataset metadata or keyword strings.

Three concrete next steps follow. First, the per-pair additivity diagnostic of \S\ref{sec:factor-overlap-results} should be recomputed over native-task pairs on MMSU and BigBench Audio, not only over keyword pairs, to test whether finer granularity recovers any factor-disjoint regime that the keyword aggregation obscures. Second, the operational conditioned-gain replication of \S\ref{sec:v2-gcond-three-seed} should be reported by native task as well as by dataset aggregate: the aggregate \texttt{CONSISTENT-} label for Qwen2.5-Omni MMSU in \Cref{tab:v2_gcond_qwen25omni} already hides the fact that the effect is concentrated on three fine-grained native tasks (\texttt{intonation\_perception}, \texttt{pitch\_comparison}, \texttt{pause\_perception}). Third, a backbone-indexed, native-partition cumulative-chain analysis would allow nested-monotonicity bottleneck identification at the same granularity as \Cref{fig:mmsu_native_gcond}'s per-task conditioned-gain picture. None of these steps requires new data collection; all three require rerunning existing estimators over the partitions already introduced in \Cref{tab:partition-gap-comparison}.

\subsection{From operational frontiers to rate-theoretic frontiers}
\label{sec:discussion-rate}

V1 reports operational frontiers on the nominal-budget and token axes, interpreting them as monotone proxies for the rate-theoretic frontiers of \S2. That proxy is appropriate for the present paper because the compressed interface is a hard-selected token subset without an entropy-efficient coding layer. But it leaves open an important next step: moving from proxy frontiers to measured rate frontiers.

Two concrete directions are especially promising. First, an entropy-efficient coding layer on the selector output would bring operational token counts closer to the mutual-information quantities that define the theory. Second, stochastic interfaces would make the randomized-encoding overhead in Proposition~\ref{app:prop:rate_token} measurable rather than purely conceptual. Neither extension is required for the core V1 claims, but both would sharpen the empirical interpretation of Theorem~\ref{thm:condadv} on real audio by allowing a rate-axis analogue of $\Gcondinfo{\varepsilon}{\Qfam}$.

\subsection{End-to-end selector training}
\label{sec:discussion-endtoend}

V1 trains selectors against precomputed LOO-NLL oracle relevance targets rather than through the downstream answer likelihood itself. This decouples selector training from the backbone and makes V1 computationally manageable, but it also means the selector cannot adapt if the oracle target is systematically misspecified. A natural V3 direction is end-to-end selector training, in which the selector's output is connected differentiably to the backbone's final answer likelihood through a soft top-$k$ or straight-through estimator. The main challenges are computational cost and gradient estimation through discrete selection, both of which are standard engineering rather than conceptual obstacles. Prompt-compression work in text provides a natural template for this transition \citep{nagle2024fundamental}.

\subsection{Scope, limitations, and when the theory applies}
\label{sec:discussion-scope}

The paper makes three types of claim. First, it makes exact theorem claims in settings where the relevant rate quantities are analytically tractable; these are the synthetic validations of \S\ref{sec:synthetic-validations}. Second, it makes structural empirical claims about theorem-aligned quantities on real data: the family-level excess risk gap, nested monotonicity, and factor overlap. Third, it makes staged operational claims about conditioned compression on real audio. V1 provided a suggestive single-seed DCASE signal; V2 changes the operational picture to a backbone- and family-dependent map, with cross-backbone positive AudioMCQ, negative DCASE means, and a large Qwen2.5-Omni temporal failure on MMSU.

The paper does \emph{not} claim that V1 or V2 directly measures the true rate-theoretic frontier on real audio, that the learned-native diagnostic table estimates the theoretical architecture gap, that V1's keyword partition is the final semantic story, or that query conditioning is uniformly beneficial across datasets and backbones. Those absences are deliberate and explicitly documented. They are not weaknesses hidden in the appendix; they define the honest scope of what the present evidence establishes.

The underlying theory is broader than the present instantiation. It applies to any compression interface that produces a random variable $Z$ and any downstream answerer that incurs per-query Bayes risks: hard token selection, audio re-encoding, prompt compression, or hybrid interfaces. The family-level excess-risk refinement and the monotonicity structure are therefore general phenomena, not artifacts of the specific selector architecture used here. Theorem~\ref{thm:condadv} is more regime-specific: it predicts strict separation when query families act on disjoint latent factor blocks, and it collapses smoothly toward zero as factor overlap increases. The present paper tests those statements conservatively, and V2 shows why the operational version must be reported by backbone and by semantic subfamily rather than as a single dataset-level scalar.

\end{document}